\documentclass[onecolumn, journal]{IEEEtran}

\usepackage[margin=1 in]{geometry}
\usepackage[utf8]{inputenc} 
\usepackage[T1]{fontenc}
\usepackage{url}              % provides \url{...}
\usepackage{cite}             % improves presentation of citations
\usepackage{bm}
\usepackage[cmex10]{amsmath}  % Use the [cmex10] option to ensure complicance
\usepackage{mleftright}       % fix to wrong spacing of \left-,
\mleftright                   % \middle- \right-commands 

\usepackage{graphicx}         % provides \includegraphics{...} to
\usepackage{booktabs}         % fixes poor spacing in tables and
\usepackage[utf8]{inputenc}
\usepackage{amsmath, amssymb, amsthm, nicefrac}
\usepackage{graphicx,subcaption}
\usepackage{tcolorbox}
\usepackage{enumerate} 
\usepackage{comment}

\title{$t \geq 1$}

\newcommand{\vc}[1]{{\color{black}{#1}}}

\newcommand{\remove}[1]{}

\newtheorem{theorem}{Theorem}[section]
\newtheorem{corollary}{Corollary}[theorem]

\newtheorem{lemma}[theorem]{Lemma}

\theoremstyle{remark}
\newtheorem{remark}{Remark}

\theoremstyle{definition}
\newtheorem{definition}{Definition}[section]
\newtheorem{assumption}{Assumption}[section]

\newcommand{\LMSE}{\texttt{LMSE}}

\let\svthefootnote\thefootnote
\newcommand\blankfootnote[1]{%
  \let\thefootnote\relax\footnotetext{#1}%
  \let\thefootnote\svthefootnote%
}

\begin{document}

\title{Differentially Private Secure Multiplication: Hiding Information in the Rubble of Noise} 

%%%%%%

\author{Viveck R. Cadambe${^\triangle}$\footnote{$^\triangle$Georgia Institute of Technology, $^*$Pennsylvania State University, $^\Box$University of California, Santa Barbara, $^\diamondsuit$Harvard University}, Ateet Devulapalli$^{*}$, Haewon Jeong$^\Box$, Flavio P. Calmon$^\diamondsuit$}

\date{}
\maketitle

\begin{abstract}
We consider the problem of private distributed multi-party multiplication. It is well-established that Shamir secret-sharing coding strategies can enable perfect information-theoretic privacy in distributed computation via the celebrated algorithm of Ben Or, Goldwasser and Wigderson (the ``BGW algorithm''). However, perfect privacy and accuracy require an honest majority,  that is, $N \geq 2t+1$ compute nodes are required to ensure privacy against any $t$ colluding adversarial nodes. We develop a novel coding formulation for secure multiplication  that allows for (i) computation over the reals, rather than over a finite field, (ii) for a controlled information leakage measured in terms of \emph{differential privacy} parameters, and (iii) approximate multiplication instead of exact multiplication with accuracy measured in terms of the mean squared error metric. Our formulation allows the development of coding schemes for the setting where the number of honest nodes can be a minority, that is $N< 2t+1.$  Our main result develops a tight characterization of the privacy-accuracy trade-off via new coding schemes and a converse argument. A particularly novel technical aspect is an intricately layered noise distribution that merges ideas from differential privacy and Shamir secret-sharing at different layers.
\blankfootnote{This paper is an extended version of our papers published in the 2022 and 2023 Proceedings of International Symposium on Information Theory.}
\end{abstract}

\section{Introduction}

Ensuring privacy in distributed data processing is a central engineering challenge in modern machine learning. Two common privacy definitions in  data processing are information-theoretic (perfect) privacy  and differential privacy \cite{dwork2006calibrating,CynthiaDworkBook2014}. Perfect information-theoretic privacy is the most stringent definition, requiring that no private information is revealed to colluding adversary nodes regardless of their computational resources.  Differential privacy, in turn, allows a tunable level of privacy and ensures that an adversary cannot distinguish inputs that differ by a small perturbation (i.e., ``neighboring'' inputs).

Coding strategies have a decades-long history of enabling perfect information-theoretic privacy in distributed computing. The most celebrated is the BGW algorithm \cite{BGW1988,evans2017pragmatic}, which ensures information-theoretically private distributed computations for a wide class of functions. The BGW algorithm adapts Shamir secret-sharing \cite{shamir1979share} --- a technique that uses Reed-Solomon codes for distributed data storage with privacy constraints --- to multiparty function computation. \vc{ Consider the task of privately computing the product of two random variables ${A},{B} \in \mathbb{F},$ where $\mathbb{F}$ is a field, using  $N$ computation nodes.} Let ${R}_{i},{S}_{i} \in \mathbb{F},i=1,2,\ldots,t$ be statistically independent random variables. In Shamir's secret sharing, node $i$ receives inputs $\tilde{{A}}_{i}= p_{1}(x_i), \tilde{{B}}_{i}=p_2(x_i),$ where, $x_1,x_2,\ldots,x_N \in \mathbb{F}$ are distinct non-zero scalars and $p_1(x),p_2(x)$ are polynomials:
\begin{equation} p_1(x) = {A} + \sum_{j=1}^{t} {R}_{j} x^{j},
p_2(x) = {B} + \sum_{i=1}^{t} {S}_{j}x^{j}.\label{eq:ss} \end{equation}

If the field $\mathbb{F}$ is finite and ${R}_{i},{S}_{i},i=1,2,\ldots,t$ are uniformly distributed over the field elements, then the input to any subset $\mathcal{S}$ of $t$ nodes is independent of the data $({A},{B}).$ The Shamir secret-sharing coding scheme allows the BGW algorithm to recover any linear combination of the inputs from any subset of $t+1$ nodes. To obtain $\alpha A+ \beta B$ for fixed constants $\alpha,\beta \in \mathbb{F},$ node $i$, \vc{ on receiving $\tilde{A}_{i},\tilde{B}_i$}, computes $\alpha \tilde{{A}}_{i}+\beta \tilde{{B}}_{i}$ for $\alpha,\beta \in \mathbb{F};$ \vc{notably, the output of the computation of node $i$ is an evaluation of the polynomial $p_{\textrm{sum}}(x) \stackrel{\Delta}{=}\alpha p_1(x)+\beta p_2(x)$ at $x=x_i$.}  The sum $\alpha {A}+\beta {B}$ --- which is the constant in the \vc{ degree $t$} polynomial $p_{\textrm{sum}}(x)$ --- can be recovered from the computation output of any $t+1$ of the $N$ nodes by interpolation of the \vc{polynomial $p_{\textrm{sum}}(x)$}. Observe, similarly, that $\tilde{{A}}_{i}\tilde{{B}}_{i}$ can be interpreted as an evaluation at $x=x_i$ of the degree $2t$ polynomial \vc{ $p_{\textrm{product}}(x) = p_{1}(x)p_{2}(x),$} whose constant term is ${A}{B}$. Thus, the product ${A}{B}$ can be recovered from any $2t+1$ nodes by interpolating the polynomial $p_{\textrm{product}}(x)$ (see Fig. \ref{fig:ourproblem} (a)). 

The BGW algorithm uses Shamir secret-sharing to perform secure MPC for the universal class of computations that can be expressed as  sums and products while maintaining (perfect) data privacy. However, notice that perfect privacy comes at an infrastructural overhead for non-linear computations. When computing the product, the BGW algorithm requires an ``honest majority'', that is, it requires $N \geq 2t+1$ computing nodes to ensure privacy against any $t$ colluding nodes. In contrast, only $t+1$ nodes are required for linear computations. This overhead becomes prohibitive for more complex functions, leading to multiple communication rounds or additional redundant computing nodes (see, for example, \cite{BGW1988, yu2019lagrange}).

\begin{figure}[!tb]
\begin{subfigure}{.48\textwidth}
  \centering
    \centering
     \includegraphics[width=2.4in]{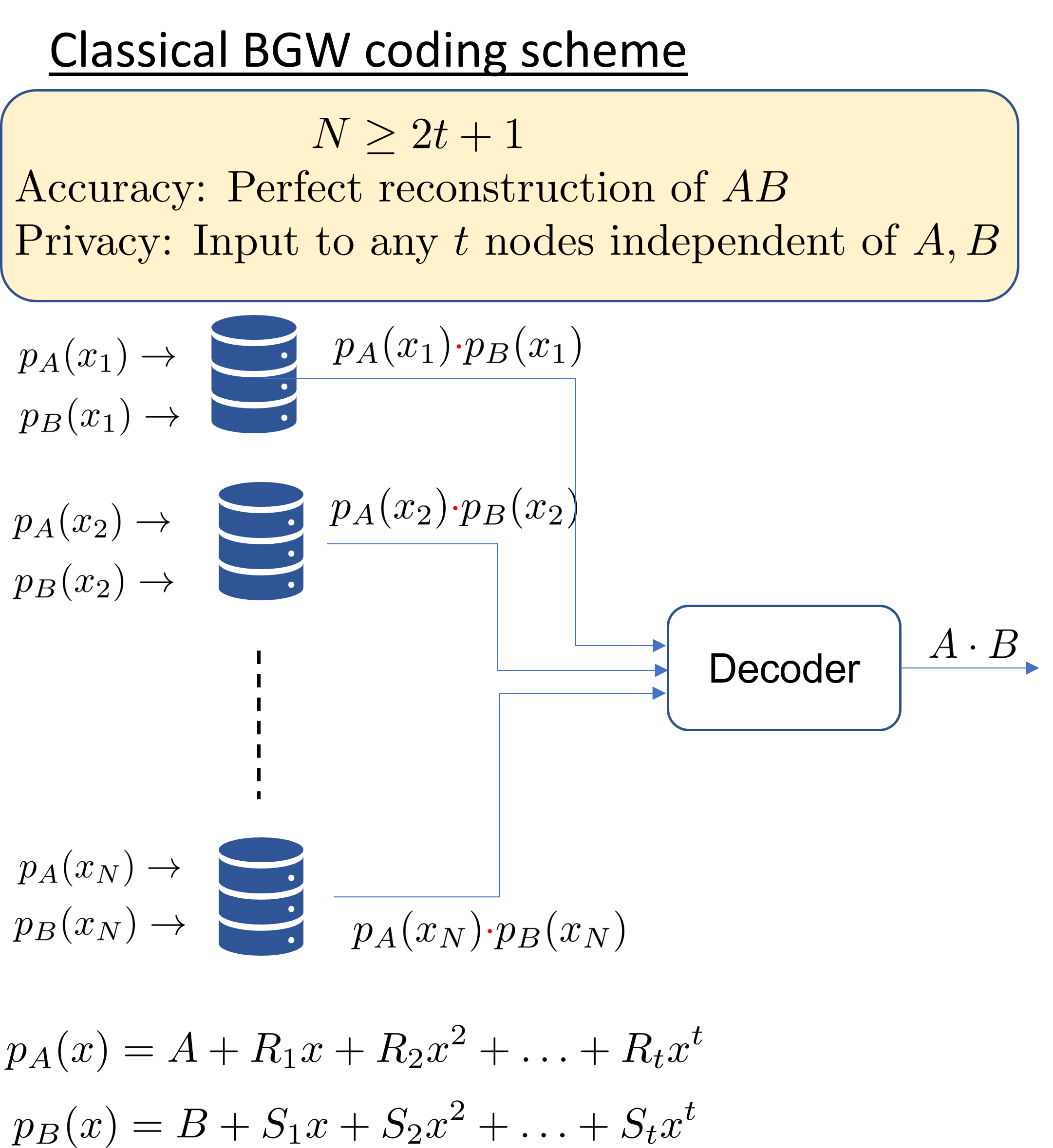}\hspace{10pt}\vline
     \vspace{25pt}
    \caption{\small{The Shamir secret-sharing coding scheme used in the BGW algorithm.}}
    \label{fig:ourproblem1}
    \end{subfigure}\hspace{5pt}
\begin{subfigure}{.48\textwidth}
  \centering
\hspace{10pt}     \includegraphics[width=2.4in]{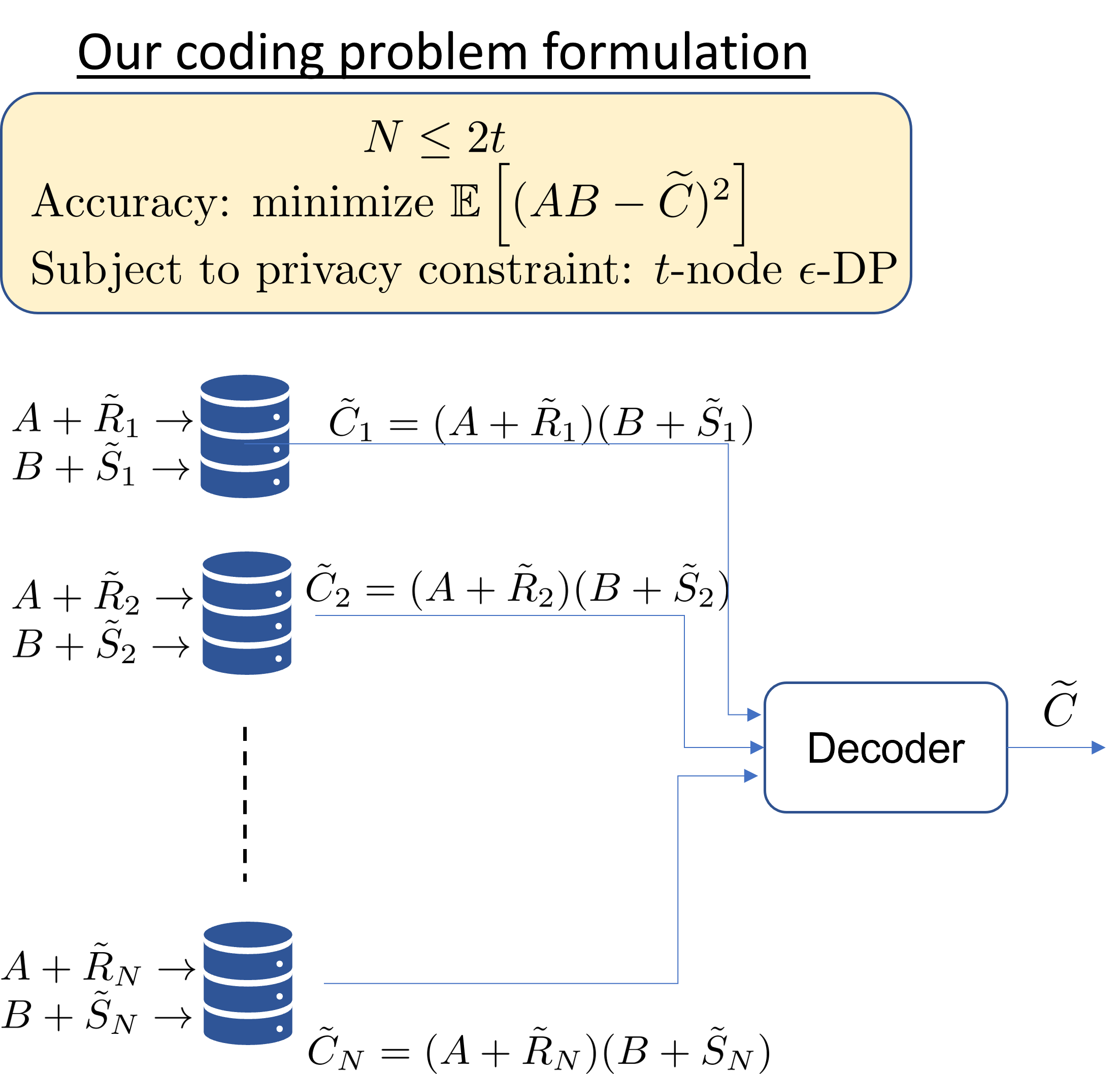}
     \vspace{35pt}
    \caption{\small{Privacy-Accuracy Trade-off for $N<2t+1$}. The privacy requirement is that the input to every subset of $t$ nodes must satisfy $\epsilon$ differential privacy.}
    \label{fig:ourproblem2}
    \end{subfigure}
   
   \caption{Pictorial depiction of our problem formulation and comparison with the coding scheme used in the BGW algorithm.}
   \label{fig:ourproblem}
\end{figure}

In this paper, we study the problem of secure multiplication for real-valued data and explore coding schemes that enable a set of fewer than $2t+1$ nodes to compute the product while keeping the data private from any $t$ nodes. Although exact recovery of the product and perfect privacy are simultaneously impossible, we propose a novel coding formulation that allows for approximations on both fronts, enabling an accuracy-privacy trade-off analysis. Our formulation utilizes differential privacy (DP) --- the standard privacy metric to quantify information leakage when perfect privacy cannot be guaranteed\cite{dwork2019differential}. It is worth noting that approximate computation suffices for several applications, particularly in machine learning,  where both training algorithms and inference outcomes are often stochastic. Also, notably, DP is a prevalent %\footnote{See, for example,  Google's Tensorflow Privacy Framework.} 
paradigm for data privacy in machine learning applications in practice (e.g.,  \cite{mcmahan2018general,TensorflowPrivacy}). 

For \emph{single-user computation}, where a user queries a database in order to compute a desired function over sensitive data, differential privacy can be ensured by adding noise to the computation output \cite{dwork2006calibrating}. The study of optimal noise distributions for privacy-utility trade-offs in the release of databases for computation of specific classes of functions is an active area of research in differential privacy literature \cite{dwork2019differential,agarwal2017price,duchi2013local,asoodeh2021local}. Our contribution is the discovery of optimal noise structures for multiplication in the \emph{multi-user setting}, where the differential privacy constraints are on a set of $t$ colluding adversaries.

\section{Summary of Main Results and Key Ideas}
\label{sec:summary}
We consider a computation system with $N$ nodes where each node receives noisy versions of inputs $A,B$ and computes their product (See Fig. \ref{fig:ourproblem2}). Specifically, for $i = 1,2,\ldots, N$, node $i$ receives $$ \tilde{A}_{i} = A+ \tilde{R}_{i}$$ 
$$ \tilde{B}_{i} = B+ \tilde{S}_{i}$$ and computes $\widetilde{C}_i = \tilde{A}_i \tilde{B}_{i}$. The goal of the decoder is to recover an estimate ${\widetilde{C}}$ of the product ${A}{B}$ from $N$ computation outputs at a certain accuracy level, measured in terms of the mean squared error $\mathbb{E}[(\widetilde{C}-AB)^2]$. The joint distributions of the noise variables $(\tilde{R}_1, \ldots, \tilde{R}_N), (\tilde{S}_1, \ldots, \tilde{S}_N)$ and should ensure that the input to any subset of $t$ nodes in the system satisfies $\epsilon$-differential privacy (abbreviated henceforth as $t$-node $\epsilon$-DP). Given $N,t$ and the DP parameter $\epsilon,$ our main result provides a tight expression for the minimum possible mean squared error at the decoder for any $N \geq t$. Of particular importance is our result for the regime $t < N < 2t+1.$

While our results provide a characterization of the accuracy in terms of differential privacy, the trade-off lends itself to an elegant, compact description when presented in terms of a \emph{signal-to-noise ratio (SNR)} metric 
for both privacy and accuracy. Privacy SNR ($\texttt{SNR}_p$) describes how well $t$ colluding nodes can extract the private inputs $A$, $B$, i.e., higher privacy SNR means poor privacy (See Sec. \ref{sec:SNR} for a precise definition). \vc{Privacy SNR guarantees can be mapped to $\epsilon$-differential privacy guarantees --- see \eqref{eq:optimalsigma} and the surrounding discussion. } Through a non-trivial converse argument, we show that for any $N < 2t+1$, the mean squared error satisfies:
\begin{equation}\label{eq:converse1}
    \mathbb{E}[(AB-\tilde{C})^2] \geq \frac{1}{(1+\texttt{SNR}_p)^2}.
\end{equation}
We provide an achievable scheme that meets the converse bound arbitrarily closely for $N \geq t+1$. \vc{Specifically, for fixed $N,t, N \geq t+1,$ and for a fixed privacy SNR $\texttt{SNR}_{p},$ we provide a sequence of achievable coding schemes that satisfy $\mathbb{E}\left[(AB-\tilde{C})^2\right]\rightarrow \frac{1}{(1+\text{SNR}_p)^2}$.} Surprisingly, (\ref{eq:converse1}) does not depend on $t,N$ --- the trade-off remains\footnote{It is instructive to note that a coding scheme that achieves a particular privacy-accuracy trade-off for $t$ colluding adversaries over $N$ nodes, can also be used to achieve the same trade-off for a system with $N' > N$. To see this, simply use the coding scheme for the first $N$ nodes and ignore the output of the remaining $N'-N$ nodes.} the same for $N \in \{t+1,t+2,\ldots 2t\}.$  Thus, our main result implies that for the regime of $t < N < 2t+1$, remarkably, having more computation nodes does not lead to increased accuracy. 

\vc{
The main technical contribution of our paper is the development of an intricate noise distribution that achieves the optimal trade-off. To get a high level understanding of the achievable scheme, consider a system with $N=3$ computation nodes where any $t=2$ nodes can collude. Assume, for expository purposes, that $A,B$ are unit variance random variables. We compare our main result with two baseline schemes respectively based on ideas of secret sharing and differential privacy literature. We describe the baselines at a high level below, with Appendix \ref{app:SS_staircase_lower}  supplying missing details.

\subsubsection*{Baseline 1: Complex-valued Shamir secret Sharing \cite{soleymani2021analog,liu2023analog}} 
As a first baseline scheme, consider Shamir secret sharing where node $i$ receives:
\begin{eqnarray}\tilde{A}_i &=& A + x_iR_1 + x_i^2 R_2 \label{eq:ss1} \\
\tilde{B}_i &=& B + x_iS_1 + x_i^2 S_2, \label{eq:ss2}\end{eqnarray}
that is $\tilde{R}_i = x_iR_1 + x_i^2 R_2, \tilde{S}_i = x_iS_1 + x_i^2 S_2,$ where $R_1, R_2,S_1,S_2$ are independent random variables. In standard Shamir secret sharing, $R_1, R_2, S_1,S_2$ are each chosen uniformly over a finite field. However, we are interested in real-valued realizations of $A,B$, and the mean squared error metric for reconstruction.
References \cite{soleymani2021analog, liu2023differentially, liu2023analog} propose treating \eqref{eq:ss1},\eqref{eq:ss2} as equations over real/complex field with $R_1,R_2,S_1,S_2$ chosen in an i.i.d. manner from some real or complex valued distributions. These references set $x_i = e^{-j (i-1)\pi/3}$ where $j=\sqrt{-1},$ and perform a privacy-accuracy trade-off analysis\footnote{References \cite{liu2023differentially, liu2023analog} perform a truncation step on the added noise in \eqref{eq:ss1},\eqref{eq:ss2} that we are ignoring here in our discussion and baseline.}. Our first baseline is inspired by these references.  Similar to them, we set $x_i = e^{-j (i-1)\pi/3}$, and choose $R_1,R_2,S_1,S_2$ independently from the same complex-valued distribution $\mathbb{P}_{R_1}.$ In Appendix \ref{app:SS_staircase_lower}, we develop a lower bound on the $2$-node DP parameter $\epsilon$ over all possible choices of $\mathbb{P}_{R_1}$. This lower bound, which is an optimistic estimate of the privacy of the schemes of our baseline, is used in the plot of Fig. \ref{fig:dpcomp}. We provide more details in comparison to the schemes of \cite{liu2023differentially,liu2023analog} in Sec. \ref{sec:related}

\subsubsection*{Baseline 2: Independent Additive Noise} 
As a second baseline scheme, assume that $\tilde{R}_i,\tilde{S}_i,i=1,2,3$ are each independent identically random variables satisfying $\tilde{\epsilon}$-DP. Because the privacy mechanisms to any any pair of nodes $(\tilde{A}_i,\tilde{A}_j)$ are independent, the input to every pair of nodes satisfies $\epsilon$-DP where $\epsilon=2\tilde{\epsilon}$. We show in Appendix \ref{app:SS_staircase_lower} that, for a fixed mean squared error, the privacy is maximized (that is $\epsilon$ is minimized) for this baseline if $\tilde{R}_i,\tilde{S}_i,i=1,2,3$ have the staircase distribution of \cite{Pramod2015}. Therefore, the privacy-accuracy plot for independent mechanisms in Fig. \ref{fig:dpcomp} is optimal for the class of mechanisms that select  $\tilde{R}_i,\tilde{S}_i,i=1,2,3$ in an i.i.d. manner.

\begin{figure}  \centering
\includegraphics[width=0.5 \textwidth]{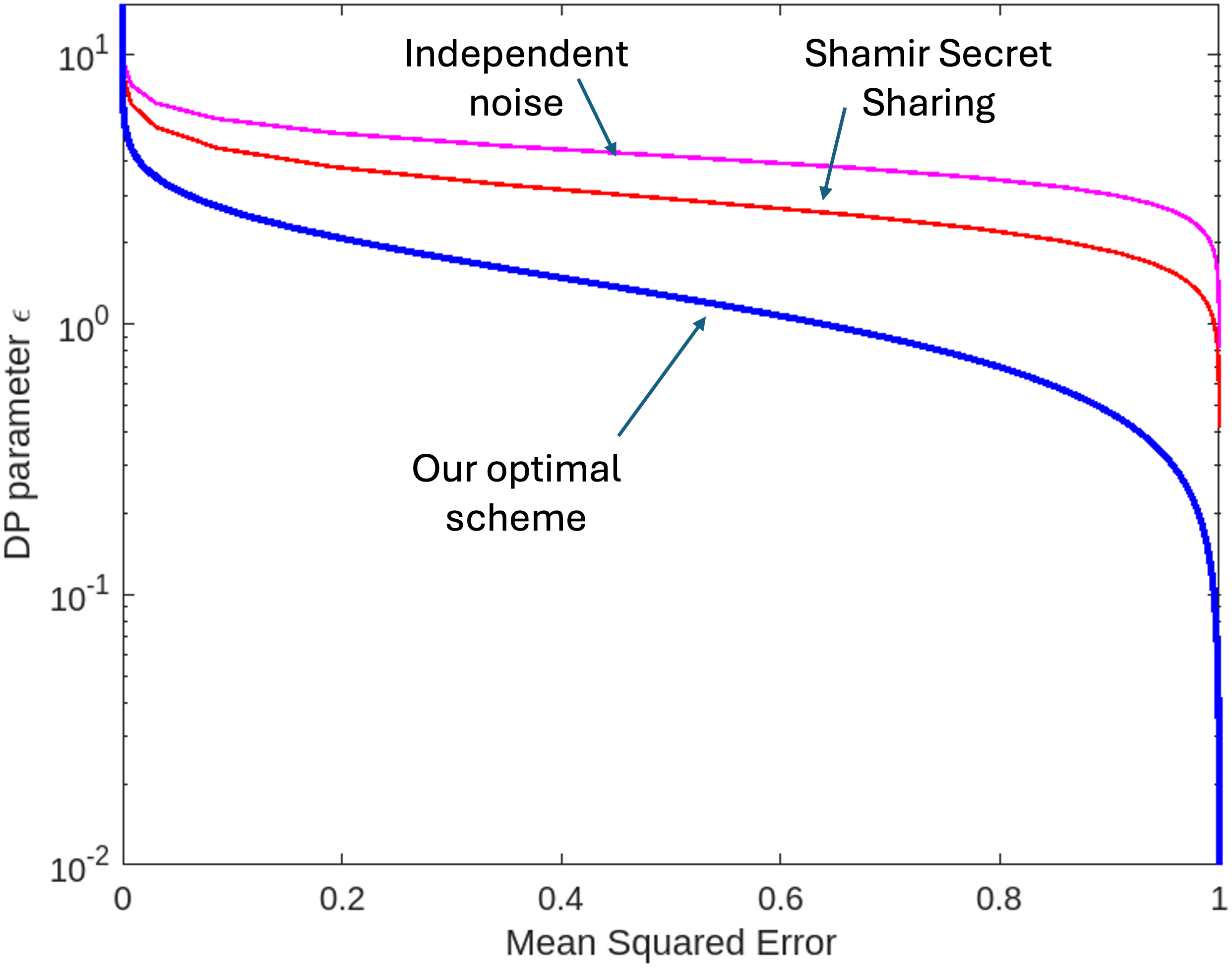}
\caption{\footnotesize{Performance of our optimal scheme of for $N=3,t=2$ in comparison to baselines (i) complex-valued Shamir Secret Sharing and (ii) independent noise across nodes.}}
\label{fig:dpcomp}
\end{figure}

\subsubsection*{High-level description of our optimal scheme}
In Fig. \ref{fig:dpcomp}, we have plotted a comparison of the performance of an (asymptotically) optimal coding scheme in comparison to the two baseline schemes. The optimal coding scheme --- the main technical contribution of our paper --- combines secret sharing ideas and noise mechanisms from DP literature through a carefully designed layered structure. In more detail, node $i$ receives $\tilde{A}_i,\tilde{B}_i,i=1,2,3$, where:

\begin{eqnarray}\tilde{A}_1 &=& A + \underbrace{R_1}_{\textrm{First Layer}} + \underbrace{\delta R_2}_{\textrm{Second Layer}} + \underbrace{\delta^{1.5}R_1}_{{\textrm{Third Layer}}} \label{eq:coding1}\\
\tilde{B}_1 &=& B + S_1 + \delta S_2 + \delta^{1.5}S_1 \\
\tilde{A}_2 &=& A + R_1 - \delta R_2 + \delta^{1.5}R_1 \\
\tilde{B}_2 &=& B + S_1 - \delta S_2 + \delta^{1.5}S_1 \\
\tilde{A}_3 &=& A + R_1 \\
\tilde{B}_3 &=& B + S_1, \label{eq:coding6}
    \end{eqnarray}
where for a fixed target DP parameter $\epsilon$, and $R_1, S_1, R_2, S_2$ are independent random variables whose distribution is described next. Random variable $R_1$ has the distribution that satisfies:
\begin{enumerate}[(i)]
\item The privacy mechanism $A \rightarrow A+R_1$ satisfies $\epsilon$-DP
\item Among all random variables that satisfy (i), $R_1$ has the smallest possible variance.
\end{enumerate}
\begin{figure*}[!t]
    \centering
   \begin{subfigure}{0.9\textwidth}  \centering\includegraphics[width=0.75\textwidth]{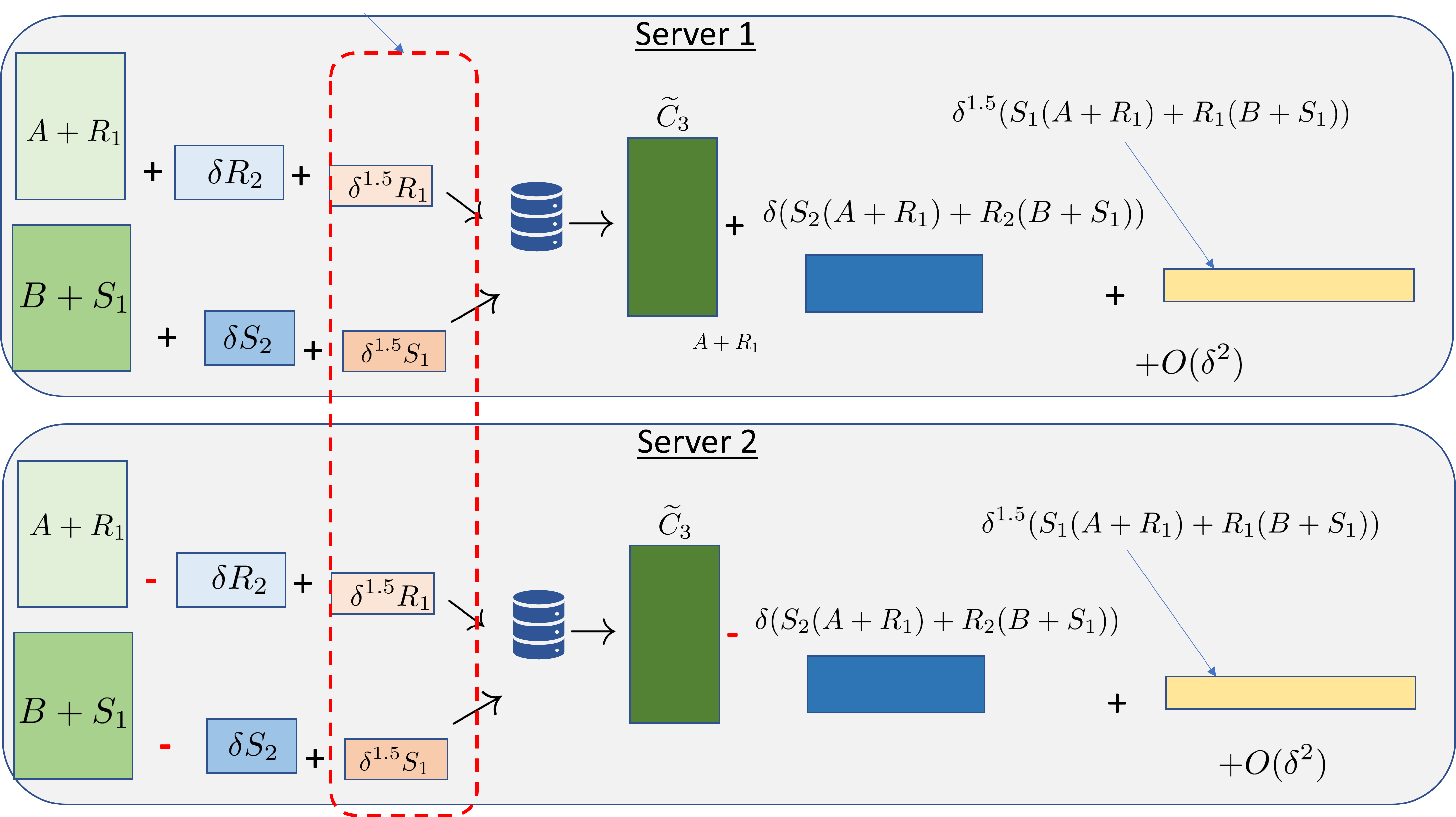} 
       \end{subfigure}
\begin{subfigure}{0.9\textwidth} \centering
    \includegraphics[width=0.75\textwidth]{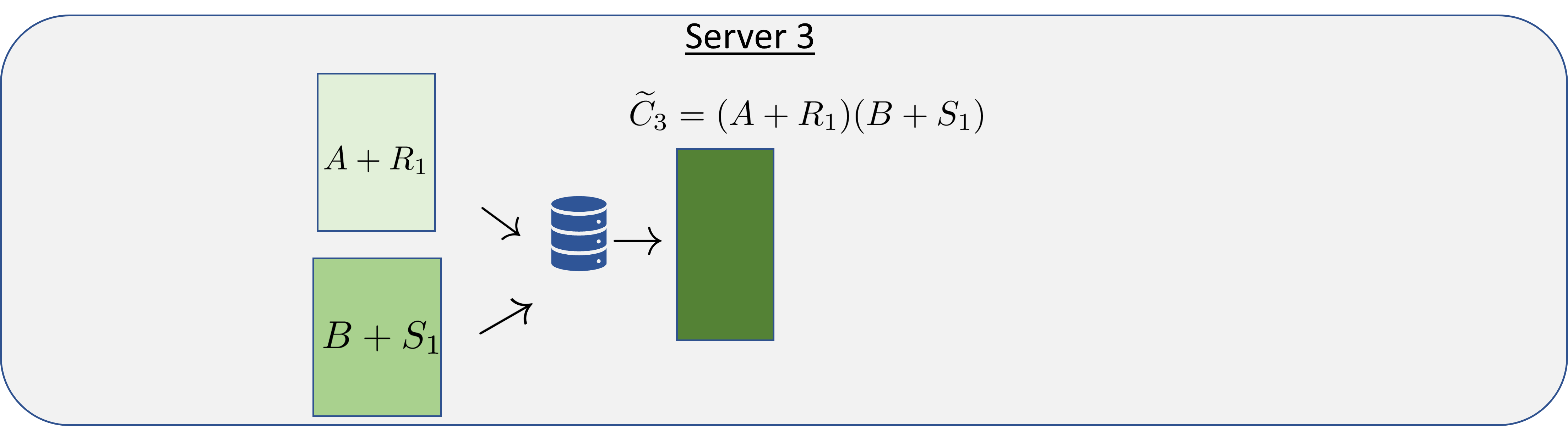}
    
\end{subfigure}
   \caption{\small{Pictorial depiction of our layered coding scheme for $t=2, N=3$ with matrix, see \eqref{eq:coding1}-\eqref{eq:coding6}}}
    \label{fig:code_construction}
\end{figure*}

In fact, the distribution of $R_1$ has been explicitly characterized and termed the \emph{staircase mechanism} in \cite{Pramod2015}. The distribution of the independent random variable $S_1$ is chosen identically to $R_1.$ The random variables $R_2,S_2$ are chosen to be Laplacian random variables with unit variance. Notice that the additive noise at the first two nodes can be interpreted as a superposition of three layers. Based on the distributions above, the first and third layers use an ptimal DP noise mechanism. The second layer is similar to secret sharing. In particular, ignoring the effect of the first and third layers, note that the linear combinations at the three nodes $A+\delta R_2, A - \delta R_2, A$ can be interpreted as pairwise linearly independent linear combinations of $(A,R_2)$ similar to Shamir Secret Sharing. In fact, this property translates to the general scheme presented in Sec. \ref{sec:achievability} as well.  

Our proof shows that as $\delta \to 0,$ the above scheme asymptotically achieves the optimal privacy-accuracy trade-off dictated by \eqref{eq:converse1}. Consider an adversary that has access to any two nodes; as an example, suppose that the adversary has access to $\tilde{A}_1,\tilde{A}_2$ and we aim to characterize the privacy loss. As $\delta \to 0,$ the magnitude of the third layer diminishes with respect to the first two layers, and consequently, $\tilde{A}_1,\tilde{A}_2$ respectiveky become statistically very close to $A+R_1 + \delta R_2, A+R_1 - \delta R_2.$ That is, $\tilde{A}_1,\tilde{A}_2$ become statistically close to degraded versions of $A+R_1.$ Because $R_1$ is designed to ensure $\epsilon$-DP, the scheme (nearly) achieves $2$-node  $\epsilon$-DP. Measured in terms of the privacy SNR metric, our proof shows that this scheme has $\texttt{SNR}_{p} \approx \frac{1}{\mathbb{E}[R_1^2]}.$

The legitimate decoder that accesses all three nodes obtains: 
\begin{eqnarray*}
\tilde{C}_1 &=& \tilde{A}_{1}\tilde{B}_1 = \tilde{C}_3 - \delta Y + \delta^{1.5} Z + O(\delta)^{2} \\
\tilde{C}_2 &=& \tilde{A}_{2}\tilde{B}_2 = \tilde{C}_3 + \delta Y + \delta^{1.5} Z + O(\delta)^{2} \\
\tilde{C}_3 &=& (A+R_1)(B+S_1)
\end{eqnarray*}
where $Y=S_2(A+R_1)+R_2(B+S_1)$ and $Z=S_1(A+R_1)+R_1(B+S_1)$. In effect, because the legitimate decoder has access to all three inputs $\tilde{C}_1,\tilde{C}_2, \tilde{C_3},$ three layers of the output are effectively available for the decoder (unlike the adversary, where the third layer is effectively hidden). Specifically, by computing $\frac{\widetilde{C}_1-\widetilde{C}_2}{2}$ and $\frac{\widetilde{C}_1+\widetilde{C}_2}{2} - \tilde{C}_3,$ the decoder can access $\tilde{C}_3,Y,Z.$ An analysis of the mean squared error of a linear MMSE estimator of $AB$ from $(\tilde{C}_3,Y,Z.)$ results in the following privacy-accuracy trade-off: 
$$ \mathbb{E}\left[\left(\widetilde{C} - AB\right)^2\right] = \frac{1}{\left(1+\frac{1}{E[R_1^2]}\right)^2} + O(\delta) \approx \frac{1}{(1+\texttt{SNR}_p)^2}$$ 

A converse argument shows that the privacy-accuracy trade-off cannot be improved even if $N=4,$ so long as $t=2.$

\begin{remark}
It is important to note that our scheme has a fundamentally different structure compared to Shamir Secret Sharing \eqref{eq:ss1},\eqref{eq:ss2}. Although our comparison applies a specific choice of evaluation points $x_1,x_2,x_3$ as prescribed by \cite{soleymani2021analog,liu2023analog}, it is an open question whether a different choice of $x_1,x_2,x_3$ along with carefully tailored distributions for $R_i,S_i|_{i=1}^{2}$ can achieve the trade-off of \eqref{eq:converse1}.  For the case of $N=2,t=2$, a positive answer to this question has been provided in \cite{devulapalli2022differentially}, however, the question is open for $N > 2.$. 
\end{remark}

\begin{remark}
Because $\delta \to 0$, in a potential practical application of our coding scheme, we require precise computation to reap the benefits. Fig. \ref{fig:dpcomp} thus simply shows the limiting performance of our scheme as $\delta \to 0.$ In settings with finite precision, we can expect some degradation. Sec. \ref{sec:precision}  discusses this aspect quantitatively. 
\end{remark}

}

\subsection{Related Work}
\label{sec:related}
\textbf{Differential Privacy and Secure MPC:}
Several prior works are motivated like us to reduce computation and communication overheads of secure MPC by connecting it with the less stringent privacy guarantee offered by DP. 
References \cite{stevens2022efficient, Yuan2021,Pettai2015,HeXi2017,JonasBohler2020,Chowdhury_crypte,talwar2022differential,PeterKairouz2015} provide methods to reduce communication overheads and improve robustness while guaranteeing differential privacy for sample aggregation algorithms, label private training, record linkage, distributed median computation \vc{and general boolean function computation}.   In comparison, we aim to develop and study coding schemes with optimal privacy-accuracy trade-offs for differentially private distributed real-valued multiplication that reduce the overhead of $t$ redundant nodes. 

\noindent \textbf{Coded Computing:} The emerging area of coded computing enables the study of codes for secure computing that enable data privacy.
Our framework resonates with the coded computing approach, as we abstract the algorithmic/protocol related aspects into a master node, and highlight the role of the error correcting code in our model. Coded computing has been applied to study code design for secure multiparty computing in \cite{pmlr-v89-yu19b,MahdiSoleymani2021,Akbari-Nodehi2021,ChangWei-Ting2018,Rouayheb_GASP,jia2021capacity,chen2021gcsa}. These references effectively extend the standard BGW setup by imposing memory constraints on the nodes \vc{(e.g., Lagrange coded computing~\cite{yu2018coding})}, or other constraints, that effectively disable each node storing information equivalent to the entire data sets. Under the imposed constraints, these references develop novel codes for exact computation and perfect privacy.  In particular, codes for secure MPC over real-valued fields have been studied in \cite{fahim2019numerically, MahdiSoleymani2021} extending the ideas of \cite{pmlr-v89-yu19b} to understand the loss of accuracy due to finite precision. In particular, reference \cite{MahdiSoleymani2021}  casts the effect of finite precision in a privacy-accuracy tradeoff framework. In contrast to all previous works in coded computing geared towards secure MPC, we operate below the threshold of perfect recovery and characterize privacy-accuracy trade-offs. Our incorporation of differential privacy for this characterization is a novel aspect of our set up.  We do not impose any memory constraints on the nodes, and imposition of such constraints can lead to interesting areas of future study.

\vc{Most similar to our work are references \cite{liu2023differentially, liu2023analog}, which are concurrent to our work. These works also consider privacy-accuracy trade-offs for both coded computing and multiparty computing via the differential privacy framework.  The coding schemes considered there are a modification of complex-valued Shamir Secret sharing considered in Fig. \ref{fig:dpcomp}; in fact, our baseline is inspired by these references. However, notably, these works provid $(\epsilon,\delta)$ differential privacy for $\delta > 0,$ unlike our work which considers $\epsilon$-DP (i.e., $\delta = 0$). An analysis of optimal schemes for multiplication for that ensure $(\epsilon,\delta)$-DP and comparison with \cite{liu2023analog,liu2023differentially} is outside the scope of this work.}

\noindent \textbf{Privacy-Utility Trade-offs.} There is a fundamental trade-off between DP and utility (see \cite{agarwal2017price,duchi2013local,asoodeh2021local} for examples in machine learning and statistics). The optimal $\epsilon$-DP noise-adding mechanism for a target moment constraint on the additive noise was characterized in \cite{Pramod2015}. For approximate DP, near-optimal additive noise mechanisms under $\ell_1$-norm and variance constraints were recently given in \cite{GengQuan2020}. 
\vc{An independent body of work in coding theory studies secret sharing schemes over finite fields, and quantifies the amount of information leakage (e.g., measured via the entropy) for general linear codes when the number of honest nodes is fewer than the threshold for perfect secrecy and perfect accuracy; see \cite{kurihara2012secret} and several follow up works. Notably, generalized Hamming weights of linear codes \cite{wei1991generalized} emerge as an important metric for characterizing the privacy of linear codes. Our work differs from this body in that we focus on the application of secret sharing techniques to secure computing (specifically secure multiplication).}

% \textcolor{red}{===============}

% Recently, 

% {\color{blue} \begin{itemize}
% \item Single user computing well studied. DP and noise distributions. 
% \item Multi user, finite fields, and exact privacy well studied. 
% \item Key lesson from multi-user privacy: linearly independent noise, expansion of number of nodes (more redundant infrastructure), lots of communication needed for complicated functions.
% \end{itemize}

% What lessons from current multi-user privacy literature are relevant for approximate computing? We show that new phenomena occur in multi-user privacy in approximate computing. Specifically, beyond modulation along linearly independent vectors, differences in magnitudes can be interpreted as separate dimensions.

% Todo list.
% \begin{itemize}
%     \item Remember to explain that we are doing scalar multiplication, and justify it.
%     \item 
% \end{itemize}
% }

% \textcolor{red}{===============\\
% I don't really know what you're trying to say up there, so I'll leave it as it is.}

%In this work, we extend the privacy-accuracy trade-off analysis in~\cite{ISIT2022} for a general $t>1$. \textcolor{red}{Summarize our results. Say why it's challenging. }

% \subsection{Main Results}
% \subsection{Related Work}

\section{System Model and Statement of Main Results}
%\textbf{Notations:} 
%We  define $[n]\triangleq \{1, 2, \cdots, n \}$. We use bold fonts for vectors and matrices. We define $(\mathbf{x})_i$ to be the $i$\textsuperscript{th} component of a vector $\mathbf{x}$ and $(\mathbf{X})_{k,l}$ be the $(k,l)$\textsuperscript{th} element of a matrix $\mathbf{X}$. Denote $\smin{\mathbf{X}},||\mathbf{X}||_2$ and $||\mathbf{X}||_F$ to be the minimum singular value, $\ell_2$ norm and Frobenius norm of a matrix $\mathbf{X}$ respectively. We use $X\sim \mathbb{Q}$ to say that the random variable $X$ has the probability distribution $\mathbb{Q}$.

We present our system model for distributed differentially private multiplication and state our main results. Our model and results are presented for the case of scalar multiplication here. Natural extensions of the results for the case of matrix multiplication is presented in Sec. \ref{sec:matrix}.

\subsection{System Model}
\label{sec:model}
%Although our 
We consider a computation system with $N$ computation nodes. ${A},{B}\in\mathbb{R}$ are random variables, and node $i \in \{1,2,\ldots,N\}$ receives:
\begin{equation}
    \tilde{{A}}_{i} = a_i {A}+\tilde{R}_{i},~~~ \tilde{{B}}_{i} = b_i {B}+\tilde{S}_{i}, 
    \label{eq:model}
\end{equation}
where $\tilde{R}_{i},\tilde{S}_{i} \in \mathbb{R}$ are random variables such that $(\tilde{R}_{1},\tilde{R}_{2},\ldots, \tilde{R}_{N},\tilde{S}_{1},\tilde{S}_{2},\ldots, \tilde{S}_{N})$ is statistically independent of $({A},{B}),$ and $a_i,b_i \in \mathbb{R}$ are constants.  
%We denote by $\mathbf{R},\mathbf{S} \in \mathbb{R}^{L \times PL},$ the following:
%$$\mathbf{R} = \begin{bmatrix}\mathbf{R}_{1}&\mathbf{R}_{2}&\ldots& \mathbf{R}_{N}\end{bmatrix}$$ 
%$$\mathbf{S} = \begin{bmatrix}\mathbf{S}_{1}&\mathbf{S}_{2}&\ldots& \mathbf{S}_{N}\end{bmatrix}.$$
In this paper, we assume no shared randomness between $(\tilde{R}_1,\tilde{R}_{2},\ldots,\tilde{R}_N)$ and $(\tilde{S}_1,\tilde{S}_{2},\ldots,\tilde{S}_{N})$ i.e., they are statistically independent: $\mathbb{P}_{\tilde{R}_1,\tilde{R}_{2},\ldots,\tilde{R}_{N},\tilde{S}_1,\tilde{S}_2,\ldots,\tilde{S}_N}=\mathbb{P}_{\tilde{R}_1,\tilde{R}_{2},\ldots,\tilde{R}_{N}}\mathbb{P}_{\tilde{S}_1,\tilde{S}_2,\ldots,\tilde{S}_N}$. For $i\in \{1,2,\ldots,N\},$ computation node $i$ outputs:
\begin{equation}
{\tilde{C}}_{i} = \tilde{{A}}_{i} \tilde{{B}}_{i}. \label{eq:comp_output}
\end{equation} 
A decoder receives the computation output from all $N$ nodes and performs a  map: $d:\mathbb{R}^{N} \rightarrow \mathbb{R}$ that is affine over $\mathbb{R}.$ That is, the decoder produces:
\begin{equation} {\widetilde{C}} = d(\tilde{{C}}_{1},\dots,\tilde{C}_N) =  \sum_{i=1}^{N} w_{i} \tilde{{C}}_{i} + w_0, \label{eq:decodout}\end{equation} 
where the coefficients $w_{i} \in \mathbb{R}, $ specify the linear map $d$. \vc{The coefficients $w_i,i=1,2,\ldots,N$ can depend on the scalars $a_i,b_i, i=1,2,\ldots, t$ and distributions $\mathbb{P}_{\tilde{R}_1, \ldots, \tilde{R}_N}\mathbb{P}_{\tilde{S}_1,\ldots, \tilde{S}_{N}}$, but cannot depend on the actual realizations of these random variables.}

A $N$-node secure multiplication \emph{coding scheme} consists of the joint distributions of $(\tilde{R}_{1},\tilde{R}_{2},\ldots, \tilde{R}_{N})$ and $(\tilde{S}_{1},\tilde{S}_{2},\ldots, \tilde{S}_{N}),$ scalars $a_1,a_2,\ldots,a_N, b_1,b_2,\ldots,b_N$\footnote{It is instructive to note that, for our problem formulation, there is no loss of generality in assuming that $a_i,b_i \in \{0,1\}$.}  and the decoding map ${d}:\mathbb{R}^N \rightarrow \mathbb{R}.$ The performance of a secure multiplication coding scheme is measured in its differential privacy parameters and its accuracy, defined next.

\begin{definition}($t$-node $\epsilon$-DP)
\label{def:edp}
    Let $\epsilon\geq 0$. A coding scheme with random noise variables $$(\tilde{R}_{1},\tilde{R}_{2},\ldots, \tilde{R}_{N}), (\tilde{S}_{1},\tilde{S}_{2},\ldots, \tilde{S}_{N})$$ and scalars $a_i,b_i \; (i \in \{1,\ldots,N\})$  satisfies $t$-node $\epsilon$-DP if, \vc{for any ${A}_{0},{B}_{0},{A}_{1},{B}_{1} \in \mathbb{R}$ that satisfy $|A_0 - A_1| \leq 1, |B_0-B_1| \leq 1$} %$\left|\left|\begin{bmatrix}{A}_{0} \\ {B}_{0} \end{bmatrix} - \begin{bmatrix}{A}_{1} \\ {B}_{1} \end{bmatrix}\right|\right|_{\infty}\leq 1$,
\begin{eqnarray}
  \max\left(\frac{\mathbb{P}\left(\mathbf{Z}^{(0)}_{\mathcal{T}} \in \mathcal{A} \right)}{\mathbb{P}\left(\mathbf{Z}^{(1)}_{\mathcal{T}} \in \mathcal{A} \right)}, \frac{\mathbb{P}\left(\mathbf{Y}^{(0)}_{\mathcal{T}} \in \mathcal{A} \right)}{\mathbb{P}\left(\mathbf{Y}^{(1)}_{\mathcal{T}} \in \mathcal{A} \right)}\right) &\leq& e^\epsilon
  \label{eq:DP}
\end{eqnarray}
for all subsets $\mathcal{T} \subseteq \{1,2,\ldots,N\},|\mathcal{T}|=t$,
for all subsets $\mathcal{A} \subset \mathbb{R}^{1 \times t}$ in the Borel $\sigma$-field, where, for $\ell=0,1$,\begin{eqnarray}
\mathbf{Y}_{\mathcal{T}}^{(\ell)} &\triangleq&  \begin{bmatrix}a_{i_1} {A}_{\ell}+\tilde{R}_{i_1}&a_{i_2}{A}_{\ell}+\tilde{R}_{i_2}&\ldots&a_{i_{|\mathcal{T}|}}{A}_{\ell}+\tilde{R}_{i_{|\mathcal{T}|}}\end{bmatrix}\label{eq:DPdefY} \\\mathbf{Z}_{\mathcal{T}}^{(\ell)} &\triangleq&  
\begin{bmatrix}  b_{i_1}{B}_{\ell}+\tilde{S}_{i_1}&b_{i_2}{B}_{\ell}+\tilde{S}_{i_2}&\ldots&b_{i_{|\mathcal{T}|}}{B}_{\ell}+\tilde{S}_{i_{|\mathcal{T}|}} \end{bmatrix} \label{eq:DPdefZ}, 
\end{eqnarray}
where $\mathcal{T}=\{i_1,i_2,\ldots,i_{|\mathcal{T}|}\}.$
\end{definition}

\vc{ Informally, the above definition means that for an adversary that acquires information from nodes $i_1, i_2, \ldots, i_t,$ the inputs $(a_{i_1}A+\tilde{R}_{i_1},a_{i_2}A+\tilde{R}_{i_2}, \ldots, a_{i_t}A+\tilde{R}_{i_t})$ should satisfy $\epsilon$-DP if the scalar $A$ is treated as a database. A similar constraint has to apply to scalar $B$.}
%{We denote by $\mathcal{P}^{\epsilon,t}_{\mathbf{R},\mathbf{S}},$ the set of all possible joint distributions $\mathbb{P}_{\mathbf{R},\mathbf{S}}\in\mathcal{P}_{\mathbf{R},\mathbf{S}}$ that satisfy $t$-node $\epsilon$-DP.} Note that (\ref{eq:DP}) depends only on the joint distribution $\mathbb{P}_{\mathbf{R},\mathbf{S}}$ and does not depend on the distributions of $\mathbf{A},\mathbf{B},$ since the definition applies for arbitrary vectors $\mathbf{A}_{0},\mathbf{B}_0,\mathbf{A}_{1},\mathbf{B}_{1}$ --- that is, those that are not necessarily drawn from $\mathbb{P}_{\mathbf{A},\mathbf{B}}.$

% \subsubsection*{\underline{Accuracy of a secure multiplication coding scheme}}

%The main goal of this paper is to characterize the trade-off between privacy and accuracy of estimation of the product $\mathbf{A}\mathbf{B}.$ In particular, we develop schemes that guarantee a certain level of DP (i.e., a certain value of parameter $\epsilon$), irrespective of the distribution of the inputs. It is, however, necessary (and standard, see \cite{PeterKairouz2015,PeterKairouz2016,MahdiSoleymani2021,jia2021capacity}) to account for the data distribution and its parameters when evaluating the  accuracy of coding schemes. The accuracy guarantees of the coding schemes developed in this paper rely on the following key assumptions:

While privacy guarantees must make minimal assumptions on the data distribution, it is common to make assumptions on the data distribution and its parameters when quantifying utility  guarantees  (e.g., accuracy) \cite{PeterKairouz2015,PeterKairouz2016,MahdiSoleymani2021,jia2021capacity}. We next state the conditions under which our accuracy guarantees hold.
\begin{assumption}
%\begin{enumerate}
${A}$ and ${B}$ are statistically independent random variables that satisfy
$$ \mathbb{E}\left[{A}^2 \right] \leq \eta, $$ $$\mathbb{E}\left[{B}^{2}\right] {\leq \eta}$$
for a parameter $\eta > 0$.
%\end{enumerate}
\label{assumption1}
\end{assumption} 
It should be noted that the above assumption implies that $\mathbb{E}[A^2B^2] \leq \eta^2.$ We measure the accuracy of a coding scheme by the mean square error of the decoded output with respect to the product $AB$. Specifically, we define:
\begin{definition}[Linear Mean Square Error ($\LMSE$)]
\label{eq:lmsedef}
For a secure multiplication coding scheme $\mathcal{C}$ consisting of joint distribution $\mathbb{P}_{\tilde{R}_1,\tilde{R}_{2},\ldots,\tilde{R}_{N},\tilde{S}_1,\tilde{S}_{2},\ldots,\tilde{S}_{N}},$ decoding map $d:\mathbb{R}^{N} \rightarrow \mathbb{R}$, the $\LMSE$ is defined as: 
\begin{equation}
    \LMSE(\mathcal{C}) = \mathbb{E}[|{AB}-\widetilde{C}|^2].
\end{equation}
where $\widetilde{{C}}$ is defined in (\ref{eq:decodout}). 
\end{definition}
The expectation in the above definition is over the joint distributions of the random variables ${A},{B},\tilde{R_i}|_{i=1}^{N},\tilde{S_i}|_{i=1}^{N}$. For fixed parameters, $N,t,\epsilon,\eta$, the goal of this paper is to characterize:
$$ \inf_{\mathcal{C}} \sup_{\mathbb{P}_{A,B}}\LMSE(\mathcal{C})$$
where the infimum\footnote{Since $|AB-\widetilde{C}|^2$ is a non-negative random variable, if $\mathbb{E}[|AB-\widetilde{C}|^2]$ does not exist for some coding scheme $\mathcal{C}$, then we interpret that $\mathbb{E}[|AB-\widetilde{C}|^2] =+\infty,$ with the convention that $+\infty$ is strictly greater than all real numbers.} is over the set of all coding schemes $\mathcal{C}$ that satisfy $t$-node $\epsilon$-DP, and the supremum is over all distributions $\mathbb{P}_{A,B}$ that satisfy Assumption \ref{assumption1}.

{In the remainder of this paper, we make the following zero mean assumptions: we assume that that $\mathbb{E}[A]=\mathbb{E}[B]=0$ and that $\mathbb{E}[\tilde{R}_i]= \mathbb{E}[\tilde{S}_i]=0, \forall i \in \{1,2,\ldots,N\}.$ With these assumptions, it suffices to assume that decoder is linear (that is, it is not just affine), since the optimal affine decoders are in fact linear. The reader may readily verify that our results hold without the zero-mean assumptions with affine decoders.}

\subsection{Signal-to-Noise Ratios}
\label{sec:SNR}
We take a two-step technical approach. First, we characterize accuracy-privacy trade-off in terms of signal-to-noise ratio (SNR) metrics (defined below). 
Second, we obtain the trade-off between mean square error and differential privacy parameters as corollaries to the SNR trade-off. %In estimation theory, the mean square error of a linear estimator for a random vector corrupted by independent additive noise is a function of the signal-to-noise ratio (defined appropriately for the vector context). {\color{red} Consequently,  our  intermediate trade-off in terms of the SNR metric can be interpreted  as  the mean squared error of the adversary is bounded.}

We define two SNR metrics: the privacy SNR and the accuracy SNR. Before defining these quantities, we begin with a fundamental lemma from linear estimation theory.

\vc{
\begin{lemma}
Let $X$ be a random variable with $\mathbb{E}[X] = 0$ and $\mathbb{E}[X^2] = \gamma^2$. Let $(Z_1, \ldots, Z_n)$ be a noise random vector that is independent of $X$ and let $$\vec{y} = \begin{bmatrix}
    \nu_1 X + Z_1 \\ \nu_2 X + Z_2 \\ \vdots \\ \nu_n X + Z_n
\end{bmatrix},$$
where  $\nu_1, \ldots, \nu_n \in \mathbb{R}$. Then, 
\begin{equation}
    \inf_{\vec{w} \in \mathbb{R}^{n\times 1}} \mathbb{E} \left[ |\vec{w} \cdot \vec{y} - X |^2 \right] = \frac{\gamma^2}{1+\texttt{SNR}},
    \label{eq:optimalw}
\end{equation}
where $\cdot$ denotes the vector dot product, and
\begin{equation}
    \texttt{SNR} = \frac{\textnormal{det}(\mathbf{K}_1)}{\textnormal{det}(\mathbf{K}_2)}-1, \label{eq:SNR}
\end{equation} where $\mathbf{K}_1$ is an $n \times n$ matrix whose $(i,j)$-th entry is $\mathbb{E}[y_i y_j] = \nu_i\nu_j \gamma^2 + \mathbb{E}[Z_iZ_j]$ and $\mathbf{K}_{2}$ is the covariance matrix of $(Z_1, Z_2, \ldots, Z_n)$. 

Furthermore, there exists $\vec{w}^{*}$ that optimizes the left hand side of \eqref{eq:optimalw}, and this vector satisfies:
$$ \mathbb{E} \left[ |\vec{w}^{*} \cdot \mathbf{y}' - X' |^2 \right] \leq  \frac{\gamma^2}{1+\texttt{SNR}} $$
for any random variable $X'$ where $\mathbb{E}[X'^2] \leq \gamma^2$ and 
$$\vec{y}' = \begin{bmatrix}
    \nu_1 X' + Z_1 \\ \nu_2 X' + Z_2 \\ \vdots \\ \nu_n X' + Z_n
\end{bmatrix}.$$
\label{lem:theSNRlemma}
\end{lemma}}

\vc{It is instructive to observe that the quantities in \eqref{eq:SNR} are covariance matrices, specifically, $\mathbf{K}_1 = \mathbb{E}[\vec{y}\vec{y}^{T}]$ and $\mathbf{K}_2 = \mathbb{E}[\vec{Z}\vec{Z}^{T}],$ where $\vec{Z} = \begin{bmatrix}Z_1 \\ Z_2 \\ \vdots \\ Z_n\end{bmatrix}.$ A proof of the lemma is placed in Appendix \ref{app:SNRlemmaproof}. The lemma motivates the following definitions.}

\begin{definition}(Privacy signal to noise ratio.)
Consider a secure multiplication coding scheme $\mathcal{C}.$ For any set $\mathcal{S}=\{s_1,s_2,\ldots,s_{|\mathcal{S}|}\} \subseteq \{1,2,\ldots,N\}$ of nodes where $s_1 < s_2 < \ldots < s_{|\mathcal{S}|}$, let $\mathbf{K}_{\mathcal{S}}^{\mathbf{R}}$ and $\mathbf{K}_{\mathcal{S}}^{\mathbf{S}}$ represent the covariance matrices of $\tilde{R}_{i}|_{i \in \mathcal{S}},\tilde{S}_{i}|_{i \in \mathcal{S}}.$ In particular, the $(i,j)$-th entry of $\mathbf{K}_{\mathcal{S}}^{\mathbf{R}}, \mathbf{K}_{\mathcal{S}}^{\mathbf{S}}$ are $\mathbb{E}[\tilde{R}_{s_i}\tilde{R}_{s_j}],\mathbb{E}[\tilde{S}_{s_i}\tilde{S}_{s_j}]$ respectively. Let $\mathbf{K}_{\mathcal{S}}^{A},\mathbf{K}_{\mathcal{S}}^{B}$ denote \vc{$|\mathcal{S}|\times|\mathcal{S}|$} matrices whose $(i,j)$-th entries respectively are $a_{s_{i}}a_{s_{j}} \eta$ and $b_{s_i}b_{s_j} \eta$ where $a_i,b_i$ are constants defined in (\ref{eq:model}). Then, the privacy signal-to-noise ratios corresponding to inputs ${A},
{B}$ denoted respectively as $\texttt{SNR}_{\mathcal{S}}^{A},\texttt{SNR}_{\mathcal{S}}^{B}$ are defined as:
$$\texttt{SNR}_{\mathcal{S}}^{A} = \frac{\text{det}(\mathbf{K}_{\mathcal{S}}^{A}+ \mathbf{K}_{\mathcal{S}}^{{R}})}{\text{det}(\mathbf{K}_{\mathcal{S}}^{{R}})}-1,$$
$$\texttt{SNR}_{\mathcal{S}}^{B} = \frac{\text{det}(\mathbf{K}_{\mathcal{S}}^{B} + \mathbf{K}_{\mathcal{S}}^{\mathbf{S}})}{\text{det}(\mathbf{K}_{\mathcal{S}}^{\mathbf{S}})}-1,$$
where $\text{`det'}$ denotes the determinant.
For \vc{integer} $t \leq N$, the $t$-node privacy signal-to-noise ratio of a $N$-node secure multiplication coding scheme $\mathcal{C}$, denoted as $\texttt{SNR}_{p}(\mathcal{C})$ is defined to be: 
$$ \texttt{SNR}_{p}(\mathcal{C}) = \max_{\mathcal{S} \subseteq \{1,2,\ldots,N\}, |\mathcal{S}|=t} \max(\texttt{SNR}_{\mathcal{S}}^{A},\texttt{SNR}_{\mathcal{S}}^{B}).$$
\label{def:SNRp}
\end{definition}

\begin{remark}\vc{ Because of Lemma \ref{lem:theSNRlemma}, we can infer that in the case of $\mathbb{E}[A^2] = \eta$, an adversary with access to the nodes in set $\mathcal{S}$ can obtain a linear combination of the inputs to these nodes to recover, for example, $A$ with a mean square error of $\frac{\eta}{1+\texttt{SNR}_\mathcal{S}^{\mathcal{A}}}$. This mean square error is an alternate metric --- as compared to DP --- for privacy leakage that will be used as an intermediate step in deriving our results. While it leads to a transparent privacy-accuracy trade-off, the privacy SNR metric has at least two shortcomings as compared to $\epsilon$-DP. First, unlike DP, it depends on some assumptions on the distribution of $\mathbb{P}_A$, e.g., those made in Assumption \ref{assumption1}. Second, because of its connection to mean squared error, the metric depends only on linear combinations of the adversary's input, unlike DP which is an information-theoretic metric and implicitly captures all possible functions of the adversary. We emphasize that we our paper presents converse and tight achievability results that apply for the more standard DP metric that does not use the above assumptions.}% The mean squared error metric is only used as an intermediate step, and we provide $\epsilon$-DP guarantees for both our achievable scheme and our converse.}
\label{rem:MSE_interpret}
\end{remark}

%\textcolor{red}{[Haewon] There's no $\mathbf{I}$ in the definition and we might want to explain the SNR formulation in words a little bit?}

Next we define the accuracy signal-to-noise ratios. From the definition of $\tilde{C}_{i}$ in (\ref{eq:comp_output}), we observe that:
$$ \widetilde{C}_i = a_ib_i AB + a_i A\tilde{S}_i + b_i B\tilde{R}_i+ \tilde{R}_i\tilde{S}_i.$$
To understand the following definition, it helps to note that in $\mathbb{E}[\widetilde{C}_{i}\widetilde{C}_{j}]$, the ``signal'' component, $\mathbb{E}[AB]$,  has the coefficient $a_ib_ia_jb_j$.

\begin{definition}(Accuracy signal to noise ratio.)\label{def:SNRa}
Consider a secure multiplication coding scheme $\mathcal{C}$ over $N$ nodes. Let $\mathbf{K}_1$ denote the $N \times N$ matrix whose $(i,j)$-th entry is $\mathbb{E}[\tilde{C}_i \tilde{C}_j]$ where $\tilde{C}_{i},\tilde{C}_{j}$ are as defined in (\ref{eq:comp_output}). Let $\mathbf{K}_{2}$ denote the matrix whose $(i,j)$-th entry is $\mathbb{E}[\tilde{C}_i \tilde{C}_j] - a_ib_ia_jb_j \eta^2$, where $a_i,b_i,a_j,b_j$ are constants associated with the coding scheme as per (\ref{eq:model}). Then, the accuracy signal-to-noise of the coding scheme $\mathcal{C},$ denoted as $\texttt{SNR}_{a},$ is defined as:
\begin{equation} \label{eq:SNRa}
    \texttt{SNR}_{a}(\mathcal{C}) = \frac{\text{det}(\mathbf{K}_{1})}{\text{det}(\mathbf{K}_{2})}-1.
\end{equation}
\end{definition}
We drop the dependence on the coding scheme $\mathcal{C}$ from $\texttt{SNR}_{a},\texttt{SNR}_{p}$ in this paper when the coding scheme is clear from the context. \vc{In Definition \ref{def:SNRa} the quantity $\mathbf{K}_1$ is the covariance matrix of $\begin{bmatrix}{\tilde{C}_1} \\  {\tilde{C}_1} \\ \vdots \\ \tilde{C}_{N} \end{bmatrix}$ for the case where $\mathbb{E}[A^2B^2] = \eta^2.$ Similarly, if $\mathbb{E}[A^2]=\mathbb{E}[B^2] = \eta,$ the quantity $\mathbf{K}_2$ is the covariance matrix of random vector $\begin{bmatrix}{\tilde{C}}_1-a_1 b_1 AB \\  {\tilde{C}}_2 - a_2b_2AB \\ \vdots \\ \tilde{C}_{N}-a_2b_2AB \end{bmatrix},$ which can be interpreted as an additive noise with respect to the ``signal'' $AB$ at the decoder.}

The following lemma follows from Lemma \ref{lem:theSNRlemma}.

\begin{lemma}
For a coding scheme $\mathcal{C}$ with accuracy signal-to-noise ratio $\texttt{SNR}_{a},$ for inputs $A,B$ that satisfy Assumption \ref{assumption1}, we have:
$$ \LMSE(\mathcal{C}) \leq \frac{\eta^2}{1+\texttt{SNR}_{a}},$$
with equality if and only if $E[A^2] = E[B^2] = \eta.$
\label{lem:LMSE_SNR}
\end{lemma}

\subsection{Statement of Main Results}
The main result of this paper is a tight characterization of the trade-off between the accuracy, evaluated in terms of the mean squared error, and the privacy, evaluated in terms of the DP parameter, for $t < N < 2t+1.$. Our main result can also be interpreted in terms of the optimal privacy-accuracy trade-off with accuracy evaluated via the accuracy signal-to-noise, $\texttt{SNR}_{a}$, and privacy, via the privacy signal-to-noise, $\texttt{SNR}_p$ as follows:
\begin{tcolorbox}\begin{equation}(1+\texttt{SNR}_a) = (1+\texttt{SNR}_p)^2.\label{eq:SNR_tradeoff}\end{equation}
\end{tcolorbox}

%\begin{remark}
 %    Lemma \ref{lem:LMSE_SNR} and Remark \ref{rem:MSE_interpret} lead to the following interpretation of our trade-off in (\ref{eq:SNR_tradeoff}). Suppose the privacy leakage were measured - rather than as the DP parameter - as the mean squared error of an adversary attempting to infer the data ($A$ or $B$) by performing linear combinations of its inputs. Then, (\ref{eq:SNR_tradeoff})  implies that among all coding schemes $\mathcal{C}$ with privacy leakage at least $\beta$ - where leakage is measured as mean squared error of an adversary restricted to performing linear combination - we have $\inf_{\mathcal{C}}\LMSE(\mathcal{C}) = \beta^2.$
%\end{remark}

We state the results more formally below, starting with the achievability result.

\begin{theorem}
Consider positive integers $N,t$ with $N > t$. For every $\delta > 0,$ and for every strictly positive parameter $\texttt{SNR}_{p} > 0$ there exists a $N$-node secure multiplication coding scheme $\mathcal{C}$ with $t$-node privacy signal-to-noise, $\texttt{SNR}_{p}$ and an accuracy $\texttt{SNR}_{a}$ that satisfies:
$$ \texttt{SNR}_{a} \geq 2 \texttt{SNR}_p + \texttt{SNR}_p^2-\delta.$$
\label{thm:main_achievability}
\end{theorem}
Notably, it suffices to show the achievability for $N=t+1.$ If $N > t+1,$ the $(t+1)$-node secure multiparty multiplication scheme can be utilized for the first $t+1$ nodes and the remaining nodes can simply receive $0$. We now translate the achievability result in terms of $\epsilon$-DP. For $\epsilon > 0,$ let $\mathcal{S}_{\epsilon}(\mathbb{P})$ denote the set of all real-valued random variables that satisfy $\epsilon$-DP, that is, $X \in\mathcal{S}_{\epsilon}(\mathbb{P})$ if and only if:
$$ \sup \frac{\mathbb{P}(X+X' \in \mathcal{A})}{\mathbb{P}(X+X'' \in \mathcal{A})} \leq e^{\epsilon}$$
where the supremum is over all constants $X',X'' \in \mathbb{R}$ that satisfy $|X'-X''| \leq 1$ and all subsets $\mathcal{A} \subset \mathbb{R}$ that are in the Borel $\sigma$-field. Let $L^2(\mathbb{P})$ denote the set of all real-valued random variables with finite variance. Let 
$$ \sigma^{*}(\epsilon) = \inf_{X \in \mathcal{S}_\epsilon(\mathbb{P}) \cap L^2(\mathbb{P})} \mathbb{E}\left[(X-\mathbb{E}[X])^2\right].$$

In plain words, $\sigma^*(\epsilon)$ denotes the smallest noise variance that achieves single user differential privacy parameter $\epsilon$. It is worth noting that $\sigma^{*}(\epsilon)$ has been explicitly characterized in \cite{Pramod2015}, Theorem 7, as:
\begin{equation} (\sigma^{*}(\epsilon))^2 = \frac{2^{2/3}e^{-2\epsilon/3}(1+e^{-2\epsilon/3})+e^{-\epsilon}} {(1-e^{-\epsilon})^2}. \label{eq:optimalsigma}
\end{equation}

\begin{corollary}
Consider positive integers $N,t$ with $N \leq 2t$. Then, for every $\epsilon,\delta > 0,$ there exists a coding scheme $\mathcal{C}$ that achieves $t$-node $\epsilon$-DP,
$$ \LMSE(\mathcal{C}) \leq \frac{\eta^2 (\sigma^*(\epsilon))^4}{(\eta+(\sigma^*(\epsilon))^2)^2}+\delta.$$
\label{cor:achievability}
\end{corollary}

Theorem \ref{thm:main_achievability} and Corollary \ref{cor:achievability} are shown in Section \ref{sec:achievability}.

\begin{remark} \vc{ For the case where $N \geq 2t+1$, perfect accuracy can be achieved by real-valued Shamir Secret Sharing. Specifically, consider the secure multiplication coding scheme that treats quantities in \eqref{eq:ss} as real-valued and sets $\tilde{A_i} = p_A(x_i), \tilde{B}_i = p_B(x_i)$ where $x_1, x_2, \ldots, x_n \in \mathbb{R}$ as arbitrary distinct scalars. Then the output of the $i$th node $\tilde{C}_i$ is an evaluation of the degree $2t$ polynomial $p_A(x)p_B(x)$ whose constant term is $AB.$ This constant term can be recovered perfectly by the decoder from $N \geq 2t+1$ nodes through polynomial interpolation. Thus, the scheme achieves $\texttt{SNR}_a = \infty.$ 
Further, by letting random variables $R_i|_{i=1}^{t},S_{i}|_{i=1}^{t}$ to be statistically i.i.d. random variables with arbitrarily large variance, the privacy SNR $\texttt{SNR}_p$ can be made arbitrarily small. To see this, consider an adversary that observes the inputs to $\mathcal{S}=\{s_1,s_2,\ldots, s_t\},$ where $s_1 < s_2 < \ldots < s_t.$ The observation of the adversary can be expressed as:
 $$ \vec{A}_{\mathcal{S}}:= A \vec{1} + \mathbf{V}\vec{R}_{\mathcal{S}} $$
$$ \vec{B}_{\mathcal{S}}:= B\vec{1} + \mathbf{V}\vec{S}_{\mathcal{S}} ,$$
where $\vec{1}$ is a $t \times 1$ all ones vector, $\vec{R}_{\mathcal{S}}, \vec{S}_{\mathcal{S}}$ are $t \times 1$ vectors whose $i$-th entries are $R_{s_i},S_{s_i}$ and $\mathbf{V}$ is a $t \times t$ Vandermonde matrix whose $(i,j)$-th entry is $x_{s_i}^{j}.$ Because $\mathbf{V}$ is invertible, the adversary's information is equivalent to:

 $$ \mathbf{V}^{-1}\vec{A}_{\mathcal{S}}= A \mathbf{V}^{-1} \vec{1} + \mathbf{V}\vec{R}_{\mathcal{S}} $$
$$ \mathbf{V}^{-1}\vec{B}_{\mathcal{S}}= B \mathbf{V}^{-1} \vec{1}+ \vec{S}_{\mathcal{S}}.$$
Clearly by letting the variance of $R_i, S_i, i=1,2,\ldots, t$ to be arbitrarily large (keeping scalars $x_1, x_2,\ldots, x_N$ fixed), we can make $\texttt{SNR}_{p}$ arbitrarily small. Therefore, the point $(\texttt{SNR}_{a}=\infty, \texttt{SNR}_{p}=0)$ is achievable. Further, if we choose $R_1, R_2, \ldots, R_t, S_1, S_2, \ldots, S_t$ to be i.i.d., Laplace random variables with variance $\sigma$, the above scheme achieves $\epsilon$-DP, with $\epsilon \to 0$ as $\sigma \to \infty$.}
\label{eq:remark5}
\end{remark}
\begin{remark}For the case of $N=t,$ we readily show $\texttt{SNR}_{a} \leq \texttt{SNR}_{p}$ and consequently, we have $\texttt{LMSE} \geq \frac{\eta^2}{1+\texttt{SNR}_{p}}.$ To see this, consider the decoder with coefficients $d_1,d_2,\ldots,d_N$. By definition of the privacy SNR and basic linear estimation theory, we have: $\mathbb{E}\left[(\sum_{i=1}^{N} d \tilde{A}_{i} - A)^2\right]\geq \frac{\eta}{1+\texttt{SNR}_{p}}.$ Then, we can bound the mean square error of the decoder as:
\begin{eqnarray*}
    \texttt{LMSE} &=& \mathbb{E}\left[\left(\sum_{i=1}^{N}w_i \tilde{A}_{i}\tilde{B}_i - AB)^2 \right)\right]\\&=&
     \mathbb{E}\left[\left(\sum_{i=1}^{N}w_i \tilde{A}_{i}({B}+\tilde{S}_i) - AB)^2 \right)\right]\\
     \\&\stackrel{(a)}{=}&
     \mathbb{E}\left[\left(\sum_{i=1}^{N}w_i \tilde{A}_{i}{B} - AB \right)^2+\left(\sum_{i=1}^{N}w_i\tilde{A}_i \tilde{S}_i  \right)^2 \right]\\ &{\geq} &
     \mathbb{E}\left[\left(\sum_{i=1}^{N}w_i \tilde{A}_{i}{B} - AB \right)^2\right]\\
     &\geq & \frac{\eta^2}{1+\texttt{SNR}_p},
\end{eqnarray*}
where in $(a),$ we have used the fact that $\tilde{S}_i$ is uncorrelated with $B$ and $\mathbb{E}[B]=0$. Therefore, if we simply add one node to go from $N=t$ to $N=t+1,$ then our main achievability result implies that the quantity $\frac{\texttt{LMSE}}{\eta^2}$ becomes equal to, or smaller than squared. Because\footnote{Even a decoder that ignores all the computation outputs and predicts $\widetilde{C}=0$ obtains $\LMSE = \eta^2$.} $\frac{\texttt{LMSE}}{\eta^2} < 1,$ our achievability result implies a potentially significant reduction in the mean squared error for the case of $N = t+1$ as compared to the case of $N=t.$  
\end{remark}

We next state our converse results.

\begin{theorem}
Consider positive integers $N,t$ with $N \leq 2t$. For any $N$ node secure multiplication coding scheme $\mathcal{C}$ with accuracy signal-to-noise ratio $\texttt{SNR}_a$ and $t$-node privacy signal-to-noise $\texttt{SNR}_{p}:$
$$ \texttt{SNR}_{a} \leq 2 \texttt{SNR}_p + \texttt{SNR}_p^2.$$
\label{thm:main_converse}
\end{theorem}

\begin{corollary}
Consider positive integers $N,t$ with $N \leq 2t$. For any coding scheme ${\mathcal{C}}$ that achieves $t$-node $\epsilon$-DP, there exists a distribution $\mathbb{P}_{A,B}$ that satisfies Assumption \ref{assumption1} and
$$ \LMSE(\mathcal{C}) \geq \frac{\eta^2( \sigma^*(\epsilon))^4}{(\eta+(\sigma^*(\epsilon))^2)^2}.$$
\label{cor:converse}
\end{corollary}
In fact, our converse shows so long as $\mathbb{E}[A^2]=\mathbb{E}[B^2]=\eta,$ the lower bound of the above corollary is satisfied. Theorem \ref{thm:main_converse} and Corollary \ref{cor:converse} are shown in Section \ref{app:converse}. %By substituting bounds for $\sigma^2(\epsilon),$ one naturally obtains bounds on the privacy-accuracy trade-offs. For instance \cite{ISIT2022} provides the following bounds: $\frac{\epsilon^{2}}{8}\leq  \sigma^*(\epsilon) \leq e^{\epsilon}-1$.

\section{Achievability: Proofs of Theorem \ref{thm:main_achievability} \vc{and Corollary \ref{cor:achievability}}} \label{sec:achievability}
To prove Theorem \ref{thm:main_achievability}, it suffices to consider the case where $N=t+1$. In our achievable scheme, we assume that node $i$ receives: 
$$\Gamma_i = \left[A~R_1~R_2~\ldots~R_t\right] \vec{v}_{i},$$
$$\Theta_i = \left[B~S_1~S_2~\ldots~S_t\right] \vec{w}_{i}.$$

\noindent where $\vec{v}_{i},\vec{w}_{i}$ are $(t+1) \times 1$ vectors. We assume that $R_i\big|_{i=1}^{t}, S_{i}\big|_{i=1}^{t}$ are zero mean unit variance statistically independent random variables.  Node $i$ performs the computation  
$$\tilde{C}_i = \Gamma_i \Theta_i.$$

Our achievable coding scheme prescribes the choice of vectors $\vec{v}_{i},\vec{w}_{i}$. Then, we analyze the achieved privacy and accuracy. Our proof for $t > 1$ is a little bit more involved than the proof for $t=1.$ The description below applies for all cases for $t,$ and includes the simplifications that arise for the case of $t=1.$

\vc{We describe the coding scheme in Sec. \ref{sec:coding_scheme}. The privacy and accuracy analyses of the coding schemes are respectively placed in Sections \ref{sec:privacy_analysis} and \ref{sec:accuracy} respectively. Corollary \ref{cor:achievability} is proved in Sec. \ref{app:DP}.}

\subsection{Description of Coding Scheme}
\label{sec:coding_scheme} 

\vc{ To prove the theorem, for a fixed parameter $\texttt{SNR}_{p} > 0$, we develop a sequence of coding schemes indexed by parameter $n$, such that each coding scheme has an privacy SNR at least $\texttt{SNR}_p$. Further, we show that for any constant $\delta > 0$, there exists a sufficiently integer $n$ such that the $n$-th coding scheme in the sequence satisfies $(1+\texttt{SNR}_a) \geq (1+\texttt{SNR}_p)^2 - \delta.$}
Let $\alpha_1^{(n)},\alpha_2^{(n)}$ be strictly positive  sequences such that:

\begin{equation} 
    \lim_{n \to \infty} \frac{\alpha_1^{(n)}}{\alpha_2^{(n)}} = 
\lim_{n \to \infty} \alpha_2^{(n)} = \lim_{n\to\infty} \frac{(\alpha_2^{(n)})^2}{\alpha_1^{(n)}} = 0 
\label{eq:limits}
\end{equation} 
Notice that the above automatically imply that $\lim_{n\to\infty} \alpha_1^{(n)} = 0$. As an example, $\alpha_{2}^{(n)}$ can be chosen to be an arbitrary sequence of positive real numbers that converge to $0$, and we can set $\alpha_{1}^{(n)} = (\alpha_{2}^{(n)})^{1.5}$; indeed, this is the choice depicted in Fig. \ref{fig:code_construction} in Sec. \ref{sec:summary}.

For $t > 1$, let $\mathbf{G}=\begin{bmatrix}\vec{g}_{1}&\vec{g}_2& \ldots & \vec{g}_{t}\end{bmatrix}$ be any $(t-1)\times t$ matrix such that:
\begin{enumerate}
\item[(C1)] every $(t-1) \times (t-1)$ sub-matrix is full rank, 
\item[(C2)] $\begin{bmatrix}
    1 & 1 & \ldots & 1\\
   \vec{g}_1 & \vec{g}_2 & \ldots & \vec{g}_t
\end{bmatrix}$ has a full rank of $t.$
\end{enumerate}

For $t > 1$, the $n$th member of the sequence of coding schemes sets
\begin{eqnarray}
     \vec{v}_{t+1} =\vec{w}_{t+1} &=& \begin{bmatrix}1 \\ x \\ 0 \\ \vdots \label{eq:ach_v1_tgreat1}\\ 0\end{bmatrix},\\ \vec{v}_{i} = \vec{w}_{i} &=& \vec{v}_{t+1} +  \begin{bmatrix}0 \\ \alpha_1^{(n)} \\ \alpha_{2}^{(n)} \vec{g}_{i}\end{bmatrix}, 1 \leq i \leq t,\label{eq:ach_v2_tgreat1}
 \end{eqnarray}
and for $t=1$, we simply have
\begin{eqnarray}
    \vec{v}_{2} = \vec{w}_{2} &=& \begin{bmatrix}1 \\ x \end{bmatrix} \label{eq:ach_v2_teq1},\\ \vec{v}_{1} = \vec{w}_{1} &=& \vec{v}_{2} + \begin{bmatrix}0 \\ \alpha_1^{(n)}R_1 \end{bmatrix},\label{eq:ach_v2_teq2}
\end{eqnarray}

\noindent where $x > 0$ is a parameter whose role becomes clear next. In the pictorial description of our coding scheme in Fig. \ref{fig:code_construction} in Sec. \ref{sec:summary}, we selected $\mathbf{G}=\begin{bmatrix} 1 & -1 \end{bmatrix}$. 

\vc{In the sequel, we denote a function $f(n) = O(g{(n)})$ if $\lim_{n \to \infty} \left|\frac{f(n)}{g(n)}\right| < \infty.$ We treat $\alpha_1^{(n)},\alpha_2^{(n)}$ as functions of $n$.}

% Notice that the input to node $i$ corresponding to $A$ for $i < t+1$ can be interpreted a superposition of three ``layers'' as follows:
%$$\underbrace{A + R_1x}_{\text{First Layer}}+ \underbrace{\alpha_2^{(n)} \begin{bmatrix}R_2& R_3 &\ldots & R_t\end{bmatrix}\vec{g}_{i}}_{\text{Second layer}}+ \underbrace{\alpha_1^{(n)}R_1}_{\text{Third Layer}}.$$  For fixed parameters $x,\eta$, the first layer has magnitude $O(1),$ the second layer has magnitude $O(\alpha_2^{(n)}),$ and the third layer has magnitude $O(\alpha_1^{(n)}).$ Similarly, the input corresponding to $B$ can also be interpreted as a superposition of three layers. The layer-based interpretation of the coding scheme will be utilized in our explanations below.

%\begin{figure*}[!ht]
 %   \centering
  % \begin{subfigure}{0.9\textwidth}  \centering\includegraphics[width=0.75\textwidth]{figs/DP-MPC-ach1.png} 
  %     \end{subfigure}
%\begin{subfigure}{0.9\textwidth} \centering
 %   \includegraphics[width=0.75\textwidth]{figs/DP-MPC-ach2.png}
    
%\end{subfigure}
 %  \caption{\small{Pictorial depiction of the coding scheme for $t=2, N=3$ with matrix $\mathbf{G}=\begin{bmatrix} 1 & -1 \end{bmatrix}$.}}
  %  \label{fig:code_construction}
  %  \vspace{-7pt}
%\end{figure*}

\subsection{Privacy Analysis}
\label{sec:privacy_analysis}
\underline{Informal privacy analysis:} For expository purposes, we first provide a coarse privacy analysis with informal reasoning.
With the above scheme, we claim that 
$\texttt{SNR}_{p} \approx \eta/x^2,$ and so, it suffices to choose $x \approx \sqrt{\frac{\eta}{\texttt{SNR}_p}}.$ Consider $A$'s privacy constraint, we require $\texttt{SNR}_{\mathcal{S}}^{A} \leq \texttt{SNR}_{p}$ for every $\mathcal{S} \subset \{1,2,\ldots,N\}, |\mathcal{S}|=t$. First we consider the scenario where $\mathcal{S}=\{1,2,\ldots,t\}$. Each node's input is of the form $A + R_1(x+ \alpha_1^{(n)}) + \alpha_2^{(n)} \begin{bmatrix}R_2& R_3 &\ldots & R_t\end{bmatrix}\vec{g}_{i}.$
Even if an adversary with access to the inputs to nodes in $\mathcal{S}$ happens to know $R_2,R_3,\ldots,R_{t}$, but not $R_1,$ the noise $(x+\alpha_1^{(n)}) R_1$ provides enough privacy, that is the privacy signal to noise ratio for this set is $\approx \eta/x^2$ for sufficiently large $n$.  

Now, consider the cases where the set $\mathcal{S}$ of $t$ colluding adversaries includes node $t+1$. In this case, the adversary has $A+R_1x$ from node $t+1.$ The other $t-1$ colluding nodes have inputs: $A + R_1(x+ \alpha_1^{(n)}) + \alpha_2^{(n)} \begin{bmatrix}R_2 & R_3 & \ldots R_3 & \ldots & R_t\end{bmatrix}\vec{g}_{i}$ for $i \in \mathcal{S}-\{t+1\}.$
Informally, this can be written as $A+R_1x + R_1 \alpha_1^{(n)} + \Omega(\alpha_2^{(n)})Z_i,$ for some random variable $Z_i$ with variance $\Theta(1)$.

On the one hand, observe that these $t-1$ nodes contain a linear combination of $A,R_1$ that is linearly independent of the input to the $(t+1)$-th node (which is $A+xR_1$).  
It might seem possible for the adversary to increase its signal-to-noise ratio beyond $\frac{\eta}{x^2}$ by combining the input of these $t-1$ nodes with node $t+1$'s input. However, observe crucially that the first layer of the inputs to these $t-1$ nodes is linearly \emph{dependent} with $t+1$'s input. The privacy signal-to-noise ratio can be  increased by a non-negligible extent at the adversary only if it is able to access information in the third layer. Since $|\alpha_2^{(n)}| \gg |\alpha_1^{(n)}|$, in order to access the information in the third layer and reduce/cancel the effect of $R_1$, the adversary must first be able to cancel the second layer terms whose magnitude is $\Omega(\alpha_2^{(n)})$. But these second layer terms are a combination of $t-1$ independent noise variables $R_2,R_3,\ldots,R_t$ that are modulated by linearly independent vectors. Hence, any non-trivial linear combination of these $t-1$ inputs necessarily contains a non-zero $\Omega(\alpha_2^{(n)})$ additive noise term.  So, their effect cannot be canceled and the $\alpha_1^{(n)} R_1$ term in the third layer is hidden from the decoder (See Fig. \ref{fig:code_construction}). Consequently, as $n \to \infty$, the input to these $t-1$ nodes is, approximately, a statistically degraded version of $A+xR_1$. Therefore, the parameter $\texttt{SNR}_p$ cannot be increased beyond $\frac{\eta}{x^2}.$ 

\underline{Formal privacy analysis:} We now present a formal privacy analysis.
We show that for any $\delta > 0$, by taking $n$ sufficiently large, we can ensure that: \begin{equation}\max(\texttt{SNR}_{\mathcal{S}}^{(A)}, \texttt{SNR}_{\mathcal{S}}^{(B)}) \leq \frac{\eta}{x^2} + \delta \label{eq:SNRpbound_ach}\end{equation}
for every subset $\mathcal{S}$ of $t$ nodes. Because of the symmetry of the coding scheme, it suffices to show that $\texttt{SNR}_{\mathcal{S}}^{(A)}$ satisfies the above relation. In our analysis, we will repeatedly use the fact that any linear combination $\sum_{i \in \mathcal{S}}\beta_i \tilde{A}_{i}$ of the inputs to the adversary satisfies:
$$ \mathbb{E}\left[\left(\left(\sum_{i\in\mathcal{S}} \beta_i \tilde{A}_{i}\right)-A\right)^2\right] \geq \frac{\eta}{1+\texttt{SNR}_{\mathcal{S}}^{(A)}} $$

First consider the case where $t+1 \notin \mathcal{S}$. For each $i \in \mathcal{S},$ the input $\tilde{A}_{i}$ is of the form $A+(x+\alpha_1^{(n)})R_1 + Z_i,$ where $Z_i$ is zero mean random variable that is statistically independent of $(R_1,A).$ Therefore, we have:
\begin{align*}
&\inf_{\beta_i \in \mathbb{R}, i \in \mathcal{S}} 
  \mathbb{E} \left[\left(\left(\sum_{j\in\mathcal{S}} \beta_j \tilde{A}_{j}\right)-A\right)^2\right]\\
&  \vc{ \inf_{\beta_i \in \mathbb{R}, i \in \mathcal{S}} 
  \mathbb{E} \left[\left(\left(\sum_{j\in\mathcal{S}} \beta_j\right)(A+(x+\alpha_1^{(n)})R_1)-A\right)^2 + \sum_{j\in\mathcal{S}}\beta_j^2 \mathbb{E}[Z_i^2]\right] }\\
&\stackrel{(a)}{\geq} \inf_{\beta \in \mathbb{R}} \mathbb{E}\left[\left(\beta(A+(x+\alpha_1^{(n)})R_1)-A\right)^2\right]\\
&= \frac{\eta}{1+\frac{\eta}{(x+\alpha_1^{(n)})^2}} \geq \frac{\eta}{1+\frac{\eta}{x^2}}.
\end{align*} 
where $(a)$ holds because $\sum_{j \in \mathcal{S}}\beta_j^2\mathbb{E}[Z_j^2$ is non-negative, and hence can be dropped to obtain a bound.
Consequently: $\texttt{SNR}_{\mathcal{S}}^{(A)} \leq \frac{\eta}{x^2}.$

Now consider the case: $t+1 \in \mathcal{S}$. Consider a linear estimator:
\begin{align*} &\hat{A}
= &\beta_{t+1} (A + xR_{1}) + \sum_{i \in \mathcal{S} \setminus \{t+1\}} \beta_{i} \left( A+ R_{1}(x+\alpha_1^{(n)})+ \alpha_{2}^{(n)} \begin{bmatrix}R_2& R_3 & \ldots & R_t\end{bmatrix} \vec{g}_{i}\right)\\
&=&A\left(\sum_{i \in \mathcal{S}} \beta_{i}\right) +  R_{1} \left(x \sum_{i \in \mathcal{S}} \beta_{i} + \alpha_1^{(n)} \sum_{i\in \mathcal{S} \setminus \{t+1\}} \beta_i\right)+ \alpha_{2}^{(n)} \begin{bmatrix}R_2& R_3  & \ldots & R_t\end{bmatrix}\left(\sum_{i\in\mathcal{S} \setminus \{ t+1\}}\beta_i \vec{g}_{i}\right)
\end{align*}

Because of property (C1), there are only two possibilities: (i) $\beta_i = 0,$ for all $i\in \mathcal{S}\setminus \{t+1\}$, or (ii) $\left(\sum_{i\in\mathcal{S} \setminus \{t+1\}}\beta_i \vec{g}_{i}\right) \neq 0$. In the former case, the linear combination is $\hat{A} = \beta_{t+1}(A_{t+1} + x R_1)$ from which, the best linear estimator has signal to noise ratio $\eta/x^2$ as desired. Consider the latter case, let $\rho > 0$ be the smallest singular value among the singular values of all the $(t-1) \times (t-1)$ sub-matrices of $\mathbf{G}$. We bound the noise power of $\hat{A}$ below; in these calculations, we use the fact that $R_i$ are zero-mean unit variance uncorrelated random variables for  $i=1,2,\ldots,t$.
\begin{align*}
	&\left(x \sum_{i \in \mathcal{S}} \beta_{i} + \alpha_1^{(n)} \sum_{i\in \mathcal{S} \setminus \{t+1\}} \beta_i\right)^2 + (\alpha_2^{(n)})^2 \mathbb{E}\left[ \left(\begin{bmatrix}R_2& R_3 & \ldots & R_t\end{bmatrix} \sum_{i\in\mathcal{S}\setminus\{t+1\}}\beta_i \vec{g}_{i}\right)^2\right]\\
&= \left(x \sum_{i \in \mathcal{S}} \beta_{i} + \alpha_1^{(n)} \sum_{i\in \mathcal{S}\setminus\{t+1\}} \beta_i\right)^2 + (\alpha_2^{(n)})^2 \left|\left|\sum_{i\in\mathcal{S}\setminus\{t+1\}}\beta_i \vec{g}_{i} \right|\right|^2 \\
& \geq \left(x \sum_{i \in \mathcal{S}} \beta_{i} + \alpha_1^{(n)} \sum_{i\in \mathcal{S}\setminus\{t+1\}} \beta_i\right)^2 + (\alpha_2^{(n)})^2 \rho^2 \sum_{i\in\mathcal{S}\setminus\{t+1\}}\beta_i^2 
\end{align*}

The signal-to-noise ratio for signal $A$  in $\hat{A}$ can be bounded as:
\begin{eqnarray*}
&&\frac{\eta \left(\sum_{i \in \mathcal{S}} \beta_{i}\right)^2 }{\left(x \sum_{i \in \mathcal{S}} \beta_{i} + \alpha_1^{(n)} \sum_{i\in \mathcal{S} \setminus \{t+1\}} \beta_i\right)^2 + (\alpha_2^{(n)})^2 \rho^2 \sum_{i\in\mathcal{S} \setminus\{t+1\}}\beta_i^2  }\\
& \stackrel{(a)}{\leq} &\frac{\eta \left(\sum_{i \in \mathcal{S}} \beta_{i}\right)^2 }{x^2 \left(\sum_{i \in \mathcal{S}} \beta_{i}\right)^2 + 2 x \alpha_1^{(n)} \left(\sum_{i \in \mathcal{S} \setminus \{t+1\}} \beta_i\right)\left(\sum_{i\in \mathcal{S}}\beta_i\right) + (\alpha_2^{(n)})^2 \rho^2 \sum_{i\in\mathcal{S} \setminus \{t+1\}}\beta_i^2  }\\
& = &\frac{\eta}{x^2 + 2 x \alpha_1^{(n)} \nu_1 + (\alpha_2^{(n)})^2 \rho^2 \nu_2^2 } 
\stackrel{(b)}{\leq} \frac{\eta}{x^2 - 2 x \alpha_1^{(n)} \sqrt{t} \nu_2 + (\alpha_2^{(n)})^2 \rho^2 \nu_2^2 } \stackrel{(c)}{\leq} \frac{\eta}{x^2 - \frac{(\alpha_1^{(n)})^2}{(\alpha_2^{(n)})^2} \frac{x^2 t}{\rho^2},}
\end{eqnarray*}

\vc{ where we have used the notation $\nu_1 \stackrel{\Delta}{=} \frac{\sum_{i \in \mathcal{S}\backslash \{t+1\}}\beta_i}{\sum_{i\in \mathcal{S}}\beta_i}$ and $\nu_2^2 \stackrel{\Delta}{=}\frac{\sum_{i\in\mathcal{S} \setminus \{t+1\}}\beta_i^2 }{(\sum_{i\in \mathcal{S}}\beta_i)^2}$}. The upper bound of (a) holds because we have replaced the denominator by a smaller quantity.
In (b), we have used the fact that $\nu_1^2 \leq t \nu_2^2$ and consequently $-\sqrt{t} \nu_2 \leq \nu_1 \leq \sqrt{t} \nu_2 .$ (c) holds because $$ \inf_{\nu_2}(\alpha_2^{(n)})^2 \rho^2 \nu_2^2 -  2 x \alpha_1^{(n)} \sqrt{t} \nu_2 = - \frac{x^2 (\alpha_1^{(n)})^2 t}{(\alpha_2^{(n)})^2 \rho^2}.$$ As $n \to \infty,$ (\ref{eq:limits}) implies that  $\frac{(\alpha_1^{(n)})^2}{(\alpha_2^{(n)})^2} \to 0$, and consequently, for any $\delta > 0,$ we can choose a sufficiently large $n$ to ensure that the right hand side of $(c)$ can be made smaller than $\frac{\eta}{x^2} + \delta$. Thus, for sufficiently large $n$, $\texttt{SNR}_p \leq \frac{\eta}{x^2}+\delta$ for any $\delta > 0.$ 

%Slightly more rigorously, the privacy argument is as follows: for any set $\mathcal{S}$ of $t$ nodes, we have:
%$$\sum_{i \in \mathcal{S}} \beta_i \Gamma_i = (A+R_1x) +  \alpha_2^{(n)} Z_1 + c \alpha_1^{(n)} R_1 $$
%where $Z_1$ is a linear combination of $R_2,R_3,\ldots,R_t$
%There are only 3 (not necessarily mutually exclusive) possibilities:
%\begin{itemize}
%\item $c=0$
%\item $c=1$
%\item $c$...
%\end{itemize}

\subsection{Accuracy Analysis}
\label{sec:accuracy}
\vc{To complete the proof of Theorem \ref{thm:main_achievability}, because of \eqref{eq:SNRpbound_ach},} it suffices to show that for any $\delta > 0,$ we can obtain $\texttt{SNR}_{a} \geq \frac{(x^2+\eta)^2}{x^4}-1-\delta$ 

\underline{Informal Accuracy Analysis:} We provide a high-level description of the accuracy analysis for the case of $t=2,N=3.$ We assume that $\mathbf{G} = [1~~-1]$ like in Fig. \ref{fig:code_construction}. The computation outputs of the three nodes are:
 \begin{align*}
     \widetilde{C}_1 &=  (A+R_1(x+\alpha_1^{(n)}) + \alpha_2^{(n)}R_2 )(B+S_1(x+\alpha_1^{(n)})+\alpha_2^{(n)}R_2)\\&
     = \underbrace{\widetilde{C}_3}_{\text{First layer}} + \underbrace{\alpha_2^{(n)}(S_2(A+R_1x)+R_2(B+S_1x))}_{\text{Second Layer}} + \underbrace{\alpha_1^{(n)}(S_1(A+R_1x)+R_1(B+S_1x))}_{\text{Third Layer}} + \underbrace{O(\alpha_1^{(n)} \alpha_2^{(n)})}_{\text{Fourth Layer}}\\
     \widetilde{C}_2 &=(A+R_1(x+\alpha_1^{(n)}) - \alpha_2^{(n)}R_2 )(B+S_1(x+\alpha_1^{(n)})-\alpha_2^{(n)}S_2)\\&
     = \underbrace{\widetilde{C}_3}_{\text{First layer}} - \underbrace{\alpha_2^{(n)}(S_2(A+R_1x)+R_2(B+S_1x))}_{\text{Second Layer}}  + \underbrace{\alpha_1^{(n)}(S_1(A+R_1x)+R_1(B+S_1x))}_{\text{Third Layer}} + \underbrace{O(\alpha_1^{(n)} \alpha_2^{(n)})}_{\text{Fourth Layer}}\\
     \widetilde{C}_3 &= (A+R_1x)(B+S_1x)
 \end{align*}
In effect, at nodes $1$ and $2,$ the computation output can be interpreted as a superposition of at least $4$ layers. The first layer has magnitude $\Theta(1),$ the second $\Theta(\alpha_2^{(n)}),$ the third $\Theta(\alpha_1^{(n)}),$ and the remaining layers have magnitude $O(\alpha_1^{(n)} \alpha_2^{(n)}).$ The decoder can eliminate the effect of the second layer by computing:
$$ \overline{C} = \frac{\widetilde{C}_1+\widetilde{C}_2}{2} = (A+R_1(x+\alpha_1^{(n)}))(B+S_1(x+\alpha_1^{(n)})) + o(\alpha_1^{(n)})$$
Notice that the decoder also has access to $\tilde{C}_3 = (A+R_1x)(B+S_1x)$. From $\overline{C}$ and $\tilde{C}_3,$ the decoder can compute:

\vc{$$\overline{\overline{C}} = \frac{\overline{C}-\tilde{C}_3}{\alpha_1^{(n)}}  = (AS_1+BR_1) + 2R_1S_1x + O(\alpha_1^{(n)}) $$}
A simple calculation of the noise covariance matrix reveals that from $\overline{\overline{C}}$ and $C_3$, the decoder can compute $AB$ with $\texttt{SNR}_a \approx \frac{\eta^2}{x^4}+\frac{2 \eta }{x^2}$ as desired.

\noindent \underline{Formal Accuracy Analysis:} To show the theorem statement, it suffices to show that for any $\delta > 0,$ we can achieve $\texttt{SNR}_{a} > 2 \frac{\eta}{x^2}+\frac{\eta^2}{x^4}-\delta$ for a sufficiently large $n$.
We show this next by constructing a specific linear combination of the observations that achieves the desired signal-to-noise ratio. Observe that with our coding scheme, the nodes compute: 
$$ \Gamma_{t+1}\Theta_{t+1} =  (A+R_1 x)(B+S_1 x)$$
and, for $i=1,\ldots,t$:
\begin{align*} \Gamma_i \Theta_i &= (A+R_1(x+\alpha_1^{(n)}))(B+S_1(x+\alpha_1^{(n)})) + \alpha_2^{(n)}\bigg( (A+R_1(x+\alpha_1^{(n)}))\begin{bmatrix}S_2&\ldots & S_t\end{bmatrix} \\&+  (B+S_1(x+\alpha_1^{(n)}))\begin{bmatrix}R_2&\ldots & R_t\end{bmatrix}\bigg) \vec{g}_{i} + O((\alpha_2^{(n)})^2) \end{align*}   

Let $\gamma_1, \gamma_2,\ldots,\gamma_t$ be scalars, not all equal to zero, such that 
$ \sum_{i=1}^{t} \gamma_i \vec{g}_i = 0.$
Because $\vec{g}_{i}$ are $t-1$ dimensional vectors, they are linearly dependent, and such scalars indeed do exist.  Condition (C2) implies that $\sum_{i=1}^{t} \gamma_i \neq 0.$ With appropriate rescaling if necessary, we assume $\sum_{i=1}^{t} \gamma_i =1.$ The decoder computes:
$\tilde{\Gamma} \tilde{\Theta} \stackrel{\Delta}{=} \sum_{i=1}^{t} \gamma_i \Gamma_i \Theta_i,$ which is equal to:
$$\tilde{\Gamma}\tilde{\Theta} = (A+R_1(x+\alpha_1^{(n)}))(B+S_1(x+\alpha_1^{(n)})) + O((\alpha_2^{(n)})^2).$$

Then, the signal-to-noise ratio achieved is at least that obtained by using the signal and noise covariance matrices of \begin{eqnarray}\Gamma_{t+1}\Theta_{t+1} &=& AB+x(AS_1+BR_1) + R_1S_1 x^2 \label{eq:ach-decoder-processing1} \\ \tilde{\Gamma}\tilde{\Theta} &=&  AB + (x+\alpha_1^{(n)})(AS_1+BR_1) + (x+\alpha_1^{(n)})^2 R_1 S_1 + O((\alpha_2^{(n)})^2).\label{eq:ach-decoder-processing2}\end{eqnarray}The analysis is done in equations (\ref{eq:SNR1})-(\ref{eq:SNR3}) next:

\begin{align}\texttt{SNR}_a &\geq \frac{\begin{vmatrix}\eta^2+2 \eta x^2 + x^4  & \eta^2 + 2\eta x(x+\alpha_1^{(n)}) + x^2(x+\alpha_1^{(n)})^2 \\\eta^2 + 2\eta x(x+\alpha_1^{(n)}) + x^2(x+\alpha_1^{(n)})^2 & \eta^2 + 2\eta (x+\alpha_1^{(n)})^2 + (x+\alpha_1^{(n)})^4 + O((\alpha_2^{(n)})^4)\end{vmatrix}}{\begin{vmatrix}2\eta x^2 + x^4  & 2\eta x(x+\alpha_1^{(n)}) + x^2(x+\alpha_1^{(n)})^2\\ 2\eta x(x+\alpha_1^{(n)}) + x^2(x+\alpha_1^{(n)})^2 &  2\eta (x+\alpha_1^{(n)})^2 + (x+\alpha_1^{(n)})^4 + O((\alpha_2^{(n)})^4) \end{vmatrix}}-1\label{eq:SNR1}\\ 
&= \frac{4 \alpha_1^{(n)}x (\eta+2x^2) + 2(\eta+x^2)^2 + (\alpha_1^{(n)})^2(\eta+2x^2)+O((\alpha_2^{(n)})^4)}{2x^2(\alpha_1^{(n)}+x^2)+O((\alpha_2^{(n)})^4)}-1\label{eq:SNR3}
%&= \frac{(\alpha_1^{(n)})^2(2x^2+1)+4(\alpha_1^{(n)})(x+x^3)+2(x^2+1)^2 +  \frac{O((\alpha_2^{(n)})^4)}{(\alpha_1^{(n)})^2} }{2x^2\left((\alpha_1^{(n)})^2 + 2 (\alpha_1^{(n)})x +  x^2 +  \frac{O((\alpha_2^{(n)})^4)}{{(\alpha_1^{(n)})^2}}\right)}-1
\end{align}
As $n \to \infty$, observe that $\alpha_1^{(n)}, \frac{(\alpha_2^{(n)})^2}{\alpha_1^{(n)}} \rightarrow 0.$ Using this in (\ref{eq:SNR3}), for any $\delta > 0$, there exists a sufficiently large $n$ to ensure that
$\texttt{SNR}_a \geq  \frac{2(x^2+\eta)^2 } {2x^4}-1 -\delta = \frac{\eta^2}{x^4}+\frac{2\eta }{x^2} - \delta.$ This completes the proof.

%Coupled with $\Gamma_{t+1}\Theta_{t+1}$, this case reverts to the case of $N=2,t=1$ - up to $O(\alpha_2^{(n)})^2)$ term which vanishes relative to the signal that is at least $\Omega(\alpha_1^{(n)})$ - from which we obtain, $SNR \approx \frac{1}{x^4}+\frac{2}{x^2}.$

\subsection{Proof of Corollary \ref{cor:achievability}}
\label{app:DP}

Our proof involves a specific realization of random variables $R_1,R_2,\ldots,R_t,S_1,S_2,\ldots,S_t$ that satisfies the conditions of the achievable scheme in Section \ref{sec:coding_scheme}. We couple this with a refined differential privacy analysis. The accuracy analysis remains the same as in the proof of Section \ref{sec:accuracy}. Here, in our description, we focus on describing $R_1,R_2,\ldots,R_t$ and showing that the $t$-node DP privacy constraints are satisfied for input $A$. A symmetric argument applies for input $B$ also.

For a fixed DP parameter $\epsilon$, let $x= {\sigma^*(\epsilon)}+\delta',$ for $\delta' > 0$.  \vc{Section \ref{sec:accuracy}} shows that the scheme achieves accuracy $\texttt{SNR}_a \approx (1+\frac{\eta}{x^2})^2$, or more precisely:
$$ \LMSE(\mathcal{C}) \leq \frac{\eta^2 (\sigma^*(\epsilon))^4}{(\eta+(\sigma^*(\epsilon))^2)^2} + \delta$$
by choosing $\delta'$ sufficiently small. Now, it remains to show that a specific realization of $R_1,R_2,\ldots,R_t$ achieves $\epsilon$-DP. For a fixed value of parameter $x$, let $\vc{\epsilon^{**}}$ be:
\begin{equation}\vc{\epsilon^{**}} = \inf_{Z, E[Z^2]\geq 1} \sup_{\mathcal{A},A_0,A_1 \in \mathbb{R},|{A}_{0}-A_1| \leq 1} \ln \left(\frac{\mathbb{P}(A_0+x Z) \in \mathcal{A}}{\mathbb{P}(A_1+xZ) \in \mathcal{A}}\right)\label{eq:epsilonstarstar}\end{equation}
where the infimum is over all real-valued random variables $Z$ satisfying the variance\footnote{Note that choosing $E[Z]=0$ does not change the value of $\epsilon^{**},$ so we simply assume $E[Z]=0$ here as well.} constraint, and  $\ln$ denotes the natural logarithm. Notice here that the noise variance $E[(x^2Z^2)] = x^2$ is strictly larger than $(\sigma^{*}(\epsilon))^2$. Because $\sigma^{*}$ is a strictly monotonically decreasing function (see (\ref{eq:optimalsigma})), we have: $\epsilon^{**} < \epsilon$. \vc{Let $\epsilon^{*}$ be a real number such that $\epsilon^{**} < \epsilon^{*} < \epsilon.$ Let $Z^*$ be a random variable such that $\mathbb{E}[(Z^*)^2]<1$ and  where
$$\sup_{\mathcal{A},A_0,A_1 \in \mathbb{R},|{A}_{0}-A_1| \leq 1} \ln \left(\frac{\mathbb{P}(A_0+x Z^{*}) \in \mathcal{A}}{\mathbb{P}(A_1+xZ^{*}) \in \mathcal{A}}\right) \leq \epsilon^{*}.$$
That is, $A+xZ^{*}$ is $\epsilon^{*}$-DP. Note that equation (\ref{eq:epsilonstarstar}) implies that  existence of random variable $Z^{*}$ satisfying the above relation.}
 We now consider the achievable scheme of Section \ref{sec:coding_scheme} with $R_{1}$ taking the same distribution as $Z^*.$ We let $R_2,R_3,\ldots,R_t$ to be independent unit variance Laplace random variables that are independent of $R_1$. Note that by construction, 
\begin{equation}
    \sup_{\mathcal{A} \in \mathbb{R},-1 \leq \lambda \leq 1} \frac{\mathbb{P}\left(A+xR_1  \in \mathcal{A} \right)}{\mathbb{P}\left(A+xR_1+\lambda\in \mathcal{A} \right)} \leq e^{\epsilon^{*}} \label{eq:ach_singlenode_DP}
\end{equation}

Here, we are considering an adversary that is aiming to learn $A$ from $\tilde{A}_{i}, i \in \mathcal{S}$ which for every $t$-sized  subset $\mathcal{S}$. Let $\overline{\epsilon}_{n}$ be the parameter such that the coding scheme specified achieves $t$-node $\overline{\epsilon}_n$-DP. We show that as $n \to \infty$, $\overline{\epsilon}_n \to \epsilon^{*},$ thus showing that for sufficiently large $n$, our scheme achieves $t$-node $\epsilon$-DP . 

First, consider the case where $t+1 \notin \mathcal{S}$. For each $i \in \mathcal{S},$ the input $\tilde{A}_{i}$ is of the form $A+R_1(x+\alpha_1^{(n)})R_1 + Z_i,$ where $Z_i$ is zero mean random variable that is statistically independent of $R_1.$ Therefore, $$ A \rightarrow A+(x+\alpha_1^{(n)}) R_1 \rightarrow \{\tilde{A}_{i}: {i\in \mathcal{S}}\}$$
forms a Markov chain. Using post-processing property and the notation of (\ref{eq:DPdefY}),(\ref{eq:DPdefZ}), we note that 
\begin{eqnarray*}&& \sup_{\mathcal{A} \in \mathbb{R}^{t}} \frac{\mathbb{P}\left(\mathbf{Y}^{(0)}_{\mathcal{S}} \in \mathcal{A} \right)}{\mathbb{P}\left(\mathbf{Y}^{(1)}_{\mathcal{S}} \in \mathcal{A} \right)} \\
&\leq &\sup_{\mathcal{A} \in \mathbb{R},-1 \leq \lambda \leq 1} \frac{\mathbb{P}\left(A+(x+\alpha_1^{(n)})R_1  \in \mathcal{A} \right)}{\mathbb{P}\left(A+(x+\alpha_1^{(n)})R_1+\lambda\in \mathcal{A} \right)}\\ %&=& \sup_{\mathcal{A} \in \mathbb{R},-1 \leq \lambda \leq 1} \frac{\mathbb{P}\left(A+xR_1  \in \mathcal{A} \right)}{\mathbb{P}\left(A+xR_1+\frac{x}{x+\alpha_1^{(n)}}\lambda\in \mathcal{A} \right)}\\
&=& \sup_{\mathcal{A} \in \mathbb{R},-\frac{x}{x+\alpha_1^{(n)}} \leq \lambda \leq \frac{x}{x+\alpha_1^{(n)}}} \frac{\mathbb{P}\left(A+xR_1  \in \mathcal{A} \right)}{\mathbb{P}\left(A+xR_1+\lambda\in \mathcal{A} \right)}\\
& \leq & e^{\epsilon^{*}}
\end{eqnarray*}
where, in the final inequality, we have used (\ref{eq:ach_singlenode_DP}) combined with the fact that: $0 < \frac{x}{x+\alpha_1^{(n)}} \leq 1$ Thus, the input to the adversary follows $t$ node $\epsilon^*$-DP.

Now, consider the case where $t+1 \in \mathcal{S}$.
To keep the notation simple, without loss of generality, we assume that $\mathcal{S}=\{2,3,\ldots,t+1\}.$ We denote $$ \mathbf{G}_{2:t} = \begin{bmatrix}\vec{g}_{2} & \vec{g}_{3} & \ldots & \vec{g}_{t} \end{bmatrix}$$ In this case, an adversary obtains \begin{align*}\vec{Z} &=& \left(A+R_1x, A+(x+\alpha_1^{(n)}) R_1 \right.\left.~+ \alpha_2^{(n)}\begin{bmatrix}R_2&R_3&\ldots&R_t\end{bmatrix}\mathbf{G}_{2:t}\right).\end{align*} Using the fact that $\mathbf{G}_{2:t}$ is invertible based on the property (C1) in Section \ref{sec:coding_scheme}, a one-to-one function on the adversary's input yields:
$$ \vec{Z'} = \left(A+R_1x, -\frac{\alpha_1^{(n)}}{\alpha_{2}^{(n)} x} A \vec{1} (\mathbf{G}_{2:t})^{-1}+ \begin{bmatrix}R_2&R_3&\ldots&R_t\end{bmatrix}\right),$$
where $\vec{1}$ is a $t-1 \times 1$ row vector.
Denoting $\vec{Z'} = (Z'_1, Z'_2,\ldots,Z'_t)$, observe that $Z_i' = \lambda_i A + \zeta_i R_i$ for some $\lambda_i,\zeta_i \in \mathbb{R}$.  By construction $Z'_1$ achieves $\epsilon^*$-DP. For $i \geq 2$, because $R_i$ is a unit variance Laplacian RV, $Z_i'$ is a privacy mechanism that achieves $\frac{\alpha_i}{\beta_i} \sqrt{2}$-DP with respect to input $A$. Because $R_1,R_2,\ldots,R_t$ are independent, $\vec{Z'}$ is an RV that achieves 
$\epsilon^{*} + \sqrt{2} \sum_{i=2}^{t}\frac{\alpha_i}{\beta_i}$-DP. Note that for $i \geq 2,$ $\lambda_i' = 1$ and $\zeta_i = \frac{\alpha_1^{(n)}g_i'}{\alpha_{2}^{(n)} x},$ where, $g_i'$ is the $i$-th element of $\vec{1} (\mathbf{G}_{2:t})^{-1})$. So, the adversary's input satisfies $t$-node $\epsilon^{*} + \sqrt{2} \sum_{i=2}^{t}\frac{\alpha_1^{(n)}g_i'}{\alpha_{2}^{(n)} x}$-DP. The proof is complete on noting that the DP parameter approaches $\epsilon^{*}$ as $n \to \infty$.

\section{Proofs of Theorem \ref{thm:main_converse} and Corollary \ref{cor:converse}}
\label{app:converse}
Recall that we consider a set up with $N$ computation nodes such that the input is private to any $t$ nodes, where $N \leq 2t.$ Consider any $N$-node secure multiplication coding scheme. \vc{ In the coding scheme, node $i$ receives inputs $\tilde{A}_{i} = a_iA + \tilde{R}_i, \tilde{B}_i = b_i B + \tilde{S}_i.$} Suppose that, for the achievable scheme, the $t$-node privacy signal-to-noise ratio $0 < \texttt{SNR}_p < \infty$ and accuracy signal-to-noise ratio $\texttt{SNR}_a.$ 
There exist uncorrelated, zero-mean, unit-variance random variables $\overline{R}_{1},\overline{R}_{2},\ldots,\overline{R}_{N}, \overline{S}_{1},\overline{S}_{2},\ldots,\overline{S}_{N}$ such that $A,B,\overline{R}_i\big|_{i=1}^{t}, S_{i}\big|_{i=1}^{t}$ are zero mean unit variance uncorrelated random variables, and the inputs to node $i$ are:
$$\Gamma_i = \left[\frac{A}{\sqrt{\eta}}~\overline{R}_1~\overline{R}_2~\ldots~\overline{R}_N\right] \vec{v}_{i}$$
$$\Theta_i = \left[\frac{B}{\sqrt{\eta}}~\overline{S}_1~\overline{S}_2~\ldots~\overline{S}_N\right] \vec{w}_{i}$$

where $\vec{v}_{i},\vec{w}_{i}$ are $N \times 1$ vectors\footnote{To see this, simply set $\begin{bmatrix} \overline{R}_{1}&\overline{R}_{2} &\ldots & \overline{R}_{N}\end{bmatrix} = \begin{bmatrix} \tilde{R}_{1}&\tilde{R}_{2} &\ldots & \tilde{R}_{N}\end{bmatrix}{\mathbf{K}^{-1/2}}$ where $\mathbf{K}$ is the $N \times N$ covariance matrix of $\begin{bmatrix} \tilde{R}_{1}&\tilde{R}_{2} &\ldots & \tilde{R}_{N}\end{bmatrix}$. $\overline{S}_{i}|_{i=1}^{N}$ can be found similarly.}.  Node $i$ performs the computation  
$$\tilde{C}_i = \Gamma_i \Theta_i,$$
and a decoder connects to the $N$ nodes and obtains:
$$\vc{d_i\widetilde{C} = \sum_{i=1}^{N}  \tilde{C}_i}$$
%where $Z$ consists of terms that do not include $AB$ - that is, $Z$ is a linear combination of $A\overline{S}_{i}|_{i=1}^{N}, B\overline{R}_i|_{i=1}^t, \overline{R}_i\overline{S}_j|_{i,j=1}^{t}$. Notably, we have that $\text{SNR}_a = \frac{\eta^2}{E[Z^2]}.$ We can write:
The error $\widetilde{C}-AB$ of the decoder can be written as:
$$\widetilde{C}-AB= \left[\frac{A}{\sqrt{\eta}}~\overline{R}_1~\overline{R}_2~\ldots~\overline{R}_N\right]\Delta \begin{bmatrix}\frac{B}{\sqrt{\eta}} \\ \overline{S}_1 \\ \overline{S}_2 \\ \vdots \\ \overline{S}_N\end{bmatrix} $$
where
\begin{equation} \Delta = \sum_{i=1}^{N} d_{i} \vec{v}_{i} \vec{w}_{i}^{T}  - \begin{bmatrix}\eta & 0 & \ldots & 0 \\ 0 & 0  & \ldots & 0 \\ & & \ddots & \\ 0 & 0 & \ldots & 0\end{bmatrix}\label{eq:Delta}\end{equation}

\vc{ Because of Lemma \ref{lem:LMSE_SNR}}, observe that, for the optimal choice of $d_1,d_2,\ldots,d_N$, \begin{equation}\mathbb{E}\left[|| \widetilde{C}-AB ||^2 \right] = || \Delta||_{F}^{2} = \frac{\eta^2}{1+\texttt{SNR}_a}.\label{eq:DeltaMMSE} \end{equation} We aim to lower bound $|| \Delta||_{F}^{2}.$  Our converse is a natural consequence of the following theorem.

\begin{theorem}
For any $N$ node secure coding scheme with $N \leq 2t$, for any set $\mathcal{S} \subset \{1,2,\ldots,N\}$ where $|\mathcal{S}|=t,$  we have:
 $$(1+\texttt{SNR}_a) \leq (1+\texttt{SNR}^{(A)}_{\mathcal{S}})(1+\texttt{SNR}^{(B)}_{\mathcal{S}^{c}}) $$
 By symmetry, we also have:
 $$(1+\texttt{SNR}_a) \leq (1+\texttt{SNR}^{(B)}_{\mathcal{S}})(1+\texttt{SNR}^{(A)}_{\mathcal{S}^{c}}) $$
 \label{thm:converse}
\end{theorem}

For any coding scheme that that satisfies $\texttt{SNR}^{(A)}_{\mathcal{S}},\texttt{SNR}^{(B)}_{\mathcal{S}} \leq \texttt{SNR}_{p}$ for every subset $\mathcal{S}$ of $t$ nodes, the above theorem automatically implies the statement of Theorem \ref{thm:main_converse}, that is:
$$(1+\texttt{SNR}_a) \leq (1+\texttt{SNR}_p)^2 $$
The proof of Theorem \ref{thm:converse} depends on the following key lemma.

\begin{lemma}
For any set $\mathcal{S}$ of nodes with $|\mathcal{S}| \leq t,$ there exists a vector $$\vec{\lambda} = \left[\lambda_1~~\lambda_2~~\ldots~~\lambda_{N}\right]^T$$
such that 
\begin{enumerate}[(i)]
    \item $$ \vec{\lambda}^{T}\vec{w}_{i} = 0, \forall i \in \mathcal{S},$$ 
    \item $$ \frac{\lambda_1^2}{||\vec{\lambda}||^2} \geq \frac{1}{1+ \texttt{SNR}^{(B)}_{\mathcal{S}}}.$$
\end{enumerate}

Symmetrically, there exists a vector $$\vec{\theta}= \left[\theta_1~~\theta_2~~\ldots~~\theta_{t+1}\right]^T$$
for every subset $\mathcal{S}$ of nodes with $ |\mathcal{S}| \leq t$ such that 
\begin{enumerate}[(i)]
    \item[(iii)] $$ \vec{\theta}^{T} \vec{v}_{i} = 0, \forall i \in \mathcal{S},$$ 
    \item [(iv)]$$ \frac{\theta_1^2}{||\vec{\theta}||^2} \geq \frac{1}{1+ \texttt{SNR}^{(A)}_{\mathcal{S}}}.$$ 
\end{enumerate}
\label{lem:nullvector_SNR}
\end{lemma}

\begin{proof}
Consider any vector $\vec{\underline{w}}=[\underline{w}_1~~\underline{w}_2~~\ldots~~\underline{w}_{N}]^{T}$ that in the span of $\{\vec{w}_{i}: i \in \mathcal{S}\}$. \vc{Let $\mathcal{S} = \{s_1, s_2, \ldots, s_{t}\}$ where $s_1 < s_2< \ldots < s_t,$ and suppose that $\vec{\underline{w}}=\sum_{i=1}^{t} \beta_i \vec{w}_{s_i}.$ Because of Lemma \ref{lem:theSNRlemma} and Definition \ref{def:SNRp}, we know that 
$$ \mathbb{E}\left[\left|\gamma(\beta_1 \Theta_{s_1} + \beta_2 \Theta_{s_2} + \ldots +\beta_t \Theta_{s_t}) - \frac{B}{\sqrt{\eta}}\right|^2\right] \geq \frac{1}{1+\texttt{SNR}_{\mathcal{S}}^{(B)}} $$
for any constant $\gamma$. 
Observing that $\Theta_i = \left[\frac{B}{\sqrt{\eta}}~\overline{S}_1~\overline{S}_2~\ldots~\overline{S}_N\right] \vec{w}_{i}$, we can re-write the above equation as:
$$\mathbb{E}\left[\left|\left[\frac{B}{\sqrt{\eta}}~\overline{S}_1~\overline{S}_2~\ldots~\overline{S}_N\right] \gamma \vec{w}- B\right|^2\right] \geq \frac{1}{1+\texttt{SNR}^{(B)}_{\mathcal{S}}}.$$ 

Because $\frac{B}{\eta}, \overline{S}_1,\ldots, \overline{S}_{N}$ are unit norm independent random variables with $\mathbb{E}[S_1]=\mathbb{E}[S_2]=\ldots = \mathbb{E}[S_N]=0,$ the above equation implies
$$(\gamma w_1-1)^2 + \sum_{i=2}^{N} \gamma^2 w_i^2 \geq \frac{1}{1+\texttt{SNR}_{\mathcal{S}}^{(B)}}.$$
Setting $\gamma = \frac{w_1}{w_1^2+w_2^2 + \ldots + w_N^2}$ and re-arranging the terms of the above equation, we have:}
\begin{equation}\frac{{\underline{w}}_{1}^2}{\underline{w}_2^2+\underline{w}_3^2+\ldots+\underline{w}_{N}^2} \leq  \texttt{SNR}^{(B)}_{\mathcal{S}}\label{eq:snrbound} \end{equation}

Because $|\mathcal{S}| \leq t,$ the nullspace of $\{\vec{w}_{i}: i \in \mathcal{S}\}$ is non-trivial. If $\begin{bmatrix}1 \\ 0 \\ 0 \\ \vdots \\ 0 \end{bmatrix}$ lies in the span of  $\{\vec{w}_{i}: i \in \mathcal{S}\}$, then $\texttt{SNR}^{(B)}_{\mathcal{S}} = \infty,$ and any non-zero vector $\vec{\lambda}$ in the null space of  $\{\vec{w}_{i}: i \in \mathcal{S}\}$ satisfies $(i)$ and $(ii)$. So, it suffices to show the existence of the $\vec{\lambda}$ that satisfies the theorem for the case where $\begin{bmatrix}1 \\ 0 \\ 0 \\ \vdots \\ 0 \end{bmatrix}$ does not lie in the span of $\{\vec{w}_{i}: i \in \mathcal{S}\}$.

By the rank-nullity theorem, there exists a vector $\vec{\underline{w}}=[\underline{w}_1~~\underline{w}_2~~\ldots~~\underline{w}_{N}]^{T}$ in the span of $\{\vec{w}_{i}: i \in \mathcal{S}\},$ and a non-zero vector $\vec{\lambda} = \left[\lambda_1~~\lambda_2~~\ldots~~\lambda_{N}\right]^{T}$ that is in the nullspace of $\{\vec{w}_{i}: i \in \mathcal{S}\},$ such that 
\begin{equation} \vec{\underline{w}}+\vec{\lambda} = \begin{bmatrix}1 \\ 0 \\ 0 \\ \vdots \\ 0 \end{bmatrix}\label{eq:1} \end{equation}
Specifically, note that $\lambda_i = -\underline{w}_i,$ for $i=2,3,\ldots, t+1.$ Because $\vec{\lambda}$ nulls $\vec{\underline{w}}$, we have:
\begin{equation*}\lambda_1 \underline{w}_1 =  - \sum_{i=2}^{N} \lambda_{i}\underline{w}_{i} \end{equation*} Consequently, we have:  \begin{equation}\lambda_1 \underline{w}_1 = \sum_{i=2}^{N} \underline{w}_i^2 = \sum_{i=2}^{N} \lambda_i^2 \label{eq:lambdawproduct}\end{equation}

Therefore: 
\begin{align}
\frac{||\vec{\lambda}||^2}{\lambda_1^2} & = 1+ \frac{\sum_{i=2}^{N} \lambda_i^2}{\lambda_1^2}\\ &=
1+ \frac{\sum_{i=2}^{N} \underline{w}_i^2}{\lambda_1^2} \label{eq:sub1} \\
&=1+ \frac{\underline{w}_1^2}{\sum_{i=2}^{N} \underline{w}_i^2} \label{eq:sub2}\\
&\leq 1+ \texttt{SNR}^{(B)}_{\mathcal{S}} \label{eq:finalstepconverselemma}
\end{align}
where in (\ref{eq:sub1}) and (\ref{eq:sub2}), we have used (\ref{eq:lambdawproduct}), \vc{and in \eqref{eq:finalstepconverselemma}, we have used \eqref{eq:snrbound}.}

Therefore: 
 $$ \frac{\lambda_1^2}{||\vec{\lambda}||^2} \geq \frac{1}{1+ \texttt{SNR}^{(B)}_{\mathcal{S}}}$$
 as required. The existence of a vector $$\vec{\theta}^{(A)}_{\mathcal{S}} = \left[\theta_1~~\theta_2~~\ldots~~\theta_{t+1}\right]$$ that satisfies (iii),(iv) in the lemma statement follows from symmetry.
\end{proof}

\begin{proof}[Proof of Theorem \ref{thm:converse}]
Consider a set $\mathcal{S}$ of $t$ nodes. Let $$\vec{\lambda} = \left[\lambda_1~~\lambda_2~~\ldots~~\lambda_{N}\right]^{T}$$ be a $N \times 1$ vector that is orthogonal to \vc{all vectors in the set $\{\vec{w}_{i}, i \in \mathcal{S}^{c}\}$} such that:
\begin{equation} \frac{\lambda_1^{2}}{||\vec{\lambda}||^2} \geq  \frac{1}{1+\texttt{SNR}^{(B)}_{\mathcal{S}^{c}}}\label{eq:lambda1} \end{equation}

Because $|\mathcal{S}|=t$ and $N \leq 2t,$ it transpires that $|\mathcal{S}^{c}| \leq t,$ and from Lemma \ref{lem:nullvector_SNR}, we know that a vector $\vec{\lambda}$ satisfying the above conditions exist.
We then have:
\begin{equation}{\Delta}\vec{\lambda}  = \sum_{i\in \mathcal{S}} c_i  \vec{v}_{i}   - \begin{bmatrix}\lambda_1 \eta\\ 0 \\ \vdots \\ 0\end{bmatrix}\label{eq:delta1}\end{equation}
where $c_i = d_{i} \vec{w}_{i}^{T} \vec{\lambda}$ and $\Delta$ is as in \eqref{eq:Delta}. Now, \vc{ from Lemma \ref{lem:theSNRlemma},} we know that:
$$\inf_{\overline{c}_{i}}\left|\left| \sum_{i\in \mathcal{S}} \overline{c}_i  \vec{v}_{i}   - \begin{bmatrix}\sqrt{\eta} \\ 0 \\ 0 \\ \vdots \\ 0\end{bmatrix} \right|\right|^2 = \frac{\eta}{1+\texttt{SNR}^{(A)}_{\mathcal{S}}} $$ 
\vc{ By multiplying the above equation by $\lambda_1^2$}, we have:
$$\left|\left|\sum_{i\in \mathcal{S}} {c}_i  \vec{v}_{i}   - \begin{bmatrix}\lambda_1 \eta \\ 0 \\ 0 \vdots \\ 0\end{bmatrix} \right|\right|^2 \geq \frac{\lambda_1^{2}\eta^2}{1+\texttt{SNR}^{(A)}_{\mathcal{S}}}$$
Taking norms in (\ref{eq:delta1}) and applying the above, we get:
\begin{equation}||{\Delta}\vec{\lambda}||^2  \geq \frac{\lambda_1^{2} \eta^2}{1+\texttt{SNR}^{(A)}_{\mathcal{S}}}
\end{equation}
By definition of the $\ell_2$ norm, we have 
$$ ||{\Delta}\vec{\lambda}||^2 \leq || \Delta ||_2^2 ||\vec{\lambda}||^2$$
Because, for any matrix, its Frobenius norm is lower bounded by its $\ell_2$ norm, we have: 

\begin{equation}||{\Delta}||_F^2  \geq \frac{\lambda_1^{2}}{||\vec{\lambda}||^2} \frac{\eta^2}{1+\texttt{SNR}^{(A)}_{\mathcal{S}}}
\end{equation}

From (\ref{eq:lambda1}), we have:
\begin{equation} \label{eq:delta_lower}
||{\Delta}||_F^2  \geq \eta^2
\frac{1}{1+\texttt{SNR}^{(B)}_{\mathcal{S}^{c}}}
\frac{1}{1+\texttt{SNR}^{(A)}_{\mathcal{S}}},
\end{equation} \vc{
which, in combination with \eqref{eq:DeltaMMSE} and the fact that $\texttt{SNR}_p \geq \max(\texttt{SNR}^{(A)}_{\mathcal{S}},\texttt{SNR}^{(B)}_{\mathcal{S}^c})$  implies that 
$$ (1+\texttt{SNR}_{a}) \leq (1+\texttt{SNR}_p)^2 $$
 }
\end{proof}

\subsection{Proof of Corollary \ref{cor:converse}}
Consider an achievable coding scheme $\mathcal{C}$ that achieves $t$-node $\epsilon$-DP. From Lemma \ref{lem:LMSE_SNR}, we know that:
\begin{equation}
    \LMSE(\mathcal{C}) = \frac{\eta^2}{1+\texttt{SNR}_a} \label{eq:converse_lmse}
\end{equation} From Theorem \ref{thm:main_converse}, we know that there exists a set $\mathcal{S} \subset\{1,2,\ldots,N\}$ such that (i) $\texttt{SNR}_{\mathcal{S}}^{(A)} \geq \sqrt{1+\texttt{SNR}_a}-1,$ or(ii) $\texttt{SNR}_{\mathcal{S}}^{(B.)} \geq \sqrt{1+\texttt{SNR}_a}-1$. Without loss of generality, assume that (i) holds for the coding scheme $\mathcal{C}.$
Consequently, there exist scalars $w_i, i \in \mathcal{S}$ such that:
$$ \sum_{i \in \mathcal{S}} w_i \tilde{A}_{i} = A+Z $$
where $Z$ is uncorrelated with $A$ and $\mathbb{E}[Z]^2 \leq \frac{\eta}{\sqrt{1+\texttt{SNR}_{a}}-1}.$ By definition of the function $\sigma^*(\epsilon),$ we have:
\begin{equation} (\sigma^*(\epsilon))^2 \leq E[Z]^2 \leq \frac{\eta}{\sqrt{1+\texttt{SNR}_{a}}-1}.\label{eq:converse_DP_SNR}\end{equation}
Combining (\ref{eq:converse_lmse}) and (\ref{eq:converse_DP_SNR}), we get the desired result.

\section{Extension to Matrix Multiplication}
\label{sec:matrix}
We consider the problem of computing the \emph{matrix} product $\mathbf{A}\mathbf{B}$, where $\mathbf{A} \in \mathbb{R}^{M\times L}$ and show how our scalar case results extend to matrix multiplication application. Our main result is an equivalence between codes for scalar multiplication and matrix multiplication under certain assumptions on the matrix multiplication code.
\textbf{Notation:} In the sequel, for a matrix $\mathbf{M}$, we denote the entry in its $i$-th row and $j$-th column by $\mathbf{M}[i,j].$

Let $\mathcal{C}$ be an arbitrary $N$-node secure multiplication coding scheme that achieves $t$-node $\epsilon$-DP and the accuracy $\texttt{SNR}_a$ for computing a scalar product $AB$.  We define $\mathcal{C}_\text{matrix}$ as a matrix extension of $\mathcal{C}$ that applies the coding scheme $\mathcal{C}$ \emph{independently} to each entry of the matrices $\mathbf{A} \in \mathbb{R}^{M \times L}$ and $\mathbf{B} \in \mathbb{R}^{L \times K}$. Specifically, node $i$ in $\mathcal{C}_{\text{matrix}}$ receives:
$$\tilde{\mathbf{A}}_i = a_i 
\mathbf{A} + \mathbf{\tilde{R}}_{i}$$
$$\tilde{\mathbf{B}}_i = b_i 
\mathbf{B} + \mathbf{\tilde{S}}_{i}$$
where the entries of the \vc{$M \times L$} random matrix $\mathbf{R}_{i}$ and the entries of the \vc{$L \times K$} random matrix $\mathbf{S}_{i}$ are chosen in an i.i.d. manner from the same distribution specified by $\mathcal{C},$ and the constant \vc{scalars} $a_i,b_i$ are also specified in $\mathcal{C}$. We evaluate the performance of $\mathcal{C}_{\text{matrix}}$ using worst-case metrics for both privacy and accuracy as follows: 
\begin{definition}(Matrix $t$-node $\epsilon$-DP)
\label{def:edp_matrix}
    Let $\epsilon\geq 0$. A coding scheme with random noise variables $$(\tilde{\mathbf{R}}_{1},\tilde{\mathbf{R}}_{2},\ldots, \tilde{\mathbf{R}}_{N}), (\tilde{\mathbf{S}}_{1}, \tilde{\mathbf{S}}_{2},\ldots, \tilde{\mathbf{S}}_{N})$$
    where $\tilde{\mathbf{R}}_{i} \in\mathbb{R}^{M \times L}, \tilde{\mathbf{S}}_{i} \in\mathbb{R}^{L \times K}$
    and scalars $a_i,b_i \; (i \in \{1,\ldots,N\})$  satisfies matrix $t$-node $\epsilon$-DP if, for any $\mathbf{A}_{0},\mathbf{A}_{1} \in \mathbb{R}^{M \times L} , \mathbf{B}_{0}, \mathbf{B}_{1} \in \mathbb{R}^{L \times K}$ that satisfy $\left|\left|\begin{bmatrix}\mathbf{A}_{0} \\ \mathbf{B}_{0}^T \end{bmatrix} - \begin{bmatrix}\mathbf{A}_{1} \\ \mathbf{B}_{1}^T \end{bmatrix}\right|\right|_\text{max} \leq 1$,
\begin{eqnarray}
\max\left(
\max_{\substack{m= 1, \ldots, M \\ l=1,\ldots, L}} \left(\frac{\mathbb{P}\left(\mathbf{Y}^{(0)}_{\mathcal{T}}[m, l] \in \mathcal{A} \right)}{\mathbb{P}\left(\mathbf{Y}^{(1)}_{\mathcal{T}}[m,l] \in \mathcal{A} \right)} \right),
\max_{\substack{l= 1, \ldots, L\\ k=1,\ldots, K}}
\left(
\frac{\mathbb{P}\left(\mathbf{Z}^{(0)}_{\mathcal{T}}[l,k] \in \mathcal{A} \right)}{\mathbb{P}\left(\mathbf{Z}^{(1)}_{\mathcal{T}}[l,k] \in \mathcal{A} \right)}\right) 
\right) &\leq& e^\epsilon
  \label{eq:DPmatrix}
\end{eqnarray}
for all subsets $\mathcal{T} \subseteq \{1,2,\ldots,N\},|\mathcal{T}|=t$,
for all subsets $\mathcal{A} \subset \mathbb{R}^{1 \times t}$ in the Borel $\sigma$-field, where, for $\ell=0,1$,
\begin{eqnarray}
\mathbf{Y}_{\mathcal{T}}^{(\ell)}[m,l] &\triangleq&  \begin{bmatrix}a_{i_1} \mathbf{A}_{\ell}[m,l]+\mathbf{\tilde{R}}_{i_1}[m,l], & \ldots, & a_{i_{|\mathcal{T}|}}\mathbf{A}_{\ell}[m,l]+\mathbf{\tilde{R}}_{i_{|\mathcal{T}|}}[m,l] \end{bmatrix}\label{eq:DPdefYmatrix} \\
\mathbf{Z}_{\mathcal{T}}^{(\ell)}[l,k] &\triangleq&  
\begin{bmatrix}  b_{i_1}\mathbf{B}_{\ell}[l,k]+\mathbf{\tilde{S}}_{i_1}[l,k], & \ldots, & b_{i_{|\mathcal{T}|}}\mathbf{B}_{\ell}[l,k] + \mathbf{\tilde{S}}_{i_{|\mathcal{T}|}}[l,k] \end{bmatrix} \label{eq:DPdefZmatrix}, 
\end{eqnarray}
where $\mathcal{T}=\{i_1,i_2,\ldots,i_{|\mathcal{T}|}\}.$ 
\end{definition}

\begin{definition} [Matrix LMSE.] For a matrix coding scheme $\mathcal{C}_\text{matrix}$, we define the LMSE as follows:
\begin{equation}
    \textsf{LMSE}(\mathcal{C}_\text{matrix}) = \max_{\substack{m= 1, \ldots, M\\ k=1,\ldots, K}} \mathbb{E} [ |(\mathbf{A}\mathbf{B})[m,k] - \tilde{\mathbf{C}}[m,k] |^2 ], 
\end{equation}
where $\tilde{\mathbf{C}}$ is a decoded matrix using an affine decoding scheme $d$.
\end{definition}

Analogous to Assumption \ref{assumption1} in the scalar case, our accuracy analysis is contingent on the data $\mathbf{A},\mathbf{B}$ satisfying the following assumption.

\begin{assumption}
%\begin{enumerate}
$\mathbf{A}$ and $\mathbf{B}$ are statistically independent matrices of dimensions $M \times L$ and $L \times K$, and they satisfy:
\begin{enumerate}[(a)]
\item \begin{equation}
    \mathbb{E}\left[   || \mathbf{A} ||_\text{max}^2 \right] \leq \eta, \; \mathbb{E}\left[ || \mathbf{B} ||_\text{max}^2 \right] {\leq \eta},
    \label{eq:assumption-matrix1}
\end{equation}
for a parameter $\eta > 0$, and 
\item \begin{equation}
    \mathbb{E}[\mathbf{A}[m,i] \mathbf{A}[m,j]] \mathbb{E}[\mathbf{B}[i,k] \mathbf{B}[j,k]] = 0,
\end{equation}
for all $1 \leq i \neq j \leq L$ and $m = 1, \cdots M, k = 1, \cdots K$.
\end{enumerate}
\label{assumption2}
\end{assumption} 

Assumption \ref{assumption2}-(b) for instance holds if the entries of $\mathbf{A}$ are uncorrelated, or if the entries of $\mathbf{B}$ are uncorrelated. Our main result states that $\mathcal{C}_{\textrm{matrix}}$ attains an identical privacy and accuracy as $\mathcal{C}$

% Under this assumption, we encode each entry in $\mathbf{A}, \mathbf{B}$ using the identical coding scheme. In other words, the node $i$ receives: 
% \begin{equation}
%     \tilde{A}_i[m,l] = A[m,l] + \tilde{R}_i[m,l], \;\; \tilde{B}_i[l,k] = B[l,k] + \tilde{S}_i[l,k] 
% \end{equation}
% where $\tilde{R}_i[m,l], \tilde{S}_i[l,k] \in \mathbb{R}$ are random variables that are chosen independently of the inputs $\mathbf{A}, \mathbf{B}$ and also 
% satisfy (i)  $(\tilde{R}_1[m,l], \ldots, \tilde{R}_N[m,l])$ are statistically independent from $(\tilde{R}_1[m',l'], \ldots, \tilde{R}_N[m',l'])$, (ii) $(\tilde{S}_1[l,k], \ldots, \tilde{S}_N[l,k])$ are statistically independent from $(\tilde{S}_1[l',k'], \ldots, \tilde{S}_N[l',k'])$, and (iii) there is no shared randomness between $(\tilde{R}_1[m,l], \ldots, \tilde{R}_N[m,l])$ and $(\tilde{S}_1[l,k], \ldots, \tilde{S}_N[l,k])$. For an achievable scheme, we will choose the noise sequences 
% $\tilde{R}_i[m,l]|_{i=1,\ldots N}, \tilde{S}_i[l,k]_{i=1,\ldots N}$ for each $(m,l)$ and $(l,k)$ following the construction given in Section~\ref{sec:coding_scheme}. Next, we extend the definition of privacy as follows:

\begin{theorem}
    Consider any (scalar) multiplication coding scheme $\mathcal{C}$. Let $\mathcal{C}_{\textrm{matrix}}$ denote its matrix extension. 
    Then, 
    \begin{enumerate}[(i)]
        \item $\mathcal{C}$ satisfies $t$-node $\epsilon$-DP if and only if $\mathcal{C}_{\textrm{matrix}}$  satisfies $t$-node matrix $\epsilon$-DP.
        \item If $\mathbf{A},\mathbf{B}$ satisfy Assumption \ref{assumption2}, then:
        $$\textrm{LMSE}(\mathcal{C}_{\textrm{matrix}})  \leq \sup\textrm{LMSE}(\mathcal{C}),$$
        where the supremum on the right hand side is over all data distributions $\mathbb{P}_{A,B}$ that satisfy Assumption \ref{assumption1} with parameter $\eta.$ Further, the bound above is met with equality if every entry of $\mathbf{A},\mathbf{B}$ has standard deviation $\eta.$ 
    \end{enumerate}
    \label{thm:matrix}
\end{theorem}

The above theorem establishes an equivalence between the trade-off for matrix multiplication and the trade-off for scalar multiplication. Specifically, using Corollaries \ref{cor:achievability} and \ref{cor:converse}, we infer that  for a fixed matrix DP parameter $\epsilon,$ the optimal LMSE for the matrix multiplication case is the same as for scalar multiplication, that is:
$$ \textrm{LMSE} \approx \frac{\eta^2 (\sigma^*(\epsilon))^4}{\eta+(\sigma^*(\epsilon))^2}^2.$$

The above equivalence must however be interpreted with some caveats. First, the equivalence assumes that the coding scheme for the matrix case extends the scalar strategy to each input matrix element in an \emph{independent} manner. The question of whether correlation in the noise distribution can reduce the LMSE for a fixed DP parameter is left open. Second, the above trade-off requires Assumption (\ref{assumption2})-(b). In some cases, this assumption may be justified - for example, if $\mathbf{B}$ has data samples drawn from some distribution in an i.i.d. manner. However, in some cases, this assumption of uncorrelated data may be too strong. The question of the optimal trade-off when this assumption is dropped remains open. 

\begin{proof}[Proof of Theorem \ref{thm:matrix}]
\underline{Proof of (i)}

An elementary proof readily from the definition of matrix $t$-node $\epsilon$-DP and the matrix extension of the coding scheme $\mathcal{C}$. Specifically, let $(\overline{R}_1,\overline{R}_2,\ldots, \overline{R}_N)$ and $(\overline{S}_1,\overline{S}_2,\ldots, \overline{S}_N)$ denote the noise random variables of coding scheme $\mathcal{C}$.  Let $(\mathbf{\tilde{R}}_1, \ldots, \mathbf{\tilde{R}}_N)$ and $(\mathbf{\tilde{S}}_1, \ldots, \mathbf{\tilde{S}}_N)$ denote the noise random variables of the coding scheme $\mathcal{C}_{\text{matrix}}.$ 

%Then, for any $(m, l, l', k)$ such that $1 \leq m \leq M, 1 \leq l,l' \leq L, 1 \leq k \leq K$, for 
%variables %$(\mathbf{\tilde{R}}_1[m,l], \ldots, \mathbf{\tilde{R}}_N[m,l])$ and $(\mathbf{\tilde{S}}_1[l,k], \ldots, \mathbf{\tilde{S}}_N[l,k])$ along with $a_1, \ldots, a_N$ and  $b_1, \ldots, b_N$
%denote the parameters of the matrix coding scheme $\mathcal{C}_{\textrm{matrix}}$.   satisfy the scalar-version $t$-node $\epsilon$-DP. 
To show the ``if'' statement, assume that $\mathcal{C}$ satisfies $t$-node $\epsilon$-DP. We show that $\mathcal{C}_{\textrm{matrix}}$ also satisfies $t$-node $\epsilon$-DP. Let $\mathbf{A}_{0},\mathbf{A}_{1},\mathbf{B}_0,\mathbf{B}_1$ denote matrices that satisfy $\left|\left|\begin{bmatrix}\mathbf{A}_{0} \\ \mathbf{B}_{0}^T \end{bmatrix} - \begin{bmatrix}\mathbf{A}_{1} \\ \mathbf{B}_{1}^T \end{bmatrix}\right|\right|_\text{max} \leq 1.$  
For an arbitrary subset $\mathcal{A}$ of the Borel sigma field, let $$m^{*}, l^{*} = \text{argmax}_{\substack{m= 1, \ldots, M \\ l=1,\ldots, L}} \left(\frac{\mathbb{P}\left(\mathbf{Y}^{(0)}_{\mathcal{T}}[m, l] \in \mathcal{A} \right)}{\mathbb{P}\left(\mathbf{Y}^{(1)}_{\mathcal{T}}[m,l] \in \mathcal{A} \right)} \right),$$
    $$ 
   l^{**}, k^{**} = \text{argmax}_{\substack{l= 1, \ldots, L \\ k=1,\ldots, K}} \left(\frac{\mathbb{P}\left(\mathbf{Z}^{(0)}_{\mathcal{T}}[l, k] \in \mathcal{A} \right)}{\mathbb{P}\left(\mathbf{Z}^{(1)}_{\mathcal{T}}[l,k] \in \mathcal{A} \right)} \right). 
   $$

   For any set $\mathcal{T} = \{i_1,i_2,\ldots,i_t\} \subset \{1,2,\ldots,N\}$:
\begin{align*}   &\max_{m=1,2,\ldots,M, l=1,2,\ldots,L}\frac{\mathbb{P}\left(\mathbf{Y}^{(0)}_{\mathcal{T}}[m^{*}, l^{*}] \in \mathcal{A} \right)}{\mathbb{P}\left(\mathbf{Y}^{(1)}_{\mathcal{T}}[m^{*}, l^{*}] \in \mathcal{A} \right)}\\
&\stackrel{(a)}=\frac{\mathbb{P}(\begin{bmatrix}a_{i_1}\mathbf{A}_0[m^*,l^*] + \overline{R}_{i_1} & \ldots &  a_{i_t}\mathbf{A}_0[m^*,l^*] + \overline{R}_{i_t}\end{bmatrix}}{\mathbb{P}(\begin{bmatrix}a_{i_1}\mathbf{A}_1[m^*,l^*] + \overline{R}_{i_1} & \ldots &  a_{i_t}\mathbf{A}_1[m^*,l^*] + \overline{R}_{i_t}\end{bmatrix})}\\&
\stackrel{(b)}\leq e^{\epsilon}.
\end{align*}
In $(a)$ above, we have used the fact that $(\overline{R}_{i_1}, \overline{R}_{i_2},\ldots, \overline{R}_{i_t})$ has the same distribution as \\ $(\mathbf{R}_{i_1}[m^*,l^*], \mathbf{R}_{i_2}[m^*,l^*],\ldots, \mathbf{R}_{i_t}[m^*,l^*])$. In (b) we have used the fact that $$|\mathbf{A}_0[m^*,l^*]-\mathbf{A}_1[m^*,l^*]|\leq \left|\left|\begin{bmatrix}\mathbf{A}_{0} \\ \mathbf{B}_{0}^T \end{bmatrix} - \begin{bmatrix}\mathbf{A}_{1} \\ \mathbf{B}_{1}^T \end{bmatrix}\right|\right|_\text{max}  1,$$ coupled with the fact that $\mathcal{C}$ satisfies $t$-node $\epsilon$-DP.
   
   %from our assumption that the random sequences for $(m^{*}, l^{*}, l^{**}, k^{**})$ satisfy $t$-node $\epsilon$-DP, we can conclude that:
   A similar argument leads us to conclude that
   \begin{equation*} 
        \frac{\mathbb{P}\left(\mathbf{Z}^{(0)}_{\mathcal{T}}[l^{**}, k^{**}] \in \mathcal{A} \right)}{\mathbb{P}\left(\mathbf{Z}^{(1)}_{\mathcal{T}}[l^{**}, k^{**}] \in \mathcal{A} \right)}
        ) \leq e^\epsilon,
   \end{equation*}
   from which we infer $\mathcal{C}_{\textrm{matrix}}$ it satisfies matrix $t$-node $\epsilon$-DP. 
   
   The ``only if'' statement also follows through a similar elementary argument, and the details are omitted here.
   
   \underline{Proof of (ii)}
Our proof revolves around showing that the signal-to-noise ratio achieved in obtaining $\mathbf{C}[m,l]$ using coding scheme $\mathcal{C}_{\textrm{matrix}}$ is the same (for the worst case distribution $\mathbb{P}_{\mathbf{A},\mathbf{B}}$) as the SNR achieved by $\mathcal{C}$, where $\mathbf{C}=\mathbf{A}\mathbf{B}.$ 
 Consider scalar random variables $A,B$  that satisfy Assumption \ref{assumption1} with $\mathbb{E}[A^2] = \mathbb{E}[B^2] = \eta.$ let $(\overline{R}_1,\overline{R}_2,\ldots, \overline{R}_N)$ and $(\overline{S}_1,\overline{S}_2,\ldots, \overline{S}_N)$ denote the noise random variables of coding scheme $\mathcal{C}$. Denote by $\overline{\mathbf{K_1}}$ and $\overline{\mathbf{K_2}}$ as the covariance matrices of the coding scheme $\mathcal{C}$ as given in \eqref{eq:SNRa}. As per definition~\ref{def:SNRa}, 
    $$
        \texttt{SNR}_a = \frac{\text{det}(\overline{\mathbf{K}}_1)}{\text{det}(\overline{\mathbf{K}}_2)} -1. 
    $$
   For $1 \leq n_1,n_2 \leq N$, the $(n_1,n_2)$-th entry in the matrix $\overline{\mathbf{K}}_1$ is in the form of 
   \begin{equation} \label{eq:scalar_exp}
       \mathbb{E} \left[ \left((A+\overline{R}_{n_1})(B+\overline{S}) \right)^2 \right] \;\; \text { if $n_1 = n_2$ or } \;\; \mathbb{E}\left[ \left((A+\overline{R}_{n_1})(B+\overline{S}_{n_1})\right) \left((A+\overline{R}_{n_2})(B+\overline{S}_{n_2}) \right) \right]~\text { if $n_1 \neq n_2$} ,
   \end{equation}
   and the entries in the $\overline{\mathbf{K}}_2$  have the form of 
   \begin{align*}
       \mathbb{E} \left[ \left((A+\overline{R}_{n_1})(B+\overline{S}_{n_1}) -AB \right)^2 \right] \;\; \text {if $n_1 = n_2$ or } & \mathbb{E}\left[ \left((A+\overline{R}_{n_1})(B+\overline{S}_{n_1})- AB \right) \left((A+\overline{R}_{n_2})(B+\overline{S}_{n_2}) - AB \right) \right]  \\&~~~~~~~~\text { if $n_1 \neq n_2$.} 
   \end{align*}
   
   To evaluate $\textsf{LMSE}(\mathcal{C}_\text{matrix})$, we analyze the accuracy SNR of each entry in $\mathbf{C}$, i.e., $\mathbf{C}[m,k] = \mathbf{A}[m,:] \mathbf{B}[:,k]^T$. Let $\mathbf{K}_{1},\mathbf{K}_{2}$ denote the covariance matrices in the corresponding accuracy SNR calculation. For nodes $n_1,n_2 \in \{1,2,\ldots,N\}$, we define vector notations $\mathbf{a} = \mathbf{A}[m,:]$, $\mathbf{b} = \mathbf{B}[:,k]^T$, $\mathbf{r} = \mathbf{\tilde{R}}_{n_1}[m,:]$, $\mathbf{s} = \mathbf{\tilde{S}}_{n_1}[:,k]^T$, $\mathbf{r'} = \mathbf{\tilde{R}}_{n_2}[m,:]$, and $\mathbf{s'} = \mathbf{\tilde{S}}_{n_2}[:,k]^T$. Then, the $(n_1,n_2)$th entry of signal   covariance matrix $\mathbf{K}_1$ for $C[m,k]$ is composed of: 
   \begin{equation}
              \mathbb{E} \left[ \left((\mathbf{a}+ \mathbf{r})\cdot (\mathbf{b} + \mathbf{s}) \right)^2 \right] \;\; \text { if $n_1 = n_2$ or } \;\; 
               \mathbb{E} \Bigr[ \left((\mathbf{a}+ \mathbf{r})\cdot (\mathbf{b} + \mathbf{s}) \right) \left((\mathbf{a}+ \mathbf{r'})\cdot (\mathbf{b} + \mathbf{s'}) \right)  \Bigr] \text { if $n_1 \neq n_2$} .
   \end{equation}
    Assuming $\mathbb{E} [a_i^2] = \eta, \mathbb{E} [b_i^2] = \eta$\footnote{Elementary linear estimation theory shows that the LMSE obtained, for a fixed noise distribution and decoding co-efficients, is monotonically decreasing in parameters
$\mathbb{E} [a_i^2],\mathbb{E} [b_i^2] < \eta$; so the standard deviations being equal to $\eta$ is the worst case.}, we show that $\mathbf{K}_{1}(n_1,n_2) = L \overline{\mathbf{K}}_1(n_1,n_2)$ in the steps below.

\begin{align}
     \mathbb{E} [(\mathbf{a} \cdot \mathbf{b})^2]  &= \mathbb{E} [ (a_1 b_1 + \cdots + a_L b_L)^2] 
     = \sum_{i=1,\ldots L} \mathbb{E} [a_i^2 b_i^2] + \sum_{i \neq j } \mathbb{E} [a_i b_i a_j b_j] = \sum_{i=1,\ldots L} \mathbb{E} [a_i^2]  \mathbb{E}[b_i^2]  = L \eta^2  \\ 
     &= L \cdot \mathbb{E} [A^2 B^2].
\end{align}

\begin{align}
     \mathbb{E} [(\mathbf{a} \cdot \mathbf{s} + \mathbf{r} \cdot \mathbf{b})^2]  &= \mathbb{E} [ (a_1 s_1 + \cdots + a_L s_L + r_1 b_1 + \cdots + r_L b_L)^2] \\ 
     &= \sum_{i=1,\ldots L} (\mathbb{E} [a_i^2 s_i^2] + \mathbb{E} [r_i^2 b_i^2]) + \sum_{i,j = 1, \ldots L} \mathbb{E} [a_i s_i r_j b_j] + \sum_{i \neq j} (\mathbb{E} [a_i s_i a_j s_j] + \mathbb{E} [r_i b_i r_j b_j]) \\ 
     &= \sum_{i=1,\ldots L} (\mathbb{E} [a_i^2] \mathbb{E}[s_i^2] + \mathbb{E} [r_i^2] \mathbb{E}[b_i^2]) =  L \left(\eta \mathbb{E}[\overline{S}_{n_1}^2] +  \eta \mathbb{E}[\overline{R}_{n_1}^2] \right) \\
     &= L \cdot \mathbb{E}[(A\overline{S}+\overline{R}_{n_1}B)^2]. 
\end{align}
Similarly, 
\begin{align}
     \mathbb{E} [(\mathbf{r} \cdot \mathbf{s})^2]  = \mathbb{E} [ (r_1 s_1 + \cdots + r_L s_L)^2] 
     = \sum_{i=1,\ldots L} \mathbb{E} [r_i^2 s_i^2] + \sum_{i \neq j } \mathbb{E} [r_i s_i r_j s_j] = \sum_{i=1,\ldots L} \mathbb{E} [r_i^2]  \mathbb{E}[s_i^2] = L \cdot\mathbb{E}[\overline{R}_{n_1}^2 \overline{S}_{n_1}^2].
\end{align}
By plugging these in, we obtain:
\begin{align}
     \mathbb{E} \left[ \left((\mathbf{a}+ \mathbf{r})\cdot (\mathbf{b} + \mathbf{s}) \right)^2 \right] = L \cdot  \mathbb{E} \left[ \left((A+\overline{R}_{n_1})(B+\overline{S}_{n_1}) \right)^2 \right]. 
\end{align}
With similar calculations, we can show that: 
\begin{align}
 \mathbb{E} \Bigr[ \left((\mathbf{a}+ \mathbf{r})\cdot (\mathbf{b} + \mathbf{s}) \right) \left((\mathbf{a}+ \mathbf{r'})\cdot (\mathbf{b} + \mathbf{s'}) \right)  \Bigr] = L \cdot  \mathbb{E}\left[ \left((A+\overline{R}_{n_1})(B+\overline{S}_{n_1})\right) \left((A+\overline{R}_{n_2})(B+\overline{S}_{n_2}) \right) \right]. 
\end{align}
%    \begin{align*}
%     \mathbb{E} [\Gamma_{t+1} \Theta_{t+1}[i,j]^2 ] &= \mathbb{E}[ (\mathbf{a} \cdot \mathbf{b} + x (\mathbf{a} \cdot \mathbf{s} + \mathbf{r} \cdot \mathbf{b}) + x^2 \mathbf{r} \cdot \mathbf{s})^2 ] \\
%     &=\mathbb{E}[(\mathbf{a} \cdot \mathbf{b})^2] + x^2 \mathbb{E}[(\mathbf{a} \cdot \mathbf{s} + \mathbf{r} \cdot \mathbf{b})^2] + x^4 \mathbb{E}[(\mathbf{r} \cdot \mathbf{s})^2 ] \\
%     &\leq L \eta^2 + x^2 \cdot 2L \eta + x^4 \cdot L \\ 
%     & = L (\eta^2 + 2 \eta x^2   + x^4). 
% \end{align*}
% \begin{align}
%     \mathbb{E}  [\Gamma_{t+1} \Theta_{t+1}[i,j]   \tilde{\Gamma}\tilde{\Theta}[i,j]] &= L 
% \left( \eta^2 + 2\eta x (x+\alpha_1^{(n)}) + x^2 (x+ \alpha_1^{(n)})^2 \right), \\
%  \mathbb{E}  [\tilde{\Gamma}\tilde{\Theta}[i,j]^2] &= L \left( \eta^2 + 2\eta (x+\alpha_1^{(n)})^2 + (x+\alpha_1^{(n)})^4 + O((\alpha_2^{(n)})^4) \right)
% \end{align}
%
Thus, we have shown that $\mathbf{K_1} = L \cdot \overline{\mathbf{K}}_1.$ It is mechanical to also show $\mathbf{K_2} = L \cdot \overline{\mathbf{K}}_2.$ Hence, the 
\begin{equation}
    \texttt{SNR}_a(\mathcal{C}_\text{matrix}) = \frac{\text{det}(\mathbf{K}_1)}{\text{det}(\mathbf{K}_2)} -1 = \frac{\text{det}(L \cdot \overline{\mathbf{K}}_1)}{\text{det}(L \cdot \overline{\mathbf{K}}_2)} -1 = \frac{\text{det}(\overline{\mathbf{K}}_1)}{\text{det}(\overline{\mathbf{K}}_2)} -1 = \texttt{SNR}_a(\mathcal{C}),
\end{equation}
which implies the theorem statement.

\end{proof}

\section{Precision}
\label{sec:precision}
The coding scheme of Section \ref{sec:achievability} requires sequences $\alpha_1^{(n)}, \alpha_2^{(n)} \to 0.$ Notably this translates to requirements of increased compute precision.  In this section, we quantify the price of our coding schemes in terms of the required precision. We compare two schemes (i) the scheme of Theorem \ref{thm:main_achievability} that requires $N=t+1$ nodes (ii) the BGW coding scheme that achieves perfect privacy and perfect accuracy with $N=2t+1$ nodes. 
We consider a lattice quantization scheme with random dither and show that with this quantizer, the former scheme requires much more computing \vc{as compared to} the latter scheme.

\
\begin{table}[htb!]
\centering
\begin{tabular}{|p{5cm}|| p{3cm}| p{3cm} |p{3.5cm}|} 
 \hline
Number of nodes $N$ & Infinite-Precision accuracy (MSE) & Target accuracy (MSE) with finite precision & Number of bits $M(\delta)$ per node required to meet target error \\ [0.5ex] 
 \hline\hline
  $2t+1$ (BGW coding scheme) & 0 & $\delta$ & $\lim_{\delta \rightarrow 0} \frac{M({\delta})}{\log \frac{1}{\delta}} = 0.5$ \\\hline
 $t+1$ (Our coding scheme) & $\frac{(\sigma^{*}(\epsilon))^4}{(1+(\sigma^{*}(\epsilon))^2}$ & $\frac{(\sigma^{*}(\epsilon))^4}{(1+(\sigma^{*}(\epsilon))^2}+\delta$ &  $\lim_{\delta \rightarrow 0}\frac{M({\delta})}{\log \frac{1}{\delta}} = 1.5$ \\
 \hline
\end{tabular}
\caption{\small{A depiction of the privacy-accuracy trade-off taking into account the number of bits of precision required. The accuracy is reported as mean square error (MSE) assuming that $\epsilon$-DP is required to be achieved; we assume $\eta=1$ for simplicity.}}
\label{table:2}
\end{table}
\vc{
A comparison between the two schemes is shown in Table \ref{table:2}; we assume for simplicity that $\eta=1$ in the table and in the remainder of this section. The study of the effect of finite precision on secure MPC schemes based on secret sharing, including the BGW coding scheme, is a rich area of study (see for example \cite{catrina2010secure, catrina2009multiparty} and several follow up works). Table \ref{table:2} is based on this analysis, see for example Sec. 3.3 in \cite{catrina2010secure}). In Table \ref{table:2}, we are only concerned with asymptotic number of bits as $\delta \to 0,$ whereas this body of work aims to study the non-asymptotic performance. In addition, to the best of our knowledge, these works require $A,B$ to be bounded, which is stronger than Assumption \ref{assumption1}. In Appendix \ref{app:BGWprecision}, we show that asymptotic performance noted in Table \ref{table:2} for $N \geq 2t+1$ holds even if we only made Assumption \ref{assumption1} for the real-valued secret sharing scheme explained in Remark \ref{eq:remark5}\footnote{\vc{Another difference between classical secure MPC and real-valued secret sharing of Remark \ref{eq:remark5} is that, in the former, \emph{perfect privacy} interpreted as statistical independence, whereas, in the latter, it is interpreted a family of schemes that can drive the DP parameter $\epsilon$ to $0$. Notice that the former notion is stronger than the latter notaion because statistical independence automatically implies $\epsilon =0.$}}.}

Table \ref{table:2} wqa1se indicates that for a mean square error increase of $\delta$ compared to the infinite precision counter-part, our scheme requires nearly $1.5 \log\left(\frac{1}{\delta}\right)$ bits of precision per computation node for arbitrarily small $\delta$, whereas the BGW coding scheme requires nearly $0.5 \log\left(\frac{1}{\delta}\right)$ bits. Since we require $N=(t+1)$ computation nodes, the total number of bits required for our scheme is $\frac{3t+3}{2} \log\left(\frac{1}{\delta}\right)$ bits, whereas real-valued secret sharing scheme requires $\frac{2t+1}{2} \log\left(\frac{1}{\delta}\right)$ bits. Our result indicates that our approaches of this paper are not to be viewed as a panacea for computation overheads of secure multiparty computation. Rather, the schemes provide a pathway for increased trust as there is explicit control on the information leakage even if $N-1$ nodes collude. This increased trust comes at the cost of reduced privacy, accuracy, and a moderate increase in the overall computation overhead. Specifically, for a fixed value of $N$, our methods allow for secure multiplication in systems where the parameter $t$ is allowed to exceed $\lceil \frac{N-1}{2}\rceil.$ 

We also emphasize that our results here pertain to a specific choice of the coding scheme of Section \ref{sec:achievability} and a specific quantization scheme. Our analysis, therefore, shows that general-purpose quantizers coupled with the achievable scheme of Section \ref{sec:achievability} incurs a computational overhead as \vc{compared to the BGW coding scheme.} The question of the existence and design of quantizers or coding schemes that reduce this computation penalty is open. We describe our setup and results in greater detail next. An brief analysis of the BGW scheme providing justification to Table \ref{table:2} is placed in Appendix \ref{app:BGWprecision}.

\subsection{Achievable scheme of Section \ref{sec:achievability} under finite compute precision}
Consider the coding scheme of Section \ref{sec:achievability}, where
$$\Gamma_i= \left[A~R_1~R_2~\ldots~R_t\right] \vec{v}_{i}$$
$$ \Theta_i = \left[B~S_1~S_2~\ldots~S_t\right] \vec{w}_{i}$$
where $\vec{v}_i,\vec{w}_i$ are specified in (\ref{eq:ach_v1_tgreat1})-(\ref{eq:ach_v2_teq2}), that is, for $t \geq 2,$
$$ \vec{v}_{t+1} = \vec{w}_{t+1} = \begin{bmatrix}1 \\ x \\ 0 \\ \vdots \\ 0\end{bmatrix}$$ $$\vec{v}_{i} = \vec{w}_{i} = \vec{v}_{t+1} + \begin{bmatrix}0 \\ \alpha_1\\ \alpha_{2}\vec{g}_{i}\end{bmatrix}, 1 \leq i \leq t$$
For $t=1$.
$$ \vec{v}_{2} = \vec{w}_{2} = \begin{bmatrix}1 \\ x \end{bmatrix}, \vec{v}_{1} = \vec{w}_{1} = \vec{v}_{2} + \begin{bmatrix}0 \\ \alpha_1 \end{bmatrix}$$
In the above equations, $\mathbf{G}$ is a constant matrix that satisfies properties $(C1)$ and $(C2)$ specified in Section \ref{sec:coding_scheme}. Here, we have suppressed the dependence on the sequence index $n$ in parameters $\alpha_1,\alpha_2$. For our discussion here, it suffices to remind ourselves that, when there is no quantization error, the coding scheme achieves the mean square error: $$\mathbb{E}[(AB-\hat{C})^2] = \frac{x^4}{(1+x^2)^2}+\delta =\frac{1}{(1+\texttt{SNR}_p)^2}+\delta,$$ where $\delta \to 0$ so long as as $\alpha_1, \frac{\alpha_2}{\alpha_1}, \frac{\alpha_1^2}{\alpha_2} \rightarrow 0.$ In the sequel, we analyze the effect of quantization error on the mean square error.

Let $\Lambda \subset \mathbb{R}$ be a lattice with Vornoi region $\mathcal{V} \subset \mathbb{R}$. Let $D_1^{(A)}, D_2^{(A)},\ldots, D_N^{(A)},D_1^{(B)}, D_2^{(B)},\ldots, D_N^{(B)} $ be random variables uniformly distributed over $\mathcal{V}$, that are independent of each other and all the random variables in the coding scheme.  Consider the coding scheme of Section \ref{sec:achievability}, but with the inputs to the computation nodes are quantized to $M$ bits. 
$$\hat{\Gamma_i} = Q_i^{(A)}\left(\Gamma_i\right),$$
$$\hat{\Theta_i} = Q_i^{(B)}\left(\Theta_i\right),$$
where $Q_i^{(A)}:\mathbb{R} \rightarrow \mathbb{R}, Q_i^{(B)}:\mathbb{R} \rightarrow \mathbb{R}$ are independent dithered lattice quantizers. Specifically, 
we let $Q_i^{(A)}(x) = Q_{\Lambda}(x)+ D_i^{(A)}$ where $Q_{\Lambda}(x): \mathbb{R} \rightarrow \Lambda$ denotes the nearest point in $\Lambda$ to $x$. We define: $$ Y_i \stackrel{\Delta}{=} \Gamma_i - \hat{\Gamma}_i$$
$$ Z_i \stackrel{\Delta}{=} \Theta_i - \hat{\Theta}_i.$$

Standard lattice quantization theory \cite{zamir2014lattice} dictates that $Y_i,Z_i$ are independent of $\Gamma_{i},\Theta_{i}.$ We assume that the lattice $\Lambda$ is designed  - possibly based on the knowledge of distributions of $\Gamma_i,\Theta_i,i=1,2,\ldots,N$ -  so that we have $M$ bit quantizers, that is: $H(\hat{\Gamma}_i),H(\hat{\Theta}_i) \leq M.$ Assuming $A,B$ are random variables with finite differential entropy, it follows that $\mathbb{E}[Y_i^2],\mathbb{E}[Z_i^2] = \Omega(2^{-2M}).$

We also assume that $Q_i^{(A)},Q_i^{(B)}$ are statistically independent of each other, implying that $Y_i$ is independent of $(B,S_1,\ldots,S_t)$ and similarly, $Z_i$ is independent of $(A,R_1, \ldots, R_t).$

The output of the $i$th node is $\hat{C}_{i} = \hat{\Gamma_{i}}\hat{\Theta_i}$ - that is, we assume that the computation node performs perfectly precise computation so long as the inputs are quantized to $M$ bits. We consider a linear decoding strategy that estimates $C=AB$ as$$\hat{C} = \sum_{i=1}^{N} d_i \hat{C}_i = \sum_{i=1}^{N} d_i\hat{\Gamma_i}\hat{\Theta_i}.$$

Consider a fixed $\delta > 0$. We assume that parameters $\alpha_{1}(\delta),\alpha_2(\delta), M(\delta),d_1(\delta),d_2(\delta),\ldots,d_{t+1}(\delta)$ are chosen to satisfy the accuracy limit
$$\mathbb{E}[(AB-\hat{C})^2] \leq  \frac{1}{(1+\texttt{SNR}_p)^2} + \delta$$
for all $\mathbb{P}_{A}\mathbb{P}_{B}$ that satisfy $\mathbb{E}[A^2], \mathbb{E}[B^2] \leq 1.$ We develop two results next.
First, we show that for any choice of $\alpha_{1}(\delta),\alpha_2(\delta), M(\delta),d_1(\delta),d_2(\delta),\ldots,d_{t+1}(\delta)$, the following lower bound holds:
$$\lim_{\delta \to 0} \frac{M(\delta)}{\log \left(\frac{1}{\delta}\right)} \geq 3/2.$$
Then, we show that there exists a positive number $\overline{\delta} > 0$ and a realization of parameters $$\alpha_{1}(\delta),\alpha_2(\delta), M(\delta),d_1(\delta),d_2(\delta),\ldots,d_{t+1}(\delta)$$ such that, if   
$$\lim_{\delta \to 0} \frac{M(\delta)}{\log \left(\frac{1}{\delta}\right)} > 3/2,$$
then 
$$\mathbb{E}[(AB-\hat{C})^2] \leq  \frac{1}{(1+\texttt{SNR}_p)^2} + \delta$$
for all $\delta < \overline{\delta}.$ In the sequel, we often suppress the dependence on $\delta$ for all parameters except the number of quantization bits $M(\delta)$ for simpler notation.  We first show the lower bound.

%As we will see, such an asymptotic analysis helps by abstracting out the effect of the actual distribution of the data $A,B$. Notice that data processing inequality automatically implies that the privacy analysis of Sec. \ref{sec:privacy_analysis} and Corollary \ref{cor:achievability} readily applies. Our approach is to lower bound the accuracy via an analysis similar to Sec. \ref{sec:accuracy}, but modifying appropriately to account for the quantization error. 

\subsection{Lower Bound}
In the sequel, we assume that $\mathbb{E}[A]=\mathbb{E}[B]=0,$ and consider $\mathbb{E}[A^2]=\mathbb{E}[B^2] = 1.$ 

\begin{align}
  &\mathbb{E}[(AB-\hat{C})^2] = \frac{1}{1+\texttt{SNR}_a}+\delta \\&= \min_{d_1,d_2,\ldots,d_N}  \mathbb{E}[(\sum_{i=1}^{N} d_i\hat{C_i}-\Gamma_i\Theta_i)^2] \nonumber \\&= \min_{d_1,d_2,\ldots,d_N} 
\left(\mathbb{E}\bigg[\bigg(\sum_{i=1}^{t}d_i\big(A+R_1 (x + \alpha_1)+ \alpha_2 \begin{bmatrix}R_2&\ldots & R_t\end{bmatrix}\vec{g}_{i} + Y_i\big)\big(B+\right.\\&\left.~~~~S_1 (x + \alpha_1) + \alpha_2 \begin{bmatrix}S_2&\ldots & S_t\end{bmatrix} \vec{g}_i+Z_i\big)+
 d_{t+1} \big(A+R_1 x +Y_{t+1} \big)\big(B+S_1 x+Z_{t+1})-AB\bigg)^2\bigg]\right)
 \\& = 
 \min_{d_1,d_2,\ldots,d_N} 
\left(\mathbb{E}\bigg[\bigg(\sum_{i=1}^{t}d_i\big(A+R_1 (x + \alpha_1)+ \alpha_2 \begin{bmatrix}R_2&\ldots & R_t\end{bmatrix}\vec{g}_{i} \big)\big(B+\right.\nonumber\\&~~~~S_1 (x + \alpha_1) + \alpha_2 \begin{bmatrix}S_2&\ldots & S_t\end{bmatrix} \vec{g}_i\big)+
 d_{t+1} \big(A+R_1 x\big)\big(B+S_1 x)-AB\bigg)^2\bigg]\nonumber\\&~~~ \left.+ \sum_{i=1}^{t+1} d_i^2(\beta_i^2 \mathbb{E}[Y_i^2] +\gamma_i^2\mathbb{E}[Z_i^2])\right) \label{eq:prec-delta-bound}
 \end{align}
where
 $$\gamma_i^2 = \mathbb{E}\left[\big(A+R_1 (x + \alpha_1)+ \alpha_2 \begin{bmatrix}R_2&\ldots & R_t\end{bmatrix}\vec{g}_{i} \big)^2\right],i=1,2,\ldots,t$$ 
 $$ \gamma_{t+1}^2 = \mathbb{E}\left[\big(A+R_1 (x + \alpha_1)\big)^2\right] $$
 $$\beta_i^2 = \mathbb{E}\left[\big(B+S_1 (x + \alpha_1) + \alpha_2 \begin{bmatrix}S_2&\ldots & S_t\end{bmatrix} \vec{g}_i\big)^2\right],i=1,2,\ldots,t$$
  $$\beta_{t+1}^2 = \mathbb{E}\left[\big(B+S_1 (x + \alpha_1)\big)^2\right]$$
 We now lower bound $\delta$ in two ways. The first imposes an upper bound on $\alpha_1$ in terms of $\delta$, for sufficiently small $\delta$. The second uses the first bound to impose a lower bound on $M$.  The first bound begins with omitting the effect of $Y_i,Z_i,i=1,2,\ldots,t+1$ from (\ref{eq:prec-delta-bound}) as follows:
 
  \begin{align}
&  \mathbb{E}[(AB-\hat{C})^2]  \\ &\geq \min_{d_1,d_2,\ldots,d_N} 
\left(\mathbb{E}\bigg[\bigg(\sum_{i=1}^{t}d_i\big(A+R_1 (x + \alpha_1)+ \alpha_2 \begin{bmatrix}R_2&\ldots & R_t\end{bmatrix}\vec{g}_{i} \big)\big(B+\right.\nonumber \\&\left.~~~~S_1 (x + \alpha_1) + \alpha_2 \begin{bmatrix}S_2&\ldots & S_t\end{bmatrix} \vec{g}_i\big)+
 d_{t+1} \big(A+R_1 x\big)\big(B+S_1 x)-AB\bigg)^2\bigg]\right) \\
 &\geq \min_{d_1,d_2,\ldots,d_N} 
 \mathbb{E}\bigg[\bigg(\left(\sum_{i=1}^{t}d_i\right)\big(A+R_1 (x + \alpha_1)\big)\big(B+S_1 (x + \alpha_1)\big) +d_{t+1}\big(A+R_1 x \big)\big(B+S_1 x)-AB\bigg)^2\bigg]  \\
 &= \frac{2x^2\left((\alpha_1)^2 + 2 (\alpha_1)x +  x^2\right)}{(\alpha_1)^2(2x^2+1)+4(\alpha_1)(x+x^3)+2(x^2+1)^2}\\
 &=\frac{2x^4}{2(x^2+1)^2} + \Theta(\alpha_1) \\& = \frac{1}{1+\texttt{SNR}_{p}^2}+\Theta(\alpha_1)
 \label{eq:precision-delta-bound1}
   \end{align}
   
From the final equation, we infer that there exists a constant $\overline{\delta} > 0$ and a constant $c'$ such that, for all $\delta < \overline{\delta},$ $\alpha_1 \leq c'\delta.$ We now derive the second bound on $\delta$. We begin the bounding process by omitting the effect of $Z_i,i=1,2,\ldots,t+1$. In the following bounds, we use the fact that $\beta_{i}^2 \geq \mathbb{E}[B^2] = 1, \forall i$. Also, we assume that there is a constant $\lambda > 0$ such that  $\mathbb{E}[Y_i^2],\mathbb{E}[Z_i^2] \geq \lambda 2^{-2M(\delta)}.$
 \begin{align}
    &\mathbb{E}[(AB-\hat{C})^2]  \\ & {\geq} 
\min_{d_1,d_2,\ldots,d_N} 
 \mathbb{E}\bigg[\bigg(\sum_{i=1}^{t}d_i\big(A+R_1 (x + \alpha_1)\big)\big(B+S_1 (x + \alpha_1)\big) +d_{t+1}\big(A+R_1 x \big)\big(B+S_1 x)-AB\bigg)^2\bigg] \nonumber \\& + \left(\sum_{i=1}^{t+1}d_i^2\beta_i^2 \mathbb{E}[Y_i]^2\right) \label{eq:prec-bound-interm1}\\ &
  = \min_{d_1,d_2,\ldots,d_N} \left(\sum_{i=1}^{t+1} d_i-1\right)^2 \mathbb{E}[A^2B^2] + \left(\mathbb{E}[A^2S_1^2]+\mathbb{E}[B^2 R_1^2]\right)\left(\left(\sum_{i=1}^{t}d_i\right) (x+\alpha_1) + xd_{t+1}\right)^2 \nonumber \\&+ \left(\left(\sum_{i=1}^{t} d_i\right)(x+\alpha_1)^2 + x^2 d_{t+1}\right)^2\mathbb{E}[R_1^2S_1^2] + \left(\sum_{i=1}^{t+1}d_i^2\beta_i^2 \right)\lambda 2^{-2M(\delta)} \label{eq:prec-bound-interm2} \\
  &\geq \min_{d_1,d_2,\ldots,d_N} \left(\sum_{i=1}^{t+1} d_i-1\right)^2  + 2 \left(\left(\sum_{i=1}^{t}d_i\right) (x+\alpha_1) + xd_{t+1}\right)^2 \nonumber \\&+ \left(\left(\sum_{i=1}^{t} d_i\right)(x+\alpha_1)^2 + x^2 d_{t+1}\right)^2 + \left(\sum_{i=1}^{t+1}d_i^2\right)\lambda 2^{-2M(\delta)} \label{eq:prec-bound-interm3} \\
% & \geq \min_{d_1,d_2,\ldots,d_N}\left( \left(\sum_{i=1}^{t+1} d_i-1\right)^2  + 2 \left(\left(\sum_{i=1}^{t}d_i\right) (x+\alpha_1) + xd_{t+1}\right)^2 \right. \nonumber \\&\left.+ \left(\left(\sum_{i=1}^{t} d_i\right)(x+\alpha_1)^2 + x^2 d_{t+1}\right)^2 + \left(\sum_{i=1}^{t+1}d_i\right)^2 \frac{2^{-2B}}{t} \right)\\
& = \min_{\overline{d},d_{t+1}} \left(\overline{d}+d_{t+1}-1\right)^2  + 2 \left(\overline{d} (x+\alpha_1) + xd_{t+1}\right)^2 + \left(\overline{d}(x+\alpha_1)^2 + x^2 d_{t+1}\right)^2 + (\overline{d}+d_{t+1})^2 \frac{\lambda 2^{-2M(\delta)}}{t}
  \label{eq:precision-bound1}
\end{align}
where, in (\ref{eq:prec-bound-interm1}), we have used the fact that $R_2,R_3,\ldots,R_{t+1}$ are independent of each other and all other variables, that is, independent of $(A,B,R_1,Y_1,Y_2,\ldots,Y_{t+1},Z_1,Z_2,\ldots,Z_{t+1})$. In (\ref{eq:precision-bound1}), we have used the notation $\overline{d}=\sum_{i=1}^{t}d_i. $The final minimization problem is strictly convex. Its optimal arguments $\overline{d}^{*},d_{t+1}^{*}$, can be found via differentiation to be:
 \begin{eqnarray} \overline{d}^{*} 
  &=&  \frac{\begin{vmatrix}  1&  
 (x(x+\alpha_1) + c)^2 \\  1 & 
((x+\alpha_1)^2+c)^2 \end{vmatrix}}{\begin{vmatrix}  (x^2+c)^2 &  
 (x(x+\alpha_1)+c)^2 \\  (x(x+\alpha_1)+c)^2& 
 ((x+\alpha_1)^2+c)^2 \end{vmatrix}}\\
  &=& \frac{(\alpha_1 + x) (2 + (\alpha_1 + x) (\alpha_1+ 2 x))}{\alpha_1 (2 x^2 (\alpha_1 + x)^2 + 
   c^2 (2 + (\alpha_1 + 2 x)^2))}\label{eq:prec-opt-dvalues1}\\
    {d}_{t+1}^{*}
 &=&  \frac{\begin{vmatrix}  (x^2+c)^2&  
 1\\  (x(x+\alpha_1)+c)^2 & 
1 \end{vmatrix}}{\begin{vmatrix}  (x^2+c)^2 &  
 (x(x+\alpha_1)+c)^2 \\  (x(x+\alpha_1)+c)^2& 
 ((x+\alpha_1)^2+c)^2 \end{vmatrix}}\\
&=& - \frac{(\alpha_1 + 2 x) (2 + (\alpha_1)^2 + 2 \alpha_1x + 2 x^2)}{\alpha_1 (2 x^2 (\alpha_1 + x)^2 + 
   c^2 (2 + (\alpha_1 + 2 x)^2))}\label{eq:prec-opt-dvalues2}
\end{eqnarray}

where $c = \sqrt{1+\frac{\lambda 2^{-2M(\delta)}}{t}}$. Substituting (\ref{eq:prec-opt-dvalues1}),(\ref{eq:prec-opt-dvalues2}) into the last term in  (\ref{eq:precision-bound1}), we get
\begin{align}  
&\mathbb{E}[(AB-\hat{C})^2] \nonumber \\ &\geq\min_{\overline{d},d_{t+1}} \bigg(\left(\overline{d}+d_{t+1}-1\right)^2  + 2 \left(\overline{d} (x+\alpha_1) + xd_{t+1}\right)^2 + \left(\overline{d}(x+\alpha_1)^2 + x^2 d_{t+1}\right)^2\bigg) + (\overline{d}^*+d_{t+1}^*)^2 \frac{\lambda 2^{-2M(\delta)}}{t}\\
&\geq \frac{1}{1+\texttt{SNR}_p^2} + (\overline{d}^*+d_{t+1}^*)^2 \frac{\lambda 2^{-2M(\delta)}}{t}
\end{align}
Using the fact that $$E[(AB-C)^2] \leq \frac{1}{1+\texttt{SNR}_p^2} + \delta,$$
we get $\delta \geq (\overline{d}^*+d_{t+1}^*)^2 \frac{\lambda 2^{-2M(\delta)}}{t}.$ Recall that for $\delta < \overline{\delta}$, we have shown that  $\alpha_1 \leq c' \delta.$ Further (\ref{eq:prec-opt-dvalues1}),(\ref{eq:prec-opt-dvalues2}) imply that for there are constants $\overline{\alpha}_1 > 0$ and $c''> 0$ such that, for $0 \leq \alpha_1 < \overline{\alpha}_1,$ $|\overline{d}^{*}+d_{t+1}^{*}| \geq \frac{c'' }{\alpha_1}$. Therefore, for $\delta < \min(\overline{\delta},\overline{\alpha_1}/c'),$ we have:
 \begin{eqnarray*} \delta &\geq& \frac{c''^2\lambda}{\alpha_1^2} \frac{2^{-2M(\delta)}}{t} \\ &\geq&  \frac{c''^2}{c'^2\delta^2} \frac{2^{-2M}}{t} \\
 \Rightarrow M(\delta) &\geq& \frac{3}{2}\log\frac{1}{\delta} + \frac{1}{2} \log\left(\frac{c''^2 \lambda}{c'^2t}\right)
 \end{eqnarray*}
 Thus, we have the following asymptotic bound: $\lim_{\delta \to 0} \frac{M(\delta)}{\log{\frac{1}{\delta}}} \geq 3/2$.

% $$\mathbb{E}[(AB-\hat{C})^2] = \min_{d_1,d_2,\ldots,d_N} \left(E\bigg[\bigg(\sum_{i=1}^{t}d_i\big(A+R_1 (x + \alpha_1)+ \alpha_2 \begin{bmatrix}R_2&\ldots & R_t\end{bmatrix}\vec{g}_{i} + Y_i\big)\big(B+\right.\\&\left.~~~~S_1 (x + \alpha_1) + \alpha_2 \begin{bmatrix}S_2&\ldots & S_t\end{bmatrix} \vec{g}_i+Z_i\big)+ d_{t+1} \big(A+R_1 x +Y_{t+1} \big)\big(B+S_1 x+Z_{t+1})-AB\bigg)^2\bigg]\right)$$

\subsection{Achievable scheme}
We set $d_1,d_2,\ldots,d_{t+1}$ to be the same as Section \ref{sec:accuracy}. That is,  $d_1,d_2,\ldots,d_{t+1}$ are set ignoring the effect of the quantization error. For ease of notation, we simply set $\alpha_2 = \alpha_1^{2/3}$. So long as $\alpha_1\to 0$, notice that $\alpha_1,\alpha_2$ satisfy the limit requirements of (\ref{eq:limits}). We set $$M(\delta)= K \log \frac{1}{\delta}$$ for some constant $K>3/2.$ For a sufficiently small $\delta,$ we show that the scheme achieves an error smaller than $\delta$ so long as $\alpha_1$ is chosen sufficiently small. 

Taking into effect the quantization error, the computation nodes output:
$$ \hat{\Gamma}_{t+1}\hat{\Theta}_{t+1} =  (A+R_1 x+Y_{t+1})(B+S_1 x+Z_{t+1})$$
and, for $i=1,\ldots,t$:
\begin{align*} \hat{\Gamma}_i \hat{\Theta}_i &= (A+R_1(x+\alpha_1)+ \alpha_1^{2/3}\begin{bmatrix}R_2&\ldots & R_t\end{bmatrix}\vec{g}_{i} + Y_{i} )( (B+S_1(x+\alpha_1))+\alpha_1^{2/3}\begin{bmatrix}S_2&\ldots & S_t\end{bmatrix} \vec{g}_{i} +Z_{i}) \end{align*}

Following the steps of Section \ref{sec:accuracy}, the decoder obtains the following analogous to (\ref{eq:ach-decoder-processing1}), (\ref{eq:ach-decoder-processing2}):
\begin{eqnarray}\hat{\Gamma}_{t+1}\hat{\Theta}_{t+1} &=& AB+x(AS_1+BR_1) + R_1S_1 x^2 + Y_{t+1}(B+Sx)+Z_{t+1}(A+Rx) + Y_{t+1}Z_{t+1}    \label{eq:ach-decoder-processing-quant1} \\ \hat{\tilde{\Gamma}}\hat{\tilde{\Theta}} &=&  AB + (x+\alpha_1)(AS_1+BR_1) + (x+\alpha_1)^2 R_1 S_1 + O(\alpha_1^{4/3})\nonumber \\&&~~+ \sum_{i=1}^{t} (L_{i,1} Y_i + L_{i,2} Z_i + L_{i,3} Y_iZ_i) . \label{eq:ach-decoder-processing-quant2} \end{eqnarray}
where $L_{i,1},L_{i,2},L_{i,3},i=1,2,\ldots,t$ are random variables that are independent of $Y_i,Z_i$ with the property that their variances are $\Theta(1)$. That is, as $\alpha_1 \to 0$, their variances depend only on the variances of $A,B,R_1,\ldots,R_{t+1},S_1\ldots,S_{t+1}$ and constants $x,\mathbf{G}.$ To make the notation of (\ref{eq:ach-decoder-processing-quant1}) consistent with (\ref{eq:ach-decoder-processing-quant2}), we denote: $$L_{t+1,1} = B+S_1x,L_{t+1,2} = A+R_1x, L_{t+1,3} = 1.$$
Let $\overline{d}^*,d_{t+1}^*$ denote the constants that obtain the \LMSE of Section \ref{sec:accuracy}, specifically, these constants are the arguments that minimize the following:
\begin{align*}& \min_{\overline{d},d_{t+1}}&\mathbb{E}\left[\left(\overline{d} \left( AB + (x+\alpha_1)(AS_1+BR_1)+R_1S_1(x+\alpha_1)^2\right)+ \right.\right.\\&&~~\left.\left.  + d_{t+1}\left(AB+x(AS_1+BR_1) + R_1S_1 x^2\right)-AB\right)^2\right]\end{align*}
expressions in (\ref{eq:SNR1})-(\ref{eq:SNR3}). The mean square error can be written as:

 \begin{align*}\min_{\overline{d},d_{t+1}}&\mathbb{E}\left[\left(\overline{d} \left( AB + (x+\alpha_1)(AS_1+BR_1)+R_1S_1(x+\alpha_1)^2\right)+ \right.\right.\\&~~\left.\left.  + d_{t+1}\left(AB+x(AS_1+BR_1) + R_1S_1 x^2\right)-AB\right)^2\right]\\&+ \sum_{i=1}^{t}(\overline{d}^{*})^2 \left(L_{i,1}^2 \mathbb{E}[Y_i^2] + L_{i,2}^2 \mathbb{E}[Z_i^2] + L_{i,3}^2 \mathbb{E}[Y_i^2Z_i^2]\right) \\&+ (d_{t+1}^{*})^2 \left(\left(L_{t+1,1}^2 \mathbb{E}[Y_{t+1}^2] + L_{t+1,2}^2 \mathbb{E}[Z_{t+1}^2] + L_{t+1,3}^2 \mathbb{E}[Y_{t+1}^2Z_{t+1}^2]\right)\right) \\ 
 &= \frac{2x^2\left((\alpha_1)^2 + 2 (\alpha_1)x +  x^2 +  \frac{O(\alpha_1^6)}{{(\alpha_1)^2}}\right)}{(\alpha_1)^2(2x^2+1)+4(\alpha_1)(x+x^3)+2(x^2+1)^2 +  \frac{O(\alpha_1^6)}{(\alpha_1)^2} }\\&+ \sum_{i=1}^{t}(\overline{d}^{*})^2 \left(L_{i,1}^2 \mathbb{E}[Y_i^2] + L_{i,2}^2 \mathbb{E}[Z_i^2] + L_{i,3}^2 \mathbb{E}[Y_i^2Z_i^2]\right) \\&+d_{t+1}^2
 \left(\left(L_{t+1,1}^2\mathbb{E}[Y_{t+1}^2] + L_{t+1,2}^2 \mathbb{E}[Z_{t+1}^2] + L_{t+1,3}^2 \mathbb{E}[Y_{t+1}^2Z_{t+1}^2]\right)\right)\\
 &= \frac{x^4}{(x^2+1)^2}+\Theta(\alpha_1) + \sum_{i=1}^{t}(\overline{d}^{*})^2 \left(L_{i,1}^2 \lambda 2^{-2M(\delta)} + L_{i,2}^2 \lambda 2^{-2M(\delta)} + L_{i,3}^2 \lambda^2 2^{-4M(\delta)}\right) \\& + (d_{t+1}^{*})^2 \left(L_{t+1,1}^2 \lambda 2^{-2M(\delta)} + L_{t+1,2}^2 \lambda 2^{-2M(\delta)}  + L_{t+1,3}^2 \lambda^2 2^{-4M(\delta)} \right)
  \end{align*}

 An analysis similar to (\ref{eq:prec-opt-dvalues1}),(\ref{eq:prec-opt-dvalues2}) implies that $|\overline{d}^*|,|d_{t+1}^*| = \Theta\left(\frac{1}{\alpha_1}\right)$. Because of our choice of $M(\delta),$ we have $2^{-2M} < \delta^3.$ Consequently, the mean square error can be expressed as:
 $$ \frac{1}{1+\texttt{SNR}_p^2} + \Theta(\alpha_1)+\Theta(\frac{\delta^3}{\alpha_1^2}).$$ 
Noting that the mean square error is $\frac{1}{1+\texttt{SNR}_p^2}+\delta$, we conclude that $\delta = \Theta(\alpha_1)$. Thus, if $M = K \ log(1/\delta)$ for any $K > 3/2,$ we can obtain an error of at most $\frac{1}{1+\texttt{SNR}_p^2}+\delta$ by choosing $\alpha_1$ sufficiently small.

\section{Conclusion}
In this paper, we propose a new coding formulation that makes connections between secret coding schemes used widely in multiparty computation literature and differential privacy. An exploration of the proposed formulation leads to counter-intuitive correlation structures and noise distributions. This work opens up several open problems and research directions.

The schemes we developed come at the cost of increased precision requirements for a desired accuracy level. A similar phenomenon has also been noted recently in coded computing \cite{HaewonJeong2021, wang2021price}. An open area of research is to understand the fundamental role of quantization and precision on multi-user privacy mechanisms starting with the application of secure multiplication. Specifically, an open question is whether there are coding and quantization schemes that achieve our privacy-accuracy trade-off limits, but provide improvements in terms of precision as compared to those presented in Section \ref{sec:precision}.

In this paper, the coding schemes we constructed as well as our precision analyses are asymptotic. An important question of practical interest is the study of regimes with finite (non-asymptotic) precision. We generated the coding scheme described in \ref{sec:coding_scheme} for $t=2,3,4$ with $\textrm{SNR}_p=1$. To satisfy \eqref{eq:limits}, we set $\alpha_1= \frac{1}{n}$ and $\alpha_2 = \alpha_1 \log(\frac{1}{\alpha_1})$. The results of the simulation are given in Fig.~\ref{fig:t234_simulation}. As we expect from the theory, as $n$ grows, the gap between $1+\text{SNR}_a$ and  $(1+\text{SNR}_p)^2$ becomes smaller. However, for $t=3$ and $t=4$, there remains a gap of $\sim 4.5$ when $n=10,000$. Notice that the results of this paper do not explain this behavior; all the trade-offs presented seem blind to the choice of $t$. A theoretical explanation of this behavior and the determination of optimal choices of $\alpha_1$ and $\alpha_2$ for fixed precision is an open question. 

Our approach to embedding information at different amplitude levels bears resemblance to interference alignment coding schemes for wireless interference networks, see  \cite{Bresler_Parekh_Tse, Jafar_Vishwanath_GDOF, Huang_Cadambe_Jafar,niesen2013interference} and references therein. We wonder if there are deeper connections between the two problems, and whether there are ideas that can be borrowed from the rich literature in wireless network signaling into differentially private multiparty computation.

Finally, while we focused on the canonical computation of matrix multiplication, the long-term promise of this direction explored in this paper is the reduction of communication and infrastructural overheads for private computation of more complex functions. Incorporating our techniques into multiparty computation schemes and developing coding schemes for more complex functions—particularly functions relevant to machine learning applications—is an exciting direction of future research.

\begin{figure*}[t]
    \centering
\includegraphics[width=0.7\textwidth]{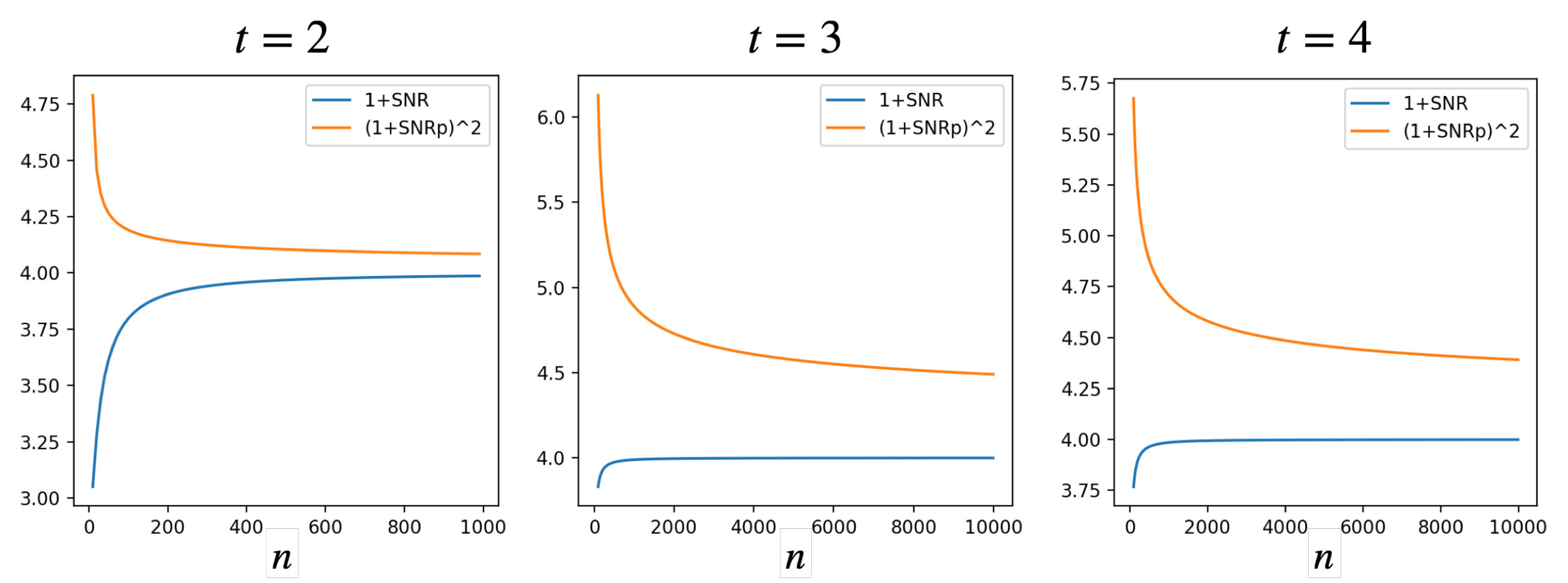}
    \caption{\small{Plotting the gap between $1+\text{SNR}_a$ and $(1+\text{SNR}_p)^2$ for the achievable scheme for $t=2,3,4$ and $N=t+1$. We vary $n$ from 10 to 10,000 and we observe that as $n$ grows the gap reduces.}}
    \label{fig:t234_simulation}
    \vspace{-7pt}
\end{figure*}

\section{Acknowledgement}
This work is partially funded by the National Science Foundation under grants CAREER 1845852, FAI 2040880,  CIF 1900750, 1763657, and 2231706.

\newpage

\bibliographystyle{IEEEtran}
\bibliography{IEEEabrv,references}

\newpage 
\appendices
\vc{\section{Details of Baseline Scheme Performance of Fig. \ref{fig:code_construction}}
\label{app:SS_staircase_lower}
In this appendix, we provide missing details for the baseline schemes described in Sec. \ref{sec:summary} and used in Fig. \ref{fig:code_construction}. We use the model and notation introduced in Sections \ref{sec:model} and \ref{sec:SNR} in this section. Recall that Fig. \ref{fig:code_construction} considered the case of $N=3$ nodes and $t=2$ colluding adversaries, and assumed that $A,B$ are unit-variance random variables, that is, $\eta=1$ as per the model of Sec. \ref{sec:model}.

\subsection*{Baseline 1: Complex-valued Shamir Secret Sharing}
We consider the scheme outlined in Sec. \ref{sec:summary}, specifically, that:
\begin{eqnarray}\tilde{A}_i &=& A + x_iR_1 + x_i^2 R_2 \label{eq:ss1repeat} \\
\tilde{B}_i &=& B + x_iS_1 + x_i^2 S_2, \label{eq:ss2repeat}\end{eqnarray}
where $x_i = e^{j \pi(i-1)/3},$ where $j = \sqrt{-1}.$ Our goal is to find a lower bound on the DP parameter $\epsilon$ defined in Definition \ref{def:edp} irrespective of the marginal distributions of the independent, identically distributed complex-valued variables $R_1,R_2,S_1,S_2.$ Assume that $E[|R_1|^2] = \sigma_{n}^2$. Let $\epsilon_{min}$ denote the infimum among all possible complex-valued distributions $\mathbb{P}_{R_1}, \mathbb{P}_{R_2}, \mathbb{P}_{S_1}, \mathbb{P}_{S_2}$ with variance $\sigma_n^2$, among all possible parameters $\overline{\epsilon}$ such that the coding scheme achieves $2$-node $\overline{\epsilon}$-DP. Consider any coding scheme $\mathcal{C}$ consistent with \eqref{eq:ss1repeat},\eqref{eq:ss2repeat}. From Lemma\footnote{The Lemma also holds for complex-valued random variables $X,X', Z_1, Z_2,\ldots Z_n$, allowing $\vec{w}$ to be complex with $\mathbb{E}[|X|^2] =\gamma^2, \mathbb{E}[|X'|^2] \leq \gamma^2$, and with complex-valued Hermitian covariance matrices, for e.g., $i,j$-th entry of $\mathbf{K}_1$ being $\mathbb{E}[y_iy_j^*].$} \ref{lem:theSNRlemma}, we know that there exist $i,j \in \{1,2,3\}$ and co-efficients $d_i,d_j$ such that:
$$ d_i \tilde{A}_i + d_j \tilde{A}_{j} = A+ Z $$
where $Z$ is a random variable with 
$ E[|Z|^2] = \frac{1}{\texttt{SNR}},$ where $ \texttt{SNR} = \frac{\det(\mathbf{K}_1)}{\det(\mathbf{K}_2)} - 1,$  $$\mathbf{K}_{1} = \begin{bmatrix}1+|x_i|^2 \sigma_n^2 + |x_i^4|\sigma_n^2 & 1+x_i x_j^* \sigma_n^2 + (x_i x_j^*)^2 \sigma_n^2 \\ 1+x_i^* x_j \sigma_n^2 + (x_i^* x_j)^2 \sigma_n^2  & 1+|x_j|^2 \sigma_n^2 + |x_j^4|\sigma_n^2\end{bmatrix}$$ is the covariance matrix of $(\tilde{A}_{1},\tilde{A}_{2})$ and $$\mathbf{K}_2 = \sigma_n^2 \begin{bmatrix}|x_i|^2 + |x_i^4| & x_i x_j^* + (x_i x_j^*)^2 \\ x_i^* x_j  + (x_i^* x_j)^2   & |x_j|^2 + |x_j^4|\end{bmatrix}$$ is the covariance matrix of $(x_i R_1 + x_i^2 R_2, x_j R_1 + x_j^2 R_2).$ By the post-processing property of differential privacy, we have:
$${\epsilon}_{min} \geq \overline{\epsilon} \stackrel{\Delta}{=} \sup_{\mathcal{A}, A_0,A_1, |A_0-A_1| \leq 1}\frac{\mathbb{P}(A_0 + Z \in \mathcal{A})}{\mathbb{P}(A_1 + Z \in \mathcal{A})}.$$
That is ${\epsilon}_{min}$ is at least as large as the DP guarantee provided by the additive noise $Z$. Because $E[|Z|^2] =  \frac{1}{\texttt{SNR}},$ we have:
$$ \overline{\epsilon} \geq \inf\sup_{\mathcal{A}, A_0,A_1, |A_0-A_1| \leq 1}\frac{\mathbb{P}(A_0 + \overline{Z} \in \mathcal{A})}{\mathbb{P}(A_1 + \overline{Z} \in \mathcal{A})},$$
where the infimum is over all random variables $\overline{Z}$ with variance $\frac{1}{\texttt{SNR}}.$ Note that $\overline{\epsilon}$ is characterized in \cite{Pramod2015} as the solution to the relation:
$$  \frac{1}{\texttt{SNR}} = (\sigma^{*}(\epsilon))^2, $$
where $\sigma^{*}$ is provided in \eqref{eq:optimalsigma}, with $\overline{Z}$ having the staircase distribution (also described in \cite{Pramod2015}). Figure \ref{fig:dpcomp} numerically computes and plots $\overline{\epsilon},$ which is the strongest possible differential privacy parameter attainable with additive noise of variance $\frac{1}{\texttt{SNR}}.$

\subsection*{Baseline 2: $\tilde{R}_i,\tilde{S}_{i},i=1,2,3$ are i.i.d. satisfying $\frac{\epsilon}{2}$-DP}
In the second baseline, we draw $\tilde{R}_i, \tilde{S}_i,i=1,2,,3$ independently and with the same distribution $\mathbb{P}_{\tilde{R}_1}.$ We aim to characterize the optimal privacy-accuracy trade-off for this family of schemes. As noted in Sec. \ref{sec:model}, we assume without loss of generality that $\mathbb{E}\left[\tilde{R}_i\right]=\mathbb{E}\left[\tilde{S}_i\right]=0$ for $i=1,2,3.$. For a fixed distribution $\mathbb{P}_{\tilde{R}_1},$ denote the mean squared error as $\texttt{MSE}(\mathbb{P}_{\tilde{R}_1}),$ that is:
$$ \mathbb{E}[(AB-\widetilde{C})^2] = \inf_{d_1,d_2,d_3}\mathbb{E}\left[\left(AB-\sum_{i=1}^{3}d_i(A+\tilde{R}_i)(B+\tilde{S}_i)\right)^2\right].$$ Observe crucially that because of the independence of $A,B,\tilde{R}_i|_{i=1}^{3},\tilde{S}_i|_{i=1}^{3},$ and the fact that $\mathbb{E}[\tilde{R}_i]=\mathbb{E}[\tilde{S}_i]=0,$ the right hand side of the above equation is exclusively a function of the variance $\mathbb{E}[\tilde{R}_1^2].$ Denote $\sigma^2 = \mathbb{E}[\tilde{R}_1^2]$. We show next that the staircase mechanism of \cite{Pramod2015} with variance $\sigma^2$ provides a privacy guarantee that is at least as strong as that provided by $\mathbb{P}_{\tilde{R}_1}.$ This proves that the setting $\mathbb{P}_{\tilde{R}_1}$ to be the staircase mechanism achieves the optimal privacy-accuracy trade-off among the family of schemes that selects $\tilde{R}_i,\tilde{S}_{i},i=1,2,3$ in an i.i.d. manner.

Let 
$$\sup_{\mathcal{A}, A_0,A_1, |A_0-A_1| \leq 1}\frac{\mathbb{P}(A_0 + \tilde{R}_1 \in \mathcal{A})}{\mathbb{P}(A_1 + \tilde{R}_1 \in \mathcal{A})} = \tilde{\epsilon},$$
that is, $A+\tilde{R}_1$ satisfies $\tilde{\epsilon}$-DP. Because the inputs to any pair of nodes $(\tilde{A}_i,\tilde{A}_j)$ satisfies DP with parameter $2 \tilde{\epsilon}$, the secure coding scheme that sets $\tilde{R}_i, \tilde{S}_i,i=1,2,,3$ independently and with the same distribution $\mathbb{P}_{\tilde{R}_1}$ satisfies $2$-node $2 \tilde{\epsilon}$-DP.

Let  
$$\arg\inf_{\mathbb{Q}_{\tilde{R}_1}: \mathbb{E}_{\mathbb{Q}}[\tilde{R}_1] = 0, \mathbb{E}_{\mathbb{Q}}[\tilde{R}_1^2] = \sigma^2}\sup_{\mathcal{A}, A_0,A_1, |A_0-A_1| \leq 1}\frac{\mathbb{Q}(A_0 + \tilde{R}_1 \in \mathcal{A})}{\mathbb{Q}(A_1 + \tilde{R}_1 \in \mathcal{A})} = \mathbb{P}^{*}_{\tilde{R}_1},$$
that is, among all zero-mean distributions with variance $\sigma^2,$  $\mathbb{P}_{\tilde{R}_{1}}^{*}$ provides the strongest possible DP guarantee. As described previously in this section, $\mathbb{P}^{*}$ has been characterized to be the staircase mechanism in \cite{Pramod2015}. Let $\epsilon^{*}$ denote the DP guarantee provided by $\mathbb{P}^{*}$; note that $\epsilon^{*} \leq \tilde{\epsilon}.$  Then, $\tilde{R}_i, \tilde{S}_i,i=1,2,,3$ independently and with the same distribution $\mathbb{P}^*_{\tilde{R}_1}$ achieves $2\epsilon^{*}$-DP, which is at least as strong (in terms of privacy) as $\mathbb{P}_{\mathbb{R}_1}.$ Further, because $\mathbb{E}_{\mathbb{P}^*}[\tilde{R}_1^2] = \sigma^2,$ and because $\texttt{MSE}(\mathbb{P}^{*}_{\tilde{R}_1})$ is a function only of the variance $\mathbb{E}[(R_1^{*})^2]$, we infer that $\texttt{MSE}(\mathbb{P}^{*}_{\tilde{R}_1}) = \texttt{MSE}(\mathbb{P}_{\tilde{R}_1}).$  

}

\vc{
\section{Proof of Lemma \ref{lem:theSNRlemma}}
\label{app:SNRlemmaproof}
    % First, note that
% \begin{equation}
%     \mathbb{E} \left[ || \mathbf{w} \cdot \mathbf{y} - X ||^2 \right] \geq \mathbb{E} \left[ || \mathbf{w}^{*} \cdot \mathbf{y} - X ||^2 \right], 
% \end{equation}
Let $\vec{w}^{*} = \arg \min_{\vec{w}}  \mathbb{E} \left[ | \vec{w} \cdot \vec{y} - X |^2 \right] $. 
\begin{align}
    f(\vec{w}) \triangleq \mathbb{E} \left[ | \vec{w} \cdot \vec{y} - X |^2 \right]  &= \mathbb{E}[(\vec{w} \cdot \vec{y})^2] - 2 \vec{w}^T \mathbb{E}[X \vec{y}] + \mathbb{E}[X^2] \\ 
    &= \vec{w}^T \mathbb{E}[\vec{y} \vec{y}^T] \vec{w} - 2 \vec{w}^T \mathbb{E}[X \vec{y}] + \mathbb{E}[X^2]
\end{align}
\begin{align}
    \frac{\partial f}{ \partial \vec{w}} = 2\mathbf{K_1} \vec{w} - 2 \mathbb{E}[X\vec{y}] = 0. 
\end{align}
Hence, 
\begin{align}
    \vec{w}^{*} =\mathbf{K}_1^{-1} \mathbb{E}[X\vec{y}].
\end{align}
\begin{align}
    f(\vec{w}^{*}) &= \mathbb{E}[X\vec{y}]^T \mathbf{K}_1^{-1} \mathbb{E}[X\vec{y}] - 2 \mathbb{E}[X\vec{y}]^T \mathbf{K}_1^{-1} \mathbb{E}[X\vec{y}] + \mathbb{E}[X^2] \label{eq:00} \\ 
    &= \mathbb{E}[X^2] - \mathbb{E}[X\vec{y}]^T\mathbf{K}_1^{-1} \mathbb{E}[X\vec{y}]
\end{align}
Since $X$ and $Z_i$'s are independent and $\mathbb{E}[X]=0$, $\mathbb{E}[Z_i X] = \mathbb{E}[Z_i] \mathbb{E}[X] = 0$. Thus we have:
\begin{align}
    \mathbb{E}[X \vec{y}] = \mathbb{E} \left[ \begin{bmatrix}
    \nu_1 X^2 + Z_1 X \\ \vdots \\ \nu_n X^2 + Z_n X
\end{bmatrix}  \right]  = \mathbb{E}[X^2] \begin{bmatrix}
    \nu_1 \\ \vdots \\ \nu_n
\end{bmatrix} = \mathbb{E}[X^2] \vec{\nu}.
\end{align}
Plugging this in, we get:
\begin{align}
    f(\vec{w}^{*}) &= \mathbb{E}[X^2] - \mathbb{E}[X^2]\gamma^2 \vec{\nu}^T\mathbf{K}_1^{-1} \vec{\nu} \label{eq:000}\\ 
    &= \mathbb{E}[X^2] ( 1- \gamma^2 \vec{\nu}^T \mathbf{K_1}^{-1} \vec{\nu}) \\ 
    &= \text{det} (  \mathbb{E}[X^2] ( 1- \gamma^2 \vec{\nu}^T \mathbf{K}_1^{-1} \vec{\nu}) ) \\ 
    &= \mathbb{E}[X^2] \cdot \text{det}(1) \cdot \text{det}(\mathbf{K}_1^{-1}) \cdot \text{det} (\mathbf{K}_1 - \gamma^2 \vec{\nu} \vec{\nu}^T) \label{eq:det_ex}\\ 
    &= \mathbb{E}[X^2] \frac{\text{det}(\mathbf{K}_2)}{\text{det}(\mathbf{K}_1)} \label{eq:11} \\
    &= \frac{\gamma^2}{1+\texttt{SNR}},
\end{align}
where \eqref{eq:det_ex} follows from the identity: 
\begin{equation}
    \text{det}(\mathbf{D} + \mathbf{ABC}) = \text{det}(\mathbf{D}) \text{det}(\mathbf{B}) \text{det} (\mathbf{B}^{-1} + \mathbf{C} \mathbf{D}^{-1} \mathbf{A}).
\end{equation}

For the case where we have $\mathbb{E}[X'^2] \leq \gamma^2,$ and $\vec{y}' = \begin{bmatrix} \nu_1 X'+Z_1\\ \vdots \\ \nu_n X'+Z_n, \end{bmatrix}$ 
we can write,
\begin{align*}
    \mathbb{E}[|\vec{w}^{*} \cdot \vec{y}' - X'|^2] &= \mathbb{E}[(\vec{w}^{*} \cdot \vec{\nu}-1)X' + (\vec{w}^{*} \cdot \vec{Z})^2  ] \\
    &= \mathbb{E}[X'^2] (\vec{w}^{*} \cdot \vec{\nu}-1)^2 +\mathbb{E}[ (\vec{w}^{*} \cdot \vec{Z})^2 ],
\end{align*}
where $\vec{Z} = \begin{bmatrix}Z_1 \\ Z_2 \\ \vdots \\ Z_n\end{bmatrix}.$
The above expression is clearly an increasing function of $\mathbb{E}[X'^2]$, and therefore is upper bounded by the mean squared error for the case where $\mathbb{E}[X'^2] = \gamma^2.$ Therefore         
 $$ \mathbb{E}[|\vec{w}^{*} \cdot \vec{y}' - X'|^2] \leq \frac{\gamma^2}{1+\texttt{SNR}.}$$
 }

\section{Finite precision analysis of real-valued Shamir secret sharing in Remark \ref{eq:remark5}}
\label{app:BGWprecision}
We consider the real valued Shamir Secret Sharing scheme where node $i$ gets - in a system with perfect precision - $\Gamma_i = p_A(x_i),\Theta_i = p_B(x_i)$ for $i=1,2,\ldots,2t+1$, where:
$$p_{A}(x) = A+R_1x+R_2x^2+\ldots+R_{t} x^t$$
$$p_{B}(x) = B+S_1x+S_2x^2+\ldots+S_{t} x^t$$
and $x_1,x_2,\ldots,x_{2t+1}$ are distinct \vc{real} scalars. \vc{ As described in Remark \ref{eq:remark5}, the DP parameter $\epsilon$ can be driven as close to $0$ as we wish by making $x_1,x_2,\ldots,x_{2t+1}$ arbitrarily large.} 
In a system with perfect precision, node $i$ outputs $\Gamma_i\Theta_i$, and the decoder obtains $\widehat{C}=AB$ with perfect accuracy as a linear combination:
\begin{equation}
    AB = \sum_{i=1}^{2t+1} d_i \Gamma_i \Theta_i
    \label{eq:BGWdecoding}
\end{equation}

Similar to Section \ref{sec:precision}, we assume a system with finite precision \vc{ where node $i$ receives quantized values $\hat{\Gamma}_i,\hat{\Theta}_i \in \mathbb{R}$ and can compute the product $\hat{\Gamma}_i \hat{\Theta}_i$ perfectly}. As in Section \ref{sec:precision}, we assume
$$ \hat{\Gamma}_i = p_A(x_i) + Y_i$$
$$ \hat{\Theta}_i = p_B(x_i) + Z_i$$
where $Y_i$ is a random variable that is independent of $\Gamma_i|_{i=1}^{2t+1}$,$\Theta_i|_{i=1}^{2t+1}$, $Z_{i}|_{i=1}^{2t+1}$, $\{Y_j: j \in \{1,2,\ldots,2t+1\}-\{i\}\}.$ Similarly $Z_i$ is a random variable that is independent of $\Gamma_i|_{i=1}^{2t+1}$,$\Theta_i|_{i=1}^{2t+1}$, $Y_{i}|_{i=1}^{2t+1}$, $\{Z_j: j \in \{1,2,\ldots,2t+1\}-\{i\}\}.$ Note that this independence property is achieved via dithered lattice quantizers as in Section \ref{sec:precision}. Node $i$ is assumed to output $\hat{\Gamma}_i\hat{\Theta}_i$ perfectly. Similarly to the reasoning in Section \ref{sec:precision}, we get $\mathbb{E}[Y_i^2],\mathbb{E}[Z_i^2] = \Omega(2^{-2M}),$ where $M$ is the number of bits of precision at each node.   
We assume that the decoder obtains:
$$ \hat{C} = \sum_{i=1}^{2t+1} d_i \hat{\Gamma}_i \hat{\Theta}_i$$
where co-efficients $d_i$ are as in (\ref{eq:BGWdecoding}). We show next that choosing $M(\delta) = K \log \left(\frac{1}{\delta}\right)$ for any $K>0.5$ suffices to ensure that $\mathbb{E}[(AB-\hat{C})^2] \leq \delta$ for sufficiently small delta.
\begin{align*}
    &\mathbb{E}[(AB-\hat{C})^2]
    \\&= \mathbb{E}[(AB-\sum_{i=1}^{2t+1} d_i \hat{\Gamma}_i \hat{\Theta}_i)^2]
  \\& =   \mathbb{E}[(AB-\sum_{i=1}^{2t+1} d_i (\Gamma_i+Y_i) (\Theta_i+Z_i))^2]
  \\& = \mathbb{E}[(AB-\sum_{i=1}^{2t+1} d_i \Gamma_i\Theta_i)^2]+\Theta(2^{-2M(\delta)})
  \\&= 0 + \Theta(2^{-2M(\delta)})
\end{align*}
Clearly, if $M(\delta) = K \log \left(\frac{1}{\delta}\right)$ for any $K > 0.5,$ we have, for sufficiently small $\delta,$ $\mathbb{E}[(AB-\hat{C})^2] \leq \delta$ as required. 

%\subsection*{Baseline 1: Complex-valued Shamir Secret Sharing}

\newpage

\end{document}